%% file: book26.tex
\documentclass[12pt,oneside,openany,article]{memoir}
\usepackage{mempatch}

\nouppercaseheads
\usepackage[amsthm,thmmarks,hyperref]{ntheorem}
\usepackage[noTeX]{mmap}
\usepackage{etex}

\usepackage[english]{babel}
\usepackage[leqno]{mathtools}

\usepackage{mathrsfs}
\usepackage{upgreek}
\usepackage[compress,square,comma,numbers]{natbib}
\usepackage{hypernat} 

\usepackage{sfmath}

\usepackage{microtype}

\usepackage{subfig}
\usepackage{wrapfig}
\usepackage{array}

\usepackage{graphicx,color}
\definecolor{gray}{gray}{0}
\pagecolor{white}

\usepackage{hyperxmp}
\usepackage{xr-hyper}
\usepackage{nameref}
\usepackage[pdftex,bookmarks,pdfnewwindow,plainpages=false,unicode]{hyperref}

\usepackage{bookmark}

\usepackage{enumitem}

\usepackage{amssymb} 

\hypersetup{
colorlinks=true,
linkcolor=black,
citecolor=black,
urlcolor=blue,
pdfauthor={Victor Ivrii},
pdftitle={Operators with Periodic Hamiltonian Flows in Domains with the Boundary},
pdfsubject={Sharp Spectral Asymptotics},
pdfkeywords={Microlocal Analysis,  Sharp Spectral Asymptotics, Periodic flows},
bookmarksdepth={4}
}

\numberwithin{equation}{chapter}

\theoremstyle{plain}
\newtheorem{theorem}{Theorem}[section]
\newtheorem{lemma}[theorem]{Lemma}
\newtheorem{proposition}[theorem]{Proposition}
\newtheorem{corollary}[theorem]{Corollary}

\theoremstyle{definition}

\newtheorem{Problem}[theorem]{Problem}
\theoremstyle{remark}
\newtheorem{remark}[theorem]{Remark}

\numberwithin{equation}{section}

\DeclareMathAlphabet{\mathpzc}{OT1}{pzc}{m}{it}

 \newcommand{\cE}{\mathcal{E}}

 \newcommand{\cK}{\mathcal{K}}

 \newcommand{\cX}{\mathcal{X}}
 \newcommand{\cY}{\mathcal{Y}}

 \newcommand{\sC}{\mathscr{C}}

 \newcommand{\sH}{\mathscr{H}}

 \newcommand{\sL}{\mathscr{L}}

\newcommand{\E}{{\mathsf{E}}}
\newcommand{\TF}{{\mathsf{TF}}}

\newcommand{\sfH}{{\mathsf{H}}}
\newcommand{\sfU}{{\mathsf{U}}}
\newcommand{\Weyl}{{\mathsf{Weyl}}}

\newcommand{\Dirac}{{\mathsf{Dirac}}}
\newcommand{\Schwinger}{{\mathsf{Schwinger}}}

\renewcommand{\d}{{\mathsf{d}}}
\newcommand{\D}{{\mathsf{D}}}

\newcommand{\N}{{\mathsf{N}}}
\newcommand{\I}{{\mathsf{I}}}

\newcommand{\T}{{\mathsf{T}}}

\newcommand{\x}{{\mathsf{x}}}
\newcommand{\y}{{\mathsf{y}}}

\newcommand{\const}{{\mathsf{const}}}
\newcommand{\dist}{{{\mathsf{dist}}}}

\newcommand{\loc}{{{\mathsf{loc}}}}
\newcommand{\new}{{{\mathsf{new}}}}

\newcommand{\bC}{{\mathbb{C}}}

\newcommand{\bR}{{\mathbb{R}}}

\newcommand{\bZ}{{\mathbb{Z}}}

\newcommand{\fH}{{\mathfrak{H}}}

\newcommand{\boldupsigma}{{\boldsymbol{\upsigma}}}
\newcommand{\orig}{{\mathsf{orig}}}

\def\1{\boldsymbol {|}}
\newcommand{\blangle}{{\boldsymbol{\langle}}}
\newcommand{\brangle}{{\boldsymbol{\rangle}}}

\newcommand{\3}{{|\!|\!|}}
\newcommand{\Def}{\mathrel{\mathop:}=}
\newcommand{\Fed}{=\mathrel{\mathop:}}

\newcommand{\diam}{\operatorname{diam}}

\renewcommand{\Im}{\operatorname{Im}}       

\renewcommand{\Re}{\operatorname{Re}}       

\newcommand{\Spec}{\operatorname{Spec}}

\newcommand{\supp}{\operatorname{supp}}

\newcommand{\tr}{\operatorname{tr}}
\newcommand{\Tr}{\operatorname{Tr}}

\newcommand{\Res}{\operatorname{Res}}

\newenvironment{claim}[1][{\textup{(\theequation)}}]{\refstepcounter{equation}\vglue10pt
\begin{trivlist}
\item[{\hskip\labelsep#1}]}{\vglue10pt\end{trivlist}}

\newenvironment{claim*}[1][{}]{\vglue10pt
\begin{trivlist}
\item[{\hskip\labelsep#1}]}{\vglue10pt\end{trivlist}}

\newenvironment{phantomequation}[1][]{\refstepcounter{equation}}{}
\newcounter{note}

\DeclareTextCommand{\textge}{PU}{\9042\145}
\DeclareTextCommand{\textle}{PU}{\9042\144}
\DeclareTextCommand{\texthat}{PD1}{\136}

\begin{document}
\title{Asymptotics of the ground state energy of heavy molecules in self-generated magnetic field}
\author{Victor Ivrii}

\maketitle
{\abstract%
We consider asymptotics of the ground state energy of heavy atoms and molecules in the self-generated magnetic field and derive it including Scott, and in some cases even Schwinger and Dirac corrections (if magnetic field is not too strong). We also consider related topics:  an excessive negative charge, ionization energy and excessive positive charge when atoms can still bind into molecules.

This preprint supersedes  \cite{ivrii:Preprint-A}, \cite{ivrii:Preprint-B}, and \cite{ivrii:Preprint-C} as it contains more detailed analysis, more results, stronger results and also some corrections.
\endabstract}

\setcounter{chapter}{-1}
\setcounter{secnumdepth}{2}

\tableofcontents

\chapter{Introduction}
\label{sect-26-1}

We are going to replace Schr\"odinger operator without magnetic field as in Chapter~\ref{book_new-sect-24} or with a constant magnetic field as in Chapter~\ref{book_new-sect-25} by Schr\"odinger operator
\begin{equation}
H=H _{A,V}=\bigl((D -A)\cdot \boldupsigma \bigr) ^2-V(x)
\label{26-1-1}
\end{equation}
with unknown magnetic field $A$ but then to add to the ground state energy  of the atom (or molecule) the energy of magnetic field (see selected term in (\ref{26-1-2})  thus arriving to
\begin{equation}
\E(A)= \inf \Spec (\sfH_{A,V} ) + \underbracket{\alpha^{-1} \int
 |\nabla \times A|^2\,dx}
\label{26-1-2}
\end{equation}
with $N$-particle quantum Hamiltonian $\sfH_{A,V}$ defined by (\ref{book_new-25-1-1}) and a parameter $\alpha\in (0,\kappa^*Z^{-1}]$ with small constant $\kappa^*>0$.

Then  finally
\begin{equation}
\E ^*=\inf_{A\in \sH^1_0} \E(A)
\label{26-1-3}
\end{equation}
defines a ground state energy with a \index{magnetic field!self generated}\emph{self-generated magnetic field}\footnote{\label{foot-26-1} This notion was introduced in series of papers  L. Erd\"os,  S. Fournais  and   J. P. Solovej~\cite{EFS1, EFS2, EFS3}; see also L. Erd\"os and J. P. Solovej~\cite{erdos:solovej}.}.

First of all we are lacking so far a semiclassical local theory and we are developing it in Section~\ref{sect-26-2} where we consider one-particle quantum Hamiltonian
\begin{equation}
H=H _{A,V}=\bigl((hD -A)\cdot \boldupsigma \bigr) ^2-V(x)
\label{26-1-4}
\end{equation}
but instead of $\inf \Spec (\sfH _{A,V})$ we consider $\Tr^-( H_{A,V})$ which as we already know is what replaces $\inf \Spec (\sfH _{A,V})$ if electrons do not interact (then if electrons interact we will need to replace $V$ by $W$ which includes a potential generated by the electron cloud and justify this by estimating an error). We define energy of magnetic field as in (\ref{26-1-2}) but with $\kappa$ replaced by $\kappa h^{d-1}$ (here $d\ge 2$ is arbitrary) and we prove that for $d=2,3$ in this framework self-generated magnetic field is weak and the asymptotics with the remainder $O(h^{2-d})$ (or even $o(h^{1-d})$ under standard assumption of the global nature) is exactly as for $\kappa=0$ (i.e. with $A=0$). Under standard assumption about trajectories we can upgrade this asymptotics to even sharper with the remainder estimate $o(h^{-1})$ and with term $\varkappa_2 h^{-1}$ (Schwinger correction term).

Then in Section~\ref{sect-26-3} we consider operator with potential having Coulomb-type singularities and combining results and arguments of Sections~\ref{sect-26-2} and~\ref{book_new-sect-12-5} prove for $d=3$  that
\begin{equation}
\Tr^-( H_{A,V})=\Weyl_1 + 2S(\kappa) h^{-2} +
O\bigl(\kappa |\log \kappa|^{\frac{1}{3}}h^{-\frac{4}{3}}+h^{-1}\bigr)
\label{26-1-5}
\end{equation}
provided  $\kappa<\kappa^*$ (which is a small constant) and there is just one singularity; when there are several singularities with a minimal distance $a\gg 1$ between them we prove that
\begin{equation}
\Tr^-( H_{A,V})=\Weyl_1 + 2S(\kappa) h^{-2} +
O\bigl(\kappa |\log \kappa|^{\frac{1}{3}}h^{-\frac{4}{3}}+h^{-1}+\kappa a^{-3}h^{-2}\bigr).
\label{26-1-6}
\end{equation}
As $\kappa\ll h^{\frac{1}{3}}|\log h|^{-\frac{1}{3}}$ under standard assumption about trajectories we can upgrade this asymptotics to even sharper with the remainder estimate $o(h^{-1})$ and with Schwinger correction term.

Further, in Section~\ref{sect-26-4} we apply these results to provide estimates from above and below for the total energy (\ref{26-1-3}). As a byproduct we also estimate  $\D(\rho_\Psi-\rho^\TF,\,\rho_\Psi-\rho^\TF)$ where $\Psi$ is a ground state for a near-minimizer $A$.

This estimate enables us  in Section~\ref{sect-26-5} to derive upper estimates for excessive negative charge, estimates or asymptotics for ionization energy, and in the free nuclei model also for minimal distance between nuclei and (in the case of molecule) for excessive positive charge.

\chapter{Local semiclassical trace asymptotics}
\label{sect-26-2}

\section{Toy-model}
\label{sect-26-2-1}

\subsection{Statement of the problem}%
\label{sect-26-2-1-1}
Let us consider  operator (\ref{26-1-1}) in $\bR^d$ with $d=3$ where $A,V$ are real-valued functions and $V\in \sL^{\frac{5}{2}}\cap \sL^{4}$,
$A\in \sH^1_0$. Then this operator is self-adjoint. We are interested in
$\Tr^- (H_{A,V})$ (the sum of all negative eigenvalues of this operator). Let
\begin{gather}
\E ^*\Def \inf_{A\in \sH^1_0(B(0,r))}\E(A),\label{26-2-1}\\
\E(A)\Def \Bigl( \Tr^-H_{A,V}  +
\kappa^{-1} h^{-2}\int  |\partial A|^2\,dx\Bigr)
\label{26-2-2}
\end{gather}
with $\partial A=(\partial_i A_j)$ a matrix; here and below $r$ is a parameter and constants do not depend on it.

The estimate from above is delivered by $A=0$ and Weyl formula with an error $O(h^{-1})$ as $V\in \sC^{2,1}$\,\footnote{\label{foot-26-2} Recall that this means that the second derivatives of $V$ are continuous with the continuity modulus $O(|\log |x-y||^{-1})$, see Section~\ref{book_new-sect-4-5}. If there is a boundary it does not pose any problem provided it is in the classically forbidden region.}
\begin{gather}
\E^*\le \Weyl_1 + O(h^{-1});
\label{26-2-3}\\
\shortintertext{where}
\Weyl(\tau) = \frac{1}{3\pi^2} h^{-3}\int (V+\tau)_+^{\frac{3}{2}}\,dx,\label{26-2-4}\\
\shortintertext{and}
\Weyl_1 = \int _{-\infty}^0 \tau\d_\tau \Weyl(\tau)=
-\frac{2}{15\pi^2} \int V_+^{\frac{5}{2}}\,dx.\label{26-2-5}
\end{gather}
Also for estimates $o(h^{-2})$ we need to include into $\Weyl_1$ the corresponding boundary term. Now our goal is to provide an estimate from below
\begin{equation}
\E^*\ge \Weyl_1 - O(h^{-1});
\label{26-2-6}
\end{equation}
We will use also $\Weyl (x,\tau)$ and $\Weyl _1(x)$ defined the same way albeit without integration with respect to $x$.

\subsection{Preliminary analysis}
\label{sect-26-2-1-2}
So, let us  estimate $\E(A)$ from below. First we need the following really simple

\begin{proposition}\label{prop-26-2-1}
Let  $V\in \sL^{\frac{5}{2}}\cap \sL^4$. Then
\begin{gather}
\E^*\ge -C h^{-3}\label{26-2-7}\\
\shortintertext{and either}
\frac {1}{\kappa h^2} \int |\partial A|^2\,dx \le Ch^{-3}\label{26-2-8}
\end{gather}
or $\E (A) \ge ch^{-3}$.
\end{proposition}

\begin{proof}
Using the Magnetic Lieb-Thirring inequality (5)  of E.~H.~Lieb, M.~Loss,~M. and J.~P.~Solovej~\cite{lieb:loss:solovej})
\begin{multline}
\int \tr e_1(x,x,\tau)\, dx \ge \\
- Ch^{-3} \int V_+^{\frac{5}{2}}\,dx
-Ch^2 \int \Bigl(h^{-2}\int |\partial A|^2\,dx\Bigr)^{\frac{3}{4}}
\Bigl(h^{-8} \int V_+^4\,dx \Bigr)^{\frac{1}{4}}
\label{26-2-9}
\end{multline}
we conclude that for any $\delta>0$
\begin{equation}
\E(A) \ge -Ch^{-3}-C\delta^3 h^{-3} +
 \bigl(\kappa^{-1}-\delta^{-1}) h^{-1}\int |\partial A|^2\,dx
\label{26-2-10}
\end{equation}
which implies both statements of the Proposition.
\end{proof}

\begin{proposition}\label{prop-26-2-2}
Let  $V_+\in  \sL^{\frac{5}{2}}\cap \sL^4$, $\kappa \le ch^{-1} $ and
\begin{equation}
V\le - C^{-1} (1+|x|)^\delta +C.
\label{26-2-11}
\end{equation}
Then there exists a minimizer $A$.
\end{proposition}

\begin{proof}
Consider a minimizing sequence $A_j$. Without any loss of the generality one can assume that $A_j\to A_\infty$ weakly in $\sH^1$ and  in $\sL^6$ and strongly in $\sL^p_\loc$ with any $p<6$\,\footnote{\label{foot-26-3} Otherwise we select a converging subsequence.}. Then $A_\infty$ is a minimizer.

Really, due to (\ref{26-2-10}) negative spectra of $H_{A_j,V}$ are discrete and the number of  negative eigenvalues is bounded by $N_{h}$. Consider   ordered eigenvalues $\lambda_{j,k}$ of $H_{A_j,V}$. Without any loss of the generality one can assume that $\lambda_{j,k}$ have limits $\lambda_{\infty,k}\le 0$ (we go to the subsequence if needed).

We claim that $\lambda_{\infty,k}$ are also eigenvalues and if $\lambda_{\infty,k}=\ldots=\lambda_{\infty, k+r-1}$ then it is eigenvalue of at least multiplicity $r$. Indeed, let $u_{j,k}$ be corresponding eigenfunctions, orthonormal in $\sL^2$. Then in virtue of $A_j$ being bounded in $\sL^6$ and $V\in \sL^4$ we can estimate
\begin{equation*}
\|Du_{j,k}\|\le K \|u_{j,k}\|_{6}^{1-\sigma}\cdot \|u_{j,k}\|^\sigma\le K\|Du_{j,k}\|^{1-\sigma}\cdot \|u_{j,k}\|^\sigma
\end{equation*}
with $\sigma>0$ which implies $\|Du_{j,k}\|\le K$. Also assumption (\ref{26-2-11}) implies that $\| (1+|x|)^{\delta/2 }u_{j,k}\|$ are bounded and therefore without any loss of the generality one can assume that $u_{j,k}$ converge strongly.

Then
\begin{gather}
\lim_{j\to\infty} \Tr^- (H_{A_j,V}) \ge \Tr^- (H_{A_\infty,V}),\label{26-2-12}\\
\liminf_{j\to\infty}\int |\partial A_j|^2\, dx \ge
\int |\partial A_\infty |^2\, dx
\label{26-2-13}
\end{gather}
and therefore $\E(A_\infty)\le \E^*$. Then $A_\infty$ is a minimizer and there are equalities in (\ref{26-2-12})--(\ref{26-2-13}) and, in particular, there no negative eigenvalues of $H_{A_\infty, V}$  other than $\lambda_{\infty,k}$. \end{proof}

\begin{remark}\label{rem-26-2-3}
We do not know if the minimizer is unique. Also we do not impose here any restrictions on $r, K$ (which may depend on $h$) in (\ref{26-2-11}) or $\kappa>0$. From now on until further notice let $A$ be a minimizer.
\end{remark}

\begin{proposition}\label{prop-26-2-4}
In the framework of Proposition~\ref{prop-26-2-2} let $A$ be a minimizer. Then
\begin{multline}
\frac{2}{\kappa h^2} \Delta A_j (x)   = \Phi_j\Def\\
-\Re\tr  \upsigma_j\Bigl( (hD -A)_x \cdot \boldupsigma  e (x,y,\tau)+
e (x,y,\tau)\,^t (hD-A)_y \cdot \boldupsigma \Bigr)  \Bigr|_{y=x}
\label{26-2-14}
\end{multline}
where $A=(A_1,A_2,A_3)$, $\boldupsigma=(\upsigma_1, \upsigma_2, \upsigma_3)$ and $e(x,y,\tau)$ is the Schwartz kernel of the spectral projector $\uptheta (- H)$ of $H=H_{A,V}$ and $\tr$ is a matrix trace.
\end{proposition}

\begin{proof}
Consider variation $\updelta A$ of $A$ and variation of $\Tr^- (H)=\Tr (H^-)$ where $H^-=H \uptheta (- H)$ is a negative part of $H$. Note that the spectral projector of $H$ is
\begin{equation}
\uptheta (\tau- H) =
\frac{1}{2\pi i} \int_{-\infty}^\tau \Res_\bR (\tau -H)^{-1}
\label{26-2-15}
\end{equation}
and therefore
\begin{multline*}
\updelta \Tr \bigl(\uptheta (\tau- H)\bigr) =
\frac{1}{2\pi i} \int_{-\infty}^\tau \Res_\bR
\Tr \bigl((\tau -H)^{-1} (\updelta H) (\tau -H)^{-1}\bigr)= \\
\frac{1}{2\pi i} \int_{-\infty}^\tau \Res_\bR
\Tr  \bigl((\updelta H) (\tau -H)^{-2} \bigr)=
-\partial_\tau \frac{1}{2\pi i} \int_{-\infty}^\tau \Res_\bR
\Tr  \bigl((\updelta H) (\tau -H)^{-1}\bigr) = \\[3pt]
-\partial_\tau \Tr \bigl((\updelta H) \uptheta (\tau-H)\bigr).
\end{multline*}
Plugging it into
\begin{equation}
\Tr^- (H)= \int _{-\infty}^0 \tau d_\tau
\Tr \bigl(\uptheta (\tau- H)\bigr)=
-\int _{-\infty}^0 \Tr \bigl(\uptheta (\tau- H)\,d\tau \bigr)
\label{26-2-16}
\end{equation}
and integrating  with respect to $\tau$ we arrive after simple calculations to
\begin{equation}
\updelta \Tr^- (H)=\Tr \bigl((\updelta H) \uptheta (\tau- H)\bigr)=
\sum_j\int \Phi_j (x) \updelta A_j(x)\,dx
\label{26-2-17}
\end{equation}
where $\Phi(x)$ is the right-hand expression of (\ref{26-2-14}). Therefore
\begin{equation}
\updelta \E(A)= \sum_j\int \bigl(\Phi_j (x)-\frac{2}{\kappa h^2}
\Delta A_j(x)\bigr) \updelta A_j(x)\,dx
\label{26-2-18}
\end{equation}
which implies (\ref{26-2-14}).
\end{proof}

\begin{proposition}\label{prop-26-2-5}
If for  $\kappa=\kappa^*$
\begin{gather}
\E^* \ge \Weyl_1 - CM\label{26-2-19}\\
\intertext{with $M\ge C h^{-1}$ then for $\kappa \le \kappa^*(1-\epsilon_0)$}
\frac{1}{\kappa h^2} \int |\partial A|^2\,dx \le C_1M.\label{26-2-20}
\end{gather}
\end{proposition}
\begin{proof}
Proof is obvious based also on the upper estimate $\E^*\le \Weyl_1+Ch^{-1}$.
\end{proof}

\subsection{Estimates}
\label{sect-26-2-1-3}

\begin{proposition}\label{prop-26-2-6}
Let estimate \textup{(\ref{26-2-20})} be fulfilled and let
\begin{equation}
\varsigma = \kappa M h \le c.
\label{26-2-21}
\end{equation}
Then as $\tau\le c$

\begin{enumerate}[label=(\roman*), fullwidth]
\item\label{prop-26-2-6-i}
Operator norm in $\sL^2$ of $(hD)^k \uptheta(\tau -H)$ does not exceed $C$ for $k=0,1,2$;

\item\label{prop-26-2-6-ii}
Operator norm in $\sL^2$ of
$(hD)^k\bigl((hD-A)\cdot\boldupsigma\bigr) \uptheta(\tau -H)$ does not exceed $C$ for $k=0,1$.
\end{enumerate}
\end{proposition}

\begin{proof}
(i) Let $u =\uptheta(\tau-H) f$. Then
$\|u\|\le \|f\|$ and since
\begin{equation}
\|A\|_{\sL^6} \le C\| \partial A\| \le C(\kappa M)^{\frac{1}{2}}h
\label{26-2-22}
\end{equation}
we conclude that
\begin{multline*}
\|hD u\| \le \|(hD-A)u\| +\|A u\| \le
\|(hD-A)u\| +C\|A \|_{\sL^6}\cdot \|u\|_{\sL^3}\le\\[2pt]
\|(hD-A)u\| +C(\kappa M)^{\frac{1}{2}}h \|u\|^{1/2}\cdot \|u\|_{\sL^6}^{1/2}\le \\[2pt]
\|(hD-A)u\| +C(\kappa M h)^{\frac{1}{2}} \|u\|^{1/2}\cdot \|h Du\|^{1/2}\le\quad \\[2pt]
\|(hD-A)u\| + \frac{1}{2} \|hDu\| + C\kappa M h  \|u\|;
\end{multline*}
therefore due to (\ref{26-2-21})
\begin{equation}
\|hD u\| \le 2\|(hD-A)u\| +C\kappa M h  \|u\|.
\label{26-2-23}
\end{equation}
On the other  hand, for $B=\nabla \times A$ and $\tau\le c$
\begin{multline*}
\|(h D-A)u\| ^2 \le C\|u\|^2 + (h |B|u,u)\le C\|u\|^2 + h\|B\|\cdot\|u\|_{\sL^4}^2 \le \\[3pt]
C\|u\|^2 +  C (\kappa M)^{\frac{1}{2}} h^2  \|u\|  \cdot \|u\|_{\sL^6} \le
C\|u\|^2 + C (\kappa M)^{\frac{1}{2}}h  \|u\|\cdot \| h D u\|\le \\[3pt]
C(1+\kappa M h^2 + \kappa ^{\frac{3}{2}}M^{\frac{3}{2}}h^2 ) \|u\|^2 + \frac{1}{2} \|(h D-A)u\| ^2
\end{multline*}
and due to (\ref{26-2-23}) we conclude that
\begin{equation}
\|(h D-A)u\|  \le C\|u\|\quad \text{and}\quad \|hDu\|\le C(1+\kappa M h) \|u\|
\label{26-2-24}
\end{equation}
provided $\kappa Mh^{1+\delta}\le c$ for sufficiently small $\delta>0$. Therefore under assumption (\ref{26-2-21})  for $k=0,1$  statement \ref{prop-26-2-6-i} is proven.

Further, since $(hD)^2=(hD-A)^2 + A(hD-A) +A hD -h[D,A]$ we in the same way as before (and using (\ref{26-2-24})) conclude that
\begin{gather*}
\|(hD)^2u\|
\le C\|u\|^2 +  \frac{1}{4} \|hD (hD-A)u\|+ \frac{1}{4} \| h^2D^2u\| \\
\shortintertext{and therefore}
\|h^2D^2u\|\le C\|u\|^2 + C \|A hDu \|
\end{gather*}
and repeating the same arguments we get  $\|h^2D^2 u\| \le C\|u\|$; so for $k=2$ statement \ref{prop-26-2-6-i}  is also proven.

\medskip\noindent
(ii) Statement \ref{prop-26-2-6-ii} is proven in the same way.
\end{proof}

\begin{corollary}\label{cor-26-2-7}
Let \textup{(\ref{26-2-20})} and \textup{(\ref{26-2-21})} be fulfilled. Then  as
$\tau\le c$
\begin{gather}
e(x,x,\tau) \le Ch^{-3}\label{26-2-25}\\
\shortintertext{and}
|\bigl((hD-A)\cdot \boldupsigma )  e(x,y,\tau)|_{x=y} |\le Ch^{-3}.
\label{26-2-26}
\end{gather}
\end{corollary}
\begin{proof}
Due to proposition \ref{prop-26-2-6} operator norms from $\sL^2$ to $\sC$  of both $\uptheta (\tau -H)$ and
$\bigl((hD-A)\cdot \boldupsigma \bigr)\uptheta (\tau -H)$ do not exceed $C$ and the same is true for an adjoint operator which imply both claims.
\end{proof}

\begin{corollary}\label{cor-26-2-8}
Let \textup{(\ref{26-2-20})} and \textup{(\ref{26-2-21})} be fulfilled and $A$ be a minimizer. Then
\begin{gather}
\|\partial A\|_{\sC^{1-\delta}} \le C\kappa h^{-1}\label{26-2-27}\\
\shortintertext{and}
\|\partial  A\|_{\sL^\infty}\le C'_\delta h^{-\frac{4}{5}-\delta}\label{26-2-28}
\end{gather}
where $\sC^\theta$ is the scale of H\"older spaces and $\delta >0$ is arbitrarily small.
\end{corollary}

\begin{proof}
Really, due to (\ref{26-2-14})   minimizer $A$ satisfies
$\|\Delta A\|_{\sL^\infty}\le C\kappa h^{-1}$. Also we know that
$\|\partial A\| \le C(\kappa Mh^2)^{\frac{1}{2}}\le Ch^{\frac{1}{2}}$ due to (\ref{26-2-21}). Then (\ref{26-2-27}) holds due to the standard properties of the elliptic equations\footnote{\label{foot-26-4} Actually we can slightly improve this statement.}.

Therefore if at some point $y$ we have $|\partial A(y)|\gtrsim \mu$, it is true in its $\epsilon  (\mu h \kappa ^{-1} )^{1-\delta}$-vicinity (provided
$\mu \le \kappa h^{-1}$) and then
\begin{gather*}
\|\partial A\|^2 \gtrsim
\mu ^2 (\mu h \kappa^{-1})^{3(1-\delta)}\\
\intertext{and  we conclude that}
\mu ^2 (\mu h \kappa^{-1})^{3(1-\delta)} \le C\kappa h^2 M \iff
\mu^{5-3\delta} \le C\kappa^{4-3\delta}h^{-1+3\delta}M
\end{gather*}
and one can see easily that (\ref{26-2-28}) holds due to (\ref{26-2-21}) and assumption $h^{-1}\le M\lesssim h^{-3}$.

On the other hand, if $\mu \ge \kappa h^{-1}$ then we need to take $\epsilon$-vicinity and then $\mu^2 \le C\kappa Mh^2 \le C h^{\frac{1}{2}}$ where we used (\ref{26-2-21}) again. Therefore (\ref{26-2-28}) has been proven. \end{proof}

\begin{remark}\label{rem-26-2-8}
\begin{enumerate}[label=(\roman*), fullwidth]
\item\label{rem-26-2-8-i}
It is not clear if it is possible to generalize this theory to arbitrary $d\ge 2$ with magnetic field energy  given by
\begin{equation}
\frac{1}{\kappa h^{d-1}}\int \bigl(|\partial A|^2-|\nabla\cdot A|^2\bigr) \,dx
\label{26-2-29}
\end{equation}
Surely one should use generalized Pauli matrices $\upsigma_j$ in the definition of the operator: as $d=2$ one can prove that $\E(A)$ is bounded from below and minimizer exists; as $d=4$ one can prove that $\E(A)$ is bounded from below as $\kappa \le \epsilon_0 h$; especially problematic is the case $d\ge 5$;

\item\label{rem-26-2-8-ii}
Therefore while arguments  of Subsection~\ref{sect-26-2-2} below remain valid for $d\ge 4$, so far they remain conditional (if a minimizer exists and satisfies some crude estimates).
\end{enumerate}
\end{remark}

\section{Microlocal analysis unleashed}
\label{sect-26-2-2}

\subsection{Sharp estimates}
\label{sect-26-2-2-1}

Now we can unleash the full power of microlocal analysis but we need to extend it to our framework. It follows by induction from (\ref{26-2-27})--(\ref{26-2-28}) and the arguments we used to derive these estimates that
\begin{equation}
\|\partial A\|_{\sC^{n-\delta}} \le C_n\kappa h^{-1-n}\qquad \forall n\in \bZ^+,
\label{26-2-30}
\end{equation}
so $A$ is ``smooth'' in $\varepsilon = h$ scale while for rough microlocal analysis as in Section~\ref{book_new-sect-2-3} one needs at least
$\varepsilon = Ch|\log h|$. We consider in this section arbitrary $d\ge 2$; see however remark~\ref{rem-26-2-8}.

\begin{proposition}\label{prop-26-2-10}
For a commutator of a pseudo-differential operator with a smooth symbol and $\sC^{\theta+1}$-function $A(x)$ a usual commutator formula holds modulo
$O( h^{\theta+1}\3\partial A\3_{\theta})$ for any non-integer $\theta >0$
where
\begin{equation}
\3f\3_{\theta}\Def \left\{\begin{aligned}
&\sum _{\alpha:|\alpha|=\theta} \sup _x |\partial^\alpha f(x)|
&&\theta\in \bZ^+,\\
&\sum _{\alpha:|\alpha|=\lfloor\theta\rfloor}
\sup _{x\ne y}|x-y|^{\lfloor\theta\rfloor-\theta}\cdot
|\partial^\alpha f(x)-\partial^\alpha f(y)|
&&\theta\notin \bZ^+.
\end{aligned}\right.
\label{26-2-31}
\end{equation}
\end{proposition}

\begin{proof}
Easy proof is left to the reader.
\end{proof}

\begin{proposition}\label{prop-26-2-11}
Assume that
\begin{gather}
\|\partial V\|_{\sC(B(0,2))}\le C_0 \label{26-2-32}\\
\shortintertext{and}
\mu \Def \|\partial  A\|_{\sC(B(0,2))} \le C_0.
\label{26-2-33}
\end{gather}
Let $U(x,y,t)$ be the Schwartz kernel of $e^{ih^{-1}tH_{A,V}}$. Then for $T\asymp 1$
\begin{enumerate}[label=(\roman*), fullwidth]
\item\label{prop-26-2-11-i}
Estimate
\begin{equation}
\|F_{t\to h^{-1}\tau} \chi_T(t) (hD_x)^\alpha(hD_y)^\beta
\psi_1 (x)\psi_2 (y) U\| \le C h^s
\label{26-2-34}
\end{equation}
holds  for all $\alpha:|\alpha|\le 2$, $\beta:|\beta|\le 2$, $s$ and all $\psi_1,\psi_2  \in \sC_0^\infty(B(0,1))$, such that
$\dist(\supp \psi_1, \supp \psi_2)\ge C_0T$ and $\tau\le c_0$; here $\|.\|$ means an operator norm from $\sL^2$ to $\sL^2$;
\item\label{prop-26-2-11-ii}
Estimate
\begin{multline}
\|F_{t\to h^{-1}\tau} \chi_T(t) (hD_x)^\alpha(hD_y)^\beta
\varphi_1 (hD_x)\varphi_2 (hD_y) U\| \le\\[3pt]
C h^s+
Ch^{\theta}\bigl(\3 A\3_{\theta+1}+\3 V\3_{\theta+1}\bigr)
 \label{26-2-35}
\end{multline}
holds for all $\alpha:|\alpha|\le 2$, $\beta:|\beta|\le 2$, $s$  and  all $\varphi_1,\varphi_2  \in \sC_0^\infty$, such that
$\dist(\supp \varphi_1, \supp \varphi_2)\ge C_0T$, and $\tau\le c_0$;
\item\label{prop-26-2-11-iii}
If also in  $B(0,2)$
\begin{equation}
\epsilon_0 \le |V|\le c
\label{26-2-36}
\end{equation}
then for a small constant $T=\epsilon$ estimate
\begin{equation}
\|F_{t\to h^{-1}\tau} \chi_T(t) (hD_x)^\alpha(hD_y)^\beta U\| \le
C h^s+
Ch^{\theta}\bigl(\3 A\3_{\theta+1}+\3 V\3_{\theta+1}\bigr)
\label{26-2-37}
\end{equation}
holds for all $\alpha:|\alpha|\le 2$, $\beta:|\beta|\le 2$, $s$  and all $\psi_1,\psi_2  \in \sC_0^\infty(B(0,1))$, such that
$\diam(\supp \psi_1\cup \supp \psi_2)\le \epsilon_0T$ and $|\tau|\le \epsilon$.
\end{enumerate}
\end{proposition}

\begin{proof}
Let $u= e^{ith^{-1}H}f$ with arbitrary $f$.

\begin{enumerate}[label=(\roman*), fullwidth]
\item\label{proof-26-2-11-i}
Statement \ref{prop-26-2-11-i} is easily proven by the same arguments as in the proof of Theorem~\ref{book_new-thm-2-1-2}: we consider just usual function $\phi(x)$ and operators of multiplication like $\upchi(\phi(x))$ so there are no ``bad'' commutators due to non-smoothness of $A$ or $V$.

\item\label{proof-26-2-11-ii}
Statement \ref{prop-26-2-11-ii} is also proven by the same arguments; however in this case $\phi=\phi(x,\xi)$ so we need to involve ``bad'' commutators but their contributions are bounded by
\begin{equation*}
C \| Q_1 u\| \cdot \Bigl( h^{1+\theta}\bigl(\3 A\3_{\theta+1}+
\3 V\3_{\theta+1}\bigr) \|u\| +h^{1+\delta} \|Q'u\|\Bigr)
\end{equation*}
in the right-hand expression while the left-hand expression is
$\epsilon h\|Q_1u\|^2$ where $Q$, $Q_1$, and $Q'$ are operators with symbols
$\upchi  (\phi(x,\xi))$, $\upchi_1 (\phi(x,\xi))$, and
$\upchi_1 (\phi(x,\xi)-\eta)$ respectively, $\eta>0$ is an arbitrarily small constant (so the latter symbol has a bit larger support than the former one), $\delta>0$ is a small exponent, $\upchi_1(t)=(-\upchi'(t))^{\frac{1}{2}}$, and f.e. $\upchi(t)=e^{-|t|^{-1}}$ as $t<0$, $\upchi(t)=0$ as $t\ge 0$.

Therefore we conclude that
\begin{equation*}
\|Qu\|\le C h^{\theta}\bigl(\3 A\3_{\theta+1}+
\3 V\3_{\theta+1}\bigr) \|u\| +C h^{\delta} \|Q'u\|
\end{equation*}
and similarly we can estimate $\|Q'u\|$ with $\|Q''u\|$ in the right-hand expression etc and thus we conclude that
\begin{equation*}
\|Qu\|\le C h^{\theta}\bigl(\3 A\3_{\theta+1}+
\3 V\3_{\theta+1}\|u\| + Ch^s\|u\|
\end{equation*}
which is what we need.

\item\label{proof-26-2-11-iii}
Statement \ref{prop-26-2-11-iii} is easily proven by the same arguments as in the proof of Theorem~\ref{book_new-thm-2-1-2}: we consider just usual function $\phi(x)$ and operators of multiplication like $\upchi(\phi(x))$ so there are no ``bad'' commutators due to non-smoothness of $A$ or $V$. However we need to consider a contribution of $u$ which is not confined to the small vicinity of $(y,\eta)$ and we need Statement~\ref{prop-26-2-11-ii} for this so the last term in the right-hand expression of (\ref{26-2-37}) is inherited.

We leave easy details to the reader.
\end{enumerate}
\end{proof}

\begin{remark}\label{rem-26-2-12}
\begin{enumerate}[label=(\roman*), fullwidth]
\item\label{rem-26-2-12-i}
Statement \ref{prop-26-2-11-i} means the finite propagation speed with respect to $x$;

\item\label{rem-26-2-12-ii}
Statement \ref{prop-26-2-11-ii} means the finite propagation speed with respect to $\xi$ and the last term in the right-hand expression of (\ref{26-2-35}) is due to the non-smoothness of $A$ and $V$;

\item\label{rem-26-2-12-iii}
Statement \ref{prop-26-2-11-iii} means that under assumption (\ref{26-2-36}) there actually is a propagation with respect to $x$;

\item\label{rem-26-2-12-iv}
So far we have not assumed that $V$ is very smooth function; we actually do not need it at all: it is sufficient to assume that $\partial V$ is very smooth in microscale $\varepsilon=h^{1-\delta}$; one can actually invoke more delicate arguments of the proof of Theorem~\ref{book_new-thm-2-3-1} and deal with microscale $\varepsilon= Ch|\log h|$.
\end{enumerate}
\end{remark}

Therefore in the framework of Proposition~\ref{prop-26-2-11}\ref{prop-26-2-11-iii} estimate
\begin{multline}
|F_{t\to h^{-1}\tau} \chi_T(t)
\bigl((hD_x)^\alpha (hD_y)^\beta U(x,y,t)\bigr)\bigr|_{x=y}|\le \\[3pt]
Ch^{1-d+s} T^{-s}+ C T^2 h^{-d+\theta}
\bigl(\3 A\3_{\theta+1}+\3 V \3_{\theta+1}\bigr)
\label{26-2-38}
\end{multline}
holds for all $\alpha:|\alpha|\le 2$, $\beta:|\beta|\le 2$, $s$  as $T=\epsilon$ and $|\tau|\le \epsilon$ where as usual
$\chi \in \sC^\infty _0([-1,-\frac{1}{2}]\cup [\frac{1}{2},1])$, $\chi_T(t)=\chi(t/T)$.

Let us consider $T\in (Ch, \epsilon)$; then we apply the standard rescaling
$t\mapsto tT^{-1}$, $x\mapsto xT^{-1}$, $h\mapsto hT^{-1}$ and assumptions (\ref{26-2-32}), (\ref{26-2-33}) are replaced by weaker assumptions
\begin{gather}
 T \|\partial V\|_{\sC(B(x,1))} \le C_0
\tag*{$\textup{(\ref*{26-2-32})}'$}\label{26-2-32-'}\\
\shortintertext{and}
T \|\partial  A\|_{\sC(B(x,1))} \le C_0.
\tag*{$\textup{(\ref*{26-2-33})}'$}\label{26-2-33-'}
\end{gather}
Further, $\3 A\3_{\theta+1}$ and $\3 V\3_{\theta+1}$ acquire factor $T^{\theta+1}$.

Furthermore, as $U(x,y,t)$ is a density with respect to $y$ we need to add factor $T^{-d}$ to the right-hand expression and due to $F_{t\to h^{-1}\tau}$ we need to add another factor $T$ and after these substitution and multiplications we arrive to

\begin{proposition}\label{prop-26-2-13}
Let $h\le T\le \epsilon$ and assumptions \ref{26-2-32-'}, \ref{26-2-33-'} and \textup{(\ref{26-2-36})} be fulfilled. Then  estimate \textup{(\ref{26-2-38})} holds.
\end{proposition}

Next we apply our standard arguments:

\begin{proposition}\label{prop-26-2-14}
In the framework of proposition~\ref{prop-26-2-13}
\begin{multline}
|F_{t\to h^{-1}\tau} \bigl[\bar{\chi}_T(t)
\bigl((hD_x)^\alpha (hD_y)^\beta U(x,y,t)\bigr)\bigr]\bigr|_{x=y}|\le \\[4pt]
Ch^{1-d}+ C T^2 h^{-d+\theta}  \bigl(\3 A\3_{\theta+1}+\3 V \3_{\theta+1}\bigr)
\label{26-2-39}
\end{multline}
provided $\bar{\chi}\in \sC^\infty_0 ([-1,1])$ and
\begin{multline}
|\bigl[\bigl((hD_x -A(x))\cdot\boldupsigma\bigr)^\alpha
\bigl((hD_y-A(y))\cdot\boldupsigma\bigr)^\beta e(x,y,t)\bigr]\bigr|_{x=y} -\\[3pt]
\shoveright{\Weyl_{\alpha,\beta} (x)|\le}\\[3pt]
Ch^{1-d}\bigl(1+\|\partial A\|_{\sC(B(x,1))}+\|\partial V\|_{\sC(B(x,1))}\bigr) +
 C h^{-d+\frac{1}{2}(\theta+1)}
\bigl(\3 A\3_{\theta+1}+\3 V \3_{\theta+1}\bigr)^{\frac{1}{2}}
\label{26-2-40}
\end{multline}
where
\begin{gather}
\Weyl_{\alpha,\beta}(x) \Def \const \, h^{-d}
\int_{\{H(x,\xi) \le \tau\}}
 \bigl((\xi -A(x))\cdot\boldupsigma \bigr)^{\alpha+\beta}\, d\xi
\label{26-2-41}\\
\intertext{is the corresponding Weyl expression and}
H(x,\xi)= \bigl((\xi -A(x))\cdot \boldupsigma \bigr)^2-V(x);
\label{26-2-42}
\end{gather}
in particular $\Weyl_{\alpha,\beta}(x)=0$ as $|\alpha|+|\beta|=1$.
\end{proposition}

\begin{proof}
Obviously summation of (\ref{26-2-38}) over $C_0h \le |t|\le T$ and a trivial estimate by $Ch^{1-d}$ of the contribution of the interval $|t|\le C_0 h$ implies (\ref{26-2-39}).

Then the standard Tauberian arguments and (\ref{26-2-39}) imply that the left-hand expression of (\ref{26-2-40}) does not exceed
\begin{equation*}
CT^{-1}h^{1-d} +
CT h^{-d+\theta}\bigl(\3 A\3_{\theta+1}+\3 V \3_{\theta+1}\bigr).
\end{equation*}
Optimizing with respect to $T\le \epsilon$ such that \ref{26-2-32-'},
\ref{26-2-33-'} hold we pick up $T=T^*$ with
\begin{multline}
T^*=\\
 \epsilon \min \Bigl(1,\,
\bigl(\|\partial A\|_{\sC(B(x,1))} +\|\partial V\|_{\sC(B(x,1))} \bigr)^{-1},\,
h^{-\frac{1}{2}(\theta-1)}
\bigl(\3 A\3_{\theta+1}+\3 V \3_{\theta+1}\bigr)^{-\frac{1}{2}}\Bigr).
\label{26-2-43}
\end{multline}
Meanwhile the Tauberian formula and (\ref{26-2-38}) imply that the contribution of an interval $\{t: |t|\asymp T\}$ with $ h\le T\le T^*$ to the Tauberian expression does not exceed the right-hand expression of (\ref{26-2-38}) divided by $T$, i.e.
\begin{equation*}
Ch^{1-d+s}T^{-s-1} + C T h^{-d+\theta}  \|\partial  A\|_{\sC^{\theta}};
\end{equation*}
summation over $T_*\Def h^{1-\delta}\le T\le T^*$ results in  the right-hand expression of (\ref{26-2-40}).

So, we need to calculate only the contribution of $\{t:|t|\le T_*\}$ but one can see easily that modulo indicated error  it coincides with $\Weyl_{\alpha,\beta}$.
\end{proof}

\begin{remark}\label{rem-26-2-15}
As $d\ge 3$  one can skip assumption \textup{(\ref{26-2-36})}.
\end{remark}

Indeed, we can apply the standard rescaling technique: $x\mapsto x \ell^{-1}$,
$h\mapsto \hbar= h\ell^{-\frac{3}{2}}$, $A\mapsto A \ell^{-\frac{1}{2}}$, $V\mapsto V\ell^{-1}$ with
$\ell =  \max(\epsilon|V| \nu^{-1},\ h^{\frac{2}{3}}\nu^{-\frac{1}{3}})$,
$\nu =(1+|\partial V|_{\sC})$; see Section~\ref{book_new-sect-5-1}.

\subsection{Application}
\label{sect-26-2-2-2}

Let us apply developed technique to estimate a minimizer.
\begin{proposition}\label{prop-26-2-16}
Let $\kappa\le c$   and let $A$ be a minimizer. Let
\begin{equation}
\mu \Def \|\partial  A\|_{\sC}\le C h^{-1+\delta}.
\label{26-2-44}
\end{equation}
As $d=2$ let assumption \textup{(\ref{26-2-36})} be also fulfilled.  Then as
$\theta\in (1,2)$ estimate
\begin{multline}
\|\partial  A\|_{\sC^{\theta-1}} + h^{\theta-1} \|\partial  A\|_{\sC^{\theta}} \le \\
C\kappa \bigl( 1+ \|V\|_{\sC^1} + h^{\frac{1}{2}(\theta-1)} \|V\|_{\sC^{\theta+1}} ^{\frac{1}{2}}\bigr)+  C\|\partial A\|'\label{26-2-45}
\end{multline}
holds with
\begin{equation}
\|\partial A\|'\Def\sup _y \|\partial A\|_{\sL^2(B(y,1))}.
\label{26-2-46}
\end{equation}
\end{proposition}

\begin{proof}
Consider expression for $\Delta A$. According to equation (\ref{26-2-14}) and Proposition~\ref{prop-26-2-14}    $|\Delta A|+ |h\partial \Delta A| $ does not exceed the right-hand expression of (\ref{26-2-40}) multiplied by
$C\kappa h^{d-1}$ i.e.
\begin{multline}
\|\Delta A\|_\sC+ \|h\partial \Delta A\|_\sC \le\\
C\kappa \Bigl( 1+|\partial A|_{\sC}  + |\partial V|_{\sC} +
 h^{\frac{1}{2}(\theta-1)}
\|\partial  A\|_{\sC^{\theta}}^{\frac{1}{2}}+ h^{\frac{1}{2}(\theta-1)} \| \partial V\|_{\sC^\theta}^{\frac{1}{2}}\Bigr)
\label{26-2-47}
\end{multline}
where we replaced $\3 A\3_{\theta+1}$ and $\3 V\3_{\theta+1}$ by larger
$\| \partial A\|_{\sC^\theta}$ and $\| \partial V\|_{\sC^\theta}$ respectively.

Then the regularity theory for elliptic equations implies that
\begin{claim}\label{26-2-48}
For any $\theta '\in (1,2)$ \ $h^{\theta'-1} \|\partial A\|_{\sC^{\theta'}}$ does not exceed this expression (\ref{26-2-47}) plus $C\|\partial A\|'$.
\end{claim}
Note that $\|\partial  A\|_{\sC}$ does not exceed
$\epsilon \|\partial  A\|_{\sC^\theta}+ C'_\epsilon \|\partial A\|'$ with arbitrarily small constant $\epsilon>0$ and therefore
\begin{equation}
h^{\theta-1} \|\partial A\|_{\sC^{\theta} }+
\epsilon^{-1}\|\partial A\|_{\sC}
\label{26-2-49}
\end{equation}
does not exceed expression (\ref{26-2-47}) plus $C'_\epsilon \|\partial A\|'$ where we used (\ref{26-2-48}) for  $\theta'=\theta$.

Comparing (\ref{26-2-49}) and (\ref{26-2-47}) we conclude that for
$\kappa \le c$ and sufficiently small constant $\epsilon>0$ we can eliminate in the derived inequality both contributions of $\partial A$ to (\ref{26-2-47}) thus we arrive to (\ref{26-2-45}). \end{proof}

Having this strong estimate to $A$ allows us to prove

\begin{theorem}\label{thm-26-2-17}
Let $\kappa\le c$,  \textup{(\ref{26-2-44})} be fulfilled, and let $d= 3$. Assume that
\begin{gather}
\|V\|_{\sC^{\theta+1}}\le c\label{26-2-50}\\
\intertext{with $\theta \in (1,2)$. Then}
\E^*= \Weyl_1 +O(h^{2-d})\label{26-2-51}\\
\intertext{and a minimizer $A$ satisfies}
\|\partial A\|\le C \kappa^{\frac{1}{2}}h^{\frac{1}{2}}\label{26-2-52}\\
\shortintertext{and}
\|\partial A\|_{\sC^{\theta-1}} + h^{\theta-1}\|\partial A\|_{\sC^{\theta}}\le
C \kappa^{\frac{1}{2}}h^{\frac{1}{2}} +C\kappa.
\label{26-2-53}
\end{gather}
\end{theorem}

\begin{proof}
(a) In virtue of (\ref{26-2-39}) the Tauberian error when calculating
$\Tr (H^-_{A,V})$ does not exceed the right-hand expression of (\ref{26-2-39}) multiplied by $C T^{-2}$ i.e.
\begin{equation}
Ch^{1-d}T^{-2} +
C  h^{-d+\theta}  \bigl(\3 A\3_{\theta+1}+\3 V \3_{\theta+1}\bigr).
\label{26-2-54}
\end{equation}
Assumption (\ref{26-2-50}) allows us to simplify this expression and take $T\asymp (1+\mu)^{-1}$; applying estimate (\ref{26-2-45}) we conclude that the Tauberian error does not exceed
\begin{equation}
C(1+\mu)^2 h^{2-d}+ C(\kappa +\|\partial A\|') h^{2-d}.
\label{26-2-55}
\end{equation}
We claim that
\begin{claim}\label{26-2-56}
Weyl error\footnote{\label{foot-26-5} I.e. error when we replace Tauberian expression by Weyl expression.} when calculating $\Tr (H^-_{A,V})$ also does not exceed (\ref{26-2-55}).
\end{claim}
Then
\begin{multline}
\E (A) \ge \\
\Weyl_1 - C(1+\mu)^2 h^{2-d}- C(\kappa +\|\partial A\|') h^{2-d} +
\kappa^{-1} h^{1-d} \|\partial A\|^2 \ge \\
\Weyl_1 - Ch^{2-d} +  \frac{1}{2\kappa } h^{1-d} \|\partial A\|^2
\label{26-2-57}
\end{multline}
because $\mu \le C\|\partial A\|' +1$ due to (\ref{26-2-45}) and assumption (\ref{26-2-50}). This implies an estimate of $\E^*$ from below and combining with the estimate $\E^* \le \E^*(0) =\Weyl_1 + Ch^{2-d}$ from above we arrive to (\ref{26-2-51}) and (\ref{26-2-52}) and then (\ref{26-2-53}) due to (\ref{26-2-45}) and assumption (\ref{26-2-50}).

\bigskip\noindent
(b) To prove (\ref{26-2-56}) let us plug $A_\varepsilon$ instead of $A$ into $e_1(x,x,0)$. Then in virtue of the rough microlocal analysis contribution to Weyl error of $\{t: T_* \le |t|\le \epsilon\}$ with $T_*=h^{1-\delta}$ would be negligible and contribution of $\{t: |t|\le T_*\}$ would be
$\Weyl_1 + O(h^{2-d})$.

\bigskip\noindent
(c) Now let us calculate an error which we made plugging $A_\varepsilon$ instead of $A$ into $e_1(x,x,0)$. Obviously it does not exceed
$Ch^{-d}\|A-A_\varepsilon\|_\sC $ and since
$\|A-A_\varepsilon\|_\sC \le
C\varepsilon ^{\theta+1}\|\partial A\|_{\sC^\theta}$ this error does not exceed
$Ch^{\theta+1-d-4\delta}\|\partial A\|_{\sC^\theta}$ which is marginally worse than what we are looking for.

However it is good enough to recover a weaker version of (\ref{26-2-51}) and (\ref{26-2-52}) with an extra factor $h^{-\delta_1}$ in their right-hand expressions. Then (\ref{26-2-45}) implies a bit  weaker version of (\ref{26-2-53}) and in particular that its left-hand expression does not exceed $C$.

Knowing this let us consider the two term approximation. With the above knowledge one can prove easily that the error in two term approximation does not exceed $Ch^{3-d -\delta'}$ with $\delta '= 100\delta$.

Then the second term in the Tauberian expression is
\begin{equation}
\int \bigl((H_{A,V}-H_{A_\varepsilon,V})e^\T_{(\varepsilon)}(x,y,0)\bigr)
\bigr|_{y=x}\,dx.
\label{26-2-58}
\end{equation}
where subscript $_{(\varepsilon)}$ means that we plugged $A_\varepsilon$ instead of  $A$ and superscript $^\T$ means that we consider Tauberian expression with $T=T^*=\epsilon$. But then the contribution of $\{t: T_*\le |t|\le T^*\}$ is also negligible and modulo $Ch^{\theta+2-d-4\delta}\|\partial A\|_{\sC^\theta}$ we get a Weyl expression. However
\begin{equation}
(H_{A,V}-H_{A_\varepsilon,V}) = -2(\xi -A_\varepsilon)\cdot (A-A_\varepsilon)+
|A-A_\varepsilon|^2
\label{26-2-59}
\end{equation}
and the first term kills Weyl expression as an integrand is odd with respect to $(\xi -A_\varepsilon)$ while the second as one can see easily makes it smaller than $Ch^{3-d -\delta'}$. Therefore (\ref{26-2-56}) has been proven.
\end{proof}

\begin{remark}\label{rem-26-2-18}
\begin{enumerate}[fullwidth, label=(\roman*)]
\item\label{rem-26-2-18-i}
For $d=2$ we cannot drop assumption (\ref{26-2-36}) at the stage we did it for $d\ge 3$. However results of the next section allow us to cure this problem using partition-and-rescaling technique.

\item\label{rem-26-2-18-ii}
Actually  as $V\in \sC^{2,1}$ we have an estimate
\begin{equation}
|\partial  A(x)-\partial A(y)|\le C \kappa |x-y| (|\log |x-y||+1) + C\mu.
\label{26-2-60}
\end{equation}
Combining with (\ref{26-2-52}) we conclude that
\begin{equation}
\|\partial A\|_{\sC} \le C \kappa^{(d+1)/(d+2)}|\log h|^{d/(d+2)} h^{1/(d+2)}
\label{26-2-61}
\end{equation}
\item\label{rem-26-2-18-iii}
If (\ref{26-2-50}) holds for $(h\partial )^m V$ with $m\in \bZ^+$ then (\ref{26-2-52}) and (\ref{26-2-52}) also hold for $(h\partial )^m A$ instead of $A$; further, if  $(h\partial )^m V\in \sC^{2,1}$ then (\ref{26-2-60}) and (\ref{26-2-61}) also hold for $(h\partial )^m A$ instead of $A$.
\end{enumerate}
\end{remark}

\subsection{Classical dynamics and sharper estimates}
\label{sect-26-2-2-3}

Now we want to improve remainder estimate $O(h^{2-d})$ to $o(h^{2-d})$. Sure, we need to impose condition to the classical dynamical system and as
$|\partial A |=O(h^\delta)$ with $\delta>0$ due to (\ref{26-2-61}) it should be dynamical system associated with the Hamiltonian flow generated by $H_{0,V}$:

\begin{claim}\label{26-2-62}
The set of periodic points of the dynamical system associated with Hamiltonian flow generated by $H_{0,V}$ has measure $0$ on the energy level $0$.
\end{claim}
Recall that on $\{(x,\xi): H_{0,V}(x,\xi)=\tau\}$ a natural density $d\upmu_\tau= dxd\xi :dH|_{H=\tau}$ is defined.

The problem is we do not have a quantum propagation theory for $H_{A,V}$ as $A$ is not a ``rough'' function. However it is rather regular function, almost $\sC^2$, and $(A-A_\varepsilon)$ is rather small:
$|A-A_\varepsilon|\le \eta \Def Ch^{2-3\delta}$ and $|\partial (A-A_\varepsilon)|\le Ch^{1-3\delta}$ and therefore we can apply a method of successive approximations with the unperturbed operator $H_{A_\varepsilon, V}$ as long as $\eta T/h\le h^\sigma$ i.e. as $T\le h^{1-4\delta}$. Here we however have no use for such large $T$ and consider $T=O(h^{-\delta})$.

Consider
\begin{equation}
F_{t\to h^{-1} \tau} \chi_T(t) U(x,y,t),
\label{26-2-63}
\end{equation}
and consider terms of successive approximations. Then if we forget about microhyperbolicity arguments the first term will be $O(h^{-d}T)$, the second $O(h^{-1-d}\eta T^2)= O(h^{1-d-\delta'})$ and the error
$O(h^{-2-d}\eta^2 T^3)=O(h^{2-d-\delta''})$.

Therefore as our goal is $O(h^{1-d})$ we need to consider the first two terms only. The first term is the same expression (\ref{26-2-63}) with $U$ replaced by $U_{(\varepsilon)}$.

Consider the second term, it corresponds to $U'_{(\varepsilon)}(x,y,t)$ which is the Schwartz kernel of operator
\begin{gather}
\sfU'_{(\varepsilon)}\Def  i h^{-1}\int _0^t
e^{i(t-t') h^{-1}H_{A_\varepsilon,V} } \bigl(H_{A,V}-H_{A_\varepsilon,V}\bigr) e^{it' h^{-1}H_{A_\varepsilon,V} } \,dt'
\label{26-2-64}\\
\shortintertext{and then}
\Tr \bigl(\sfU'_{(\varepsilon)}  \psi \bigr)= ih^{-1} \Tr \Bigl(\bigl(H_{A,V}-H_{A_\varepsilon,V}\bigr)
  e^{ih^{-1}tH_{A_\varepsilon,V}} \psi^1 (t) \Bigr)\label{26-2-65}\\
\shortintertext{with}
\psi^1(t) \Def \int _0^t e^{ih^{-1}t'H_{A_\varepsilon,V}} \psi e^{-ih^{-1}t'H_{A_\varepsilon,V}}\,dt'\notag
\end{gather}
is $h$-pseudo-differential operator with a rough symbol and $\psi^1 (t) \sim t$.

Really, one can prove  easily studying first the Hamiltonian flow equation and then the transport equations that
$\psi_t\Def
e^{ih^{-1}t H_{A_\varepsilon,V}}\psi e^{-ih^{-1}t H_{A_\varepsilon,V}}$  is a $h$-pseudo-differential operator with a rough symbol and its corresponding norm is bounded.

Note that
\begin{equation*}
ih^{-1} F_{t\to h^{-1}\tau}  e^{ih^{-1}tH_{A_\varepsilon,V}} \psi^1 (t)=
(2\pi) \int \bigl(F_{t\to h^{-1}\tau'}  e^{ih^{-1}tH_{A_\varepsilon,V}} \bigr) \hat{f}(h^{-1}(\tau -\tau')) \,d\tau'
\end{equation*}
with $\hat{f}=F_{t\to \tau} f_t$, $f_t=\chi_T(t)\psi^1(t)$ and therefore (\ref{26-2-64})--(\ref{26-2-65}) imply that
\begin{equation}
|F_{t\to h^{-1}\tau}\chi_T(t) \Tr \sfU'_{(\varepsilon)} \psi |\le
C\eta T^2 h^{-d}
\label{26-2-66}
\end{equation}
where in comparison with the trivial estimate we gained factor $h$.

We can plug here $T'\in (T_*,T)$ instead of $T$ and taking summation by $T'$ from $T_*=\epsilon$ to $T$ we conclude that (\ref{26-2-66}) also holds for $\chi_T(t)$ replaced by  $\bigl(\bar{\chi}_T(t)-\bar{\chi}_{T_*}(t)\bigr)$ (provided $\bar{\chi}=1$ on $(-\frac{1}{2},\frac{1}{2})$) and  since
$\eta T^2\le h^{1+\delta}$ as $T\le h^{-\delta}$ the right-hand expression (\ref{26-2-66}) does not exceed $Ch^{1-d+\delta}$.

On the other hand, our traditional methods imply that as $d\ge 3$
\begin{equation}
|F_{t\to h^{-1} \tau} \chi_T(t)
\Tr  \bigl(e^{it h^{-1}H_{A_\varepsilon,V} } \psi  \bigr)|\le
Ch^{1-d} T \upmu (\Pi_{T, \zeta}) + C_{T,\zeta} h^{1-d+\delta}
\label{26-2-67}
\end{equation}
where $\Pi _T$ is the set of  points on energy level $0$, periodic with periods not exceeding $T$,  $\Pi _{T,\zeta}$ is its $\zeta$-vicinity, $\zeta>0$ is arbitrarily small.

Here again we can plug any $T'\in (T_*,T)$   instead of $T$  and after summation with respect to $T'$ we conclude that (\ref{26-2-67}) also holds with $\chi_T(t)$ replaced by  $\bigl(\bar{\chi}_T(t)-\bar{\chi}_{T_*}(t)\bigr)$.

Combining with estimate for $\bigl(e^{it h^{-1}H_{A_\varepsilon,V} }-
 e^{it h^{-1}H_{A_\varepsilon,V} }\bigr)$ we conclude that
\begin{equation*}
|F_{t\to h^{-1} \tau} \bigl(\bar{\chi}_T(t)-\bar{\chi}_{T_*}(t)\bigr)
\Tr  \bigl(e^{it' h^{-1}H_{A,V} } \psi  \bigr)|\le
Ch^{1-d} T \upmu (\Pi_{T, \zeta}) + C_{T,\zeta} h^{1-d+\delta}
\end{equation*}
and since
\begin{equation*}
|F_{t\to h^{-1} \tau} \bar{\chi}_{T_*}(t)
\Tr  \bigl(e^{it' h^{-1}H_{A,V} } \psi  \bigr)|\le Ch^{1-d}
\end{equation*}
we conclude that
\begin{multline}
|F_{t\to h^{-1} \tau} \bar{\chi}_T(t)
\Tr  \bigl(e^{it' h^{-1}H_{A,V} } \psi  \bigr)|\le \\
Ch^{1-d}
+Ch^{1-d} T \upmu (\Pi_{T, \zeta}) + C_{T,\zeta} h^{1-d+\delta}
\label{26-2-68}
\end{multline}
and then the Tauberian error does not exceed the right-hand expression of (\ref{26-2-68}) multiplied by $Ch T^{-2}$ and it is less than $CT^{-1}h^{2-d}$.

Consider now the Tauberian expression and again apply two-term approximation for $e^{ih^{-1}tH_{A,V}}$ considering $e^{ih^{-1}tH_{A_\varepsilon,V}}$ as an unperturbed operator; then the error will be less than $Ch^{2-d+\delta}$.

Consider the second term after taking trace; it is $O(h^{2-d-4\delta})$, so it is just slightly too large. Further, if $\psi=I$ one can calculate it easily and observe that it is $O(h^{2-d+\delta})$ provided $V\in \sC^{2,1}$.

Finally, the first term is what we get for $e^{ih^{-1}tH_{A_\varepsilon,V}}$ and in virtue of rough microlocal analysis contribution of $\{t:\,T_*\le |t|\le T\}$ does not exceed
$Ch^{2-d}  \upmu (\Pi_{T, \zeta}) + C_{T,\zeta} h^{2-d+\delta}$ and contribution of $\{t:\, |t|\le T_*\}$ is $\Weyl_1+ O(h^{2-d+\delta})$.

Then we arrive to

\begin{theorem}\label{thm-26-2-19}
Let $\kappa\le c$,  \textup{(\ref{26-2-44})} and \textup{(\ref{26-2-50})} be fulfilled, and let $d= 3$. Furthermore, let condition \textup{(\ref{26-2-62})} be fulfilled\,\footnote{\label{foot-26-6} I.e. $\upmu _0(\Pi_\infty)=0$.}. Then
\begin{gather}
\E^*= \Weyl^*_1 +o(h^{2-d})\label{26-2-69}\\
\shortintertext{where}
\Weyl^*_1 =\Weyl_1 +\varkappa h^{2-d}\int V_+^{\frac{d}{2}}\Delta V  \,dx
\label{26-2-70}
\end{gather}
calculated in the standard way for $H_{0,V}$ and a minimizer $A$ satisfies similarly improved versions of \textup{(\ref{26-2-52})} and \textup{(\ref{26-2-53})}.
\end{theorem}

\begin{remark}\label{rem-26-2-20}
\begin{enumerate}[label=(\roman*), fullwidth]
\item\label{rem-26-2-20-i}
Under stronger assumptions to the Hamiltonian flow one can recover better estimates like  $O(h^{2-d}|\log h|^{-2})$ or even $O(h^{2+\delta-d})$ (like in  Subsubsection~\ref{book_new-sect-4-4-4-3}.3 ``\nameref{book_new-sect-4-4-4-3}'').

\item\label{rem-26-2-20-ii}
We leave to the reader to calculate the numerical constants $\varkappa_*$ here and in (\ref{26-2-75}) below, $\varkappa=\varkappa_1-\frac{2}{d}\varkappa_2$.

\item\label{rem-26-2-20-iii}
However, even if $\psi\ne I$ we can observe that it is sufficient to consider only principal terms and then the second  term in approximations is also $O(h^{2-d+\delta})$ provided $V\in \sC^{2,1}$ as long as principal symbol of $\psi (x)$ is even with respect to $\xi$, in particular, if $\psi=\psi(x)$.
\end{enumerate}
\end{remark}

\section{Local theory}
\label{sect-26-2-3}

\subsection{Localization and estimate from above}
\label{sect-26-2-3-1}

The results of the previous Subsection have two shortcomings: first, they impose the excessive initial requirement (\ref{26-2-21}) to $\kappa$ as a priory $M\le c h^{-3}$; second, they are not local. However curing the second shortcoming we make the way to addressing the first one as well using the partition and rescaling technique.

We can localize $\Tr^-(H)=\Tr(H^-)$ which is the first term in $\E(A)$ either in  our traditional way as $\Tr (H^-\psi^2)$ or in the way favored by some mathematical physicists\footnote{\label{foot-26-7} See f.~e.  L.~Erd\"os, S.~Fournais. and J.~P.~Solovej \cite{EFS1}.}: namely, we take
$\Tr^- (\psi H\psi )$ where in both cases
$\psi\in \sC^\infty_0(B(0,\frac{1}{2}))$, $0\le \psi \le 1$ and some other  conditions will be imposed to it later.

Note that
\begin{equation}
\Tr^-(\psi H\psi ) \ge \Tr (\psi H^-\psi )=\int e_1(x,x,0) \psi^2(x)\,dx.
\label{26-2-71}
\end{equation}
Really, decompose operator $H=H \uptheta(-H) + H(1-\uptheta(-H))$ where $\uptheta (\tau-H)$ is a spectral projector of $H$ and therefore in the operator sense $H\ge H^-\Def H \uptheta(-H)$ and $\psi H\psi \ge \psi H^-\psi$ and therefore all negative eigenvalues of  $\psi H\psi$ are greater than or equal to eigenvalues of the negative operator $\psi H^-\psi$ and then
\begin{equation}
\Tr^-(\psi H\psi)\ge \Tr (\psi H ^- \psi)=
\Tr \int^0_{-\infty} \tau d_\tau \uptheta(\tau-H) \psi^2
\label{26-2-72}
\end{equation}
which is exactly the right-hand expression of (\ref{26-2-71}).

\begin{remark}\label{rem-26-2-21}
Each approach has its own advantages.
\begin{enumerate}[label=(\roman*), fullwidth]
\item\label{rem-26-2-21-i}
In particular, no need to localize $A$ (see (ii)) and the fact that proposition~\ref{prop-26-2-5} obviously remains true are advantages of
$\Tr^-(\psi H\psi)$-localization.

\item\label{rem-26-2-21-ii}
Further, as $\Tr^-(\psi H\psi )$ does not depend on $A$ outside of $B(0,\frac{3}{4})$ we may assume that $A=0$ outside of $B(0,1)$. Really, we can always subtract a constant from $A$ without affecting traces and also cut-off $A$ outside of $B(0,1)$ in a way such that $A'=A$ in $B(0,\frac{3}{4})$ and
$\|\partial A'\| \le c\|\partial A\|_{B(0,1)}$; the price is to multiply $\kappa$ by $c^{-1}$--as long as principal parts of asymptotics coincide.

\item\label{rem-26-2-21-iii}
On the other hand, additivity rather than sub-additivity (\ref{26-2-87}) and the trivial estimate from the above are advantages of
$\Tr (\psi H ^-\psi)$-localization; therefore it is more advantageous

\item\label{rem-26-2-21-iv}
In the next Chapter~\ref{book_new-sect-27} (in Section~\ref{book_new-sect-27-2}) we will use  more $\Tr^-(\psi H\psi )$-localization for preliminary estimates from below and simplify many arguments of this Section. We apply these modifications and simplifications to this Section in the final version of the Book.
\end{enumerate}
\end{remark}

We will use both methods and here we provide an upper estimate for larger
expression $\Tr^-(\psi H\psi )$ and a lower estimate for lesser expression $\Tr (\psi H ^- \psi)$. Let us estimate from the above:

\begin{proposition}\label{prop-26-2-22}
Assume that $V\in \sC^{2,1}$, $d\ge 2$.  Let $\ell(x)$ be a scaling function\footnote{\label{foot-26-8} I.e. $\ell \ge 0$ and
$|\partial \ell|\le \frac{1}{2}$.} and $\psi$ be a function such that  $|\partial^\alpha \psi|\le c\psi \ell^{-\sigma|\alpha|}$ for all
$\alpha:|\alpha |\le 2$  and $|\psi |\le c\ell^{\sigma(1+\delta)}$ with $\sigma>1$ and $\delta>0$\,\footnote{\label{foot-26-9} Such compactly supported functions obviously exist.}.

Then, as $A=0$,
\begin{gather}
\Tr^- (\psi H\psi) = \int \Weyl_1(x)\psi^2(x)\,dx + O(h^{2-d})
\label{26-2-73}\\
\intertext{and under assumption \textup{(\ref{26-2-62})}}
\Tr^- (\psi H\psi) = \int \Weyl^*_1(x)\psi^2(x)\,dx + o(h^{2-d})
\label{26-2-74}\\
\shortintertext{with}
\Weyl^*_1(x)= \Weyl_1(x) + \varkappa_1 h^{-1}V_+^{\frac{d}{2}}\Delta V +
\varkappa_2 h^{-1}V_+^{\frac{d}{2}-1}|\nabla V|^2\label{26-2-75}
\end{gather}
calculated in the standard way for $H_{0,V}$.
\end{proposition}

\begin{proof}
Let us consider $\tilde{H}=\psi H\psi $ as a Hamiltonian and let $\tilde{e}(x,y,\tau)$ be the Schwartz kernel of its spectral projector. Then
\begin{equation}
\Tr^- (\psi H\psi)= \int  \tilde{e}_1 (x,x,0)\,dx =
\sum_j \int  \tilde{e}_1 (x,x,0)\psi_j^2\,dx
\label{26-2-76}
\end{equation}
where $\psi_j^2$ form a partition of unity in $\bR^d$ and we need to calculate the right hand expression. The problem is that $\tilde{H}$ is not a usual Schr\"odinger operator because of degenerating factor $\psi$ on each side.

Consider first an $\epsilon \ell$-admissible partition of unity in $B(0,1)$. Let us consider $\gamma$-scale in such element where
$\gamma = \epsilon \ell^\sigma $ and we will use $1$ scale in $\xi$.  Then after rescaling $x\mapsto x \gamma^{-1}$ the semiclassical parameter rescales
$h\mapsto h_\new = h \gamma^{-1}$ and the contribution of each $\gamma$-subelement to a semiclassical remainder  does not exceed
$C\psi^2 (h/\gamma)^{2-d}$ with $\psi \le \gamma^{1+\delta}$ having the same magnitude over element as $\gamma \ge 2h$. Then contribution of $\ell$-element to a semiclassical error does not exceed
$C\psi^2 (h/\gamma)^{2-d}\times \ell^d \gamma^{2-d} \asymp Ch^{2-d} \psi^2\gamma^{-2} \ell^{d} \le Ch^{2-d}\ell^{d+2\delta}$.

Note that expression (\ref{26-2-76}) only increases if we sum only with respect to elements where $\ell^\sigma \ge h$. Therefore we arrive to estimate
\begin{equation*}
\Tr^- (\psi H\psi)\le \int \Weyl_1(x)\psi^2(x)\,dx + Ch^{2-d}
\end{equation*}
where integration is taken over domain $\{x:\, \ell(x)\ge h^{1/\sigma}\}$. Note that we can extend this integral to $\bR^d$: really, it will add negative term with absolute value not exceeding $Ch^{-d}\times h^{2+\delta}$ as
$\psi \le h^{1+\delta}$ there and it is absorbed by the remainder estimate. \end{proof}

\begin{corollary}\label{cor-26-2-23}
In the framework of proposition~\ref{prop-26-2-22}
\begin{gather}
\E_\psi ^* \Def  \inf_{A} \E_\psi (A) \le
\int \Weyl_1(x)\psi^2(x)\,dx + Ch^{2-d}
\label{26-2-77}\\
\intertext{and under assumption \textup{(\ref{26-2-62})}}
\E_\psi ^*  \le \int \Weyl_1(x)\psi^2(x)\,dx + Ch^{2-d}\label{26-2-78}\\
\shortintertext{with}
\E _\psi (A)\Def \Tr^- (\psi H\psi)+ \frac{1}{\kappa h^2}\int |\partial A|^2\,dx \label{26-2-79}.
\end{gather}
\end{corollary}

Really, we just pick $A=0$.

\subsection{Estimate from below}
\label{sect-26-2-3-2}

Now let us estimate \emph{redefined\/} $\E_\psi (A)$,
\begin{equation}
\E_\psi (A) \Def\int e_1(x,x,0)\psi ^2(x) \,dx +
\frac{1}{\kappa h^{d-1}}\int  |\partial A|^2\,dx.
\label{26-2-80}
\end{equation}
from below. However we need an equation for an optimizer and it would be easier for us to deal with even lesser expression involving $\tau$-regularization. Let us rewrite the first term in the right-hand expression in the form
\begin{multline*}
\int^0_{-\infty}  \bar{\varphi}(\tau/L) \tau\, d_\tau e(x,x,\tau) +
\int^0_{-\infty}  (1-\bar{\varphi}(\tau/L)) \tau\, d_\tau e(x,x,\tau)\ge\\
\int ^L_{-\infty}\Bigl( \bar{\varphi} (\tau/L) (\tau-L)\, d_\tau e(x,x,\tau) +
  (1-\bar{\varphi}(\tau/L)) \tau\, d_\tau e(x,x,\tau)\Bigr)
\end{multline*}
where $\bar{\varphi}\in \sC^\infty_0 ([-1,1])$ equals $1$ in $[-\frac{1}{2},\frac{1}{2}]$ and let us estimate from  below
\begin{multline}
\E'_\psi (A)\Def
\int \Bigl(\int^L_{-\infty} \bar{\varphi}(\tau/L) (\tau-L) d_\tau e(x,x,\tau)(x) + \\
(1-\bar{\varphi}(\tau/L)) (\tau-L)\, d_\tau e(x,x,\tau)\Bigr)\psi ^2(x)\,dx+ 
\frac{1}{kh^{1-d}}\int |\partial A|^2\,dx  
\label{26-2-81}
\end{multline}

Let us generalize Proposition~\ref{prop-26-2-4}:

\begin{proposition}\label{prop-26-2-24}
Let $A$ be a minimizer of $\E'_\psi (A)$. Then
\begin{multline}
\frac{2}{\kappa h^{1-d}}\Delta A_j (x)= \Phi_j\Def \\
\Re\tr \upsigma_j\Bigl(  (hD-A)_x \cdot \boldupsigma \cK (x,y,\tau) +
\cK (x,y,\tau) \,^t (hD-A)_y \cdot \boldupsigma\Bigr)\Bigr|_{y=x}
\label{26-2-82}
\end{multline}
with
\begin{multline*}
\cK=
\int ^L_{-\infty} {\mathsf{SK}}\Bigl[\bar{\varphi} (\tau/L) (\tau-L) \Res_\bR (\tau -H)^{-1} \psi^2 (\tau -H)^{-1}  +\\
(1-\bar{\varphi}(\tau/L)) \tau (\tau-L)\Res_\bR (\tau -H)^{-1}
\psi^2 (\tau -H)^{-1} \Bigr] (x,y)\,d\tau \
\end{multline*}
where we use a temporary notation ${\mathsf{SK}}[B](x,y)$ for the Schwartz kernel of operator $B$.
\end{proposition}

\begin{proof}
Follows immediately from the proof of Proposition~\ref{prop-26-2-4}.
\end{proof}

\begin{proposition}\label{prop-26-2-25}
Let $d=3$ and assumptions  \textup{(\ref{26-2-20})} and \textup{(\ref{26-2-21})} be fulfilled. Then as $\tau\le c$
\begin{enumerate}[label=(\roman*), fullwidth]
\item\label{prop-26-2-25-i}
Operator norm in $\sL^2$ of $(hD)^k (\tau -H)^{-1}$ does not exceed
$C|\Im \tau|^{-1}$ for $k=0,1,2$;

\item\label{prop-26-2-25-ii}
Operator norm in $\sL^2$ of
$(hD)^2\bigl((hD-A)\cdot\boldupsigma\bigr) (\tau -H)^{-1}$ does not exceed $C|\Im \tau|^{-1}$ for $k=0,1,2$.
\end{enumerate}
\end{proposition}

\begin{proof}
Proof follows the same scheme as the proof of Proposition~\ref{prop-26-2-6}.
\end{proof}

\begin{proposition}\label{prop-26-2-26}
Let $d=3$  and assumptions \textup{(\ref{26-2-20})} and \textup{(\ref{26-2-21})} be fulfilled. Then $|\Phi (x)|\le Ch^{-3}$.
\end{proposition}

\begin{proof}
Let us  estimate
\begin{equation}
|\int \tau \varphi (\tau/L)  \Res_\bR
\mathsf{SK} \Bigl[ T (\tau -H)^{-1} \psi ^2 (\tau -H)^{-1} \Bigr] (x,y)\,d\tau|
\label{26-2-83}
\end{equation}
where  $L\le c$ and $\varphi \in \sC^\infty_0 ([-1,1])$  and also a similar expression with a factor $(\tau -L)$ instead of $\tau$; here either $T=I$, or $T=(hD_k-A_k)_x$ or $T=(hD_k-A_k)_y$.

Proposition \ref{prop-26-2-25} implies that the Schwartz kernel of the integrand does not exceed $Ch^{-3}|\Im \tau|^{-2}$ and therefore expression (\ref{26-2-83}) does not exceed $CL^2 \times h^{-3}L^{-2}= Ch^{-3}$.

Then what comes out in $\Phi$ from the term with the factor
$\bar{\phi}(\tau /h)$  does not exceed $Ch^{-3}$.

Then representing
$\bigl(1-\bar{\phi}(\tau /h)\bigr)$ as a sum of $\varphi (\tau /L)$ with $L= 2^n h$ with $n=0,\ldots, \lfloor |\log h|\rfloor +c$ we estimate the output of each term by $Ch^{-3}$ and thus the whole sum by $Ch^{-3}|\log h|$.

\medskip
To get rid off the logarithmic factor we use equality
\begin{equation}
(\tau -H)^{-1}\psi (\tau -H)^{-1}=-\partial (\tau -H)^{-1}\psi
+ (\tau -H)^{-2}[h,\psi ](\tau -H)^{-1};
\label{26-2-84}
\end{equation}
if we plug only the second part we recover a factor $h/L$ where $h$ comes from the commutator and $1/L$ from the increased singularity; an extra operator factor in the commutator is not essential. Then summation over partition results in $Ch^{-3}$.

Plugging only the first part we do not use the above decomposition but an equality
$\Res_\bR (\tau -H)^{-1}\, d\tau =d_\tau \uptheta (\tau -H)$.
\end{proof}

\begin{corollary}\label{cor-26-2-27}
Let $d=3$, \textup{(\ref{26-2-20})} and \textup{(\ref{26-2-21})} be fulfilled and $A$ be a minimizer. Then \textup{(\ref{26-2-27})} and \textup{(\ref{26-2-28})} hold.
\end{corollary}

\begin{proof}
Proof follows the proof of corollary~\ref{cor-26-2-8}.
\end{proof}

Now we can recover both proposition~\ref{prop-26-2-16} and main theorems~\ref{thm-26-2-17} and~\ref{thm-26-2-19}:

\begin{theorem}\label{thm-26-2-28}
Let $d=3$ and assumptions \textup{(\ref{26-2-20})} and $\kappa\le c$ be fulfilled. Then
\begin{enumerate}[label=(\roman*), fullwidth]
\item The following estimate holds:
\begin{equation}
\E^*_\psi- \int \Weyl_1(x)\psi^2(x)\,dx  = O(h^{2-d})
\label{26-2-85}
\end{equation}
and and a minimizer $A$ satisfies \textup{(\ref{26-2-52})} and \textup{(\ref{26-2-53})};
\item Furthermore, let assumption \textup{(\ref{26-2-62})} be fulfilled (i.e.
$\upmu _0(\Pi_\infty)=0$). Then
\begin{equation}
\E^*_\psi- \int \Weyl^*_1(x)\psi^2(x)\,dx  =o(h^{2-d})\label{26-2-86}
\end{equation}
and a minimizer $A$ satisfies similarly improved versions of \textup{(\ref{26-2-52})} and \textup{(\ref{26-2-53})}.
\end{enumerate}
\end{theorem}

\section{Rescaling}
\label{sect-26-2-4}

We consider only $d=3$ here.

\subsection{Case \texorpdfstring{$\kappa \le 1$}{\textkappa  \textle 1}}
\label{sect-26-2-4-1}

We already have an upper estimate: see Corollary~\ref{cor-26-2-23}. Let us prove a lower estimate\footnote{\label{foot-26-10} But only for $\E_\psi(A)$ defined by (\ref{26-2-79}).}. Consider an error
\begin{equation}
\Bigl(\int \Weyl_1(x)\psi^2\,dx - \E_{\psi}(A)\Bigr)_+.
\label{26-2-87}
\end{equation}
Obviously $\Tr^-$ is sub-additive
\begin{equation}
\Tr^-(\sum _j \psi_j H\psi_j)\ge \sum_j \Tr^-(\psi_j H\psi_j)
\label{26-2-88}
\end{equation}
and therefore so is $\E_\psi (A)$  under assumption that
$\psi_j\in \sC^2_0(B(x_j,\frac{1}{2}\ell_j))$ where multiplicity of covering by $B(x_j,\ell_j))$ does not exceed $C_0$ and we are allowed to replace $\kappa$ by $C_1 \kappa$ in the right-hand expression\footnote{\label{foot-26-11} Really, $\|\partial A\|^2 \ge c\sum_j \|\partial A_j\|^2 $ with
$A_j(x)= (A(x)  -A_j(x_j))\psi'_j$ with
$\psi'_j\in \sC_0^2(B(x_j,\frac{7}{8}\ell_j))$ equal $1$ in $B(x_j,\frac{3}{4}\ell_j))$.}.

Then we need to consider each partition element and use a lower  estimate for it. While considering partition we use so called \emph{ISM identity}\index{ISM identity}: as
\begin{equation}
\sum _j\psi_j^2=1
\label{26-2-89}\\[-2pt]
\end{equation}
we have
\begin{multline}
H=
\sum_j \bigl(\psi_j H\psi_j + \frac{1}{2}[[H,\psi_j],\psi_j]\bigr)=\\
\sum_j \psi_j \bigl(H +
\frac{1}{2} \sum_k [[H,\psi_k],\psi_k] \bigr)\psi_j
\label{26-2-90}
\end{multline}
where the second equality is due to the fact that $[[H,\psi_j],\psi_j$ is an ordinary function.

In virtue of Proposition~\ref{prop-26-2-5}, from the very beginning we need to consider
\begin{equation}
M= \kappa^\beta h^{-1-\alpha}
\label{26-2-91}
\end{equation}
with $\alpha=2$, $\beta=0$ and $\kappa \le c$. But we need to satisfy precondition (\ref{26-2-21}) which is then
\begin{equation}
\kappa ^{\beta+1} h^{-\alpha}\le c.
\label{26-2-92}
\end{equation}
Therefore, if condition (\ref{26-2-92}) is fulfilled with $\alpha=0$ we conclude that the final error is indeed $O(h^{-1})$ or even $o(h^{-1})$ (under assumption (\ref{26-2-62})) without any precondition.

Let precondition (\ref{26-2-92}) fail. Let us use $\gamma$-admissible partition of unity $\psi_j$ with $\psi_j$ satisfying after rescaling assumptions of Proposition~\ref{prop-26-2-22}.

Note that rescaling $x\mapsto x\gamma^{-1}$ results in
$h\mapsto h_\new=h \gamma^{-1}$ and after rescaling in the new coordinates  $\|\partial  A\|^2$ acquires factor $\gamma^{d-2}$ and thus factor $\kappa^{-1}h^{-2}$ becomes $\kappa^{-1}h^{-2}\gamma^{d-2}  =
\kappa_\new^{-1} h_\new^{-2}$ with $\kappa \mapsto \kappa_\new=\kappa \gamma$.

Then \emph{after rescaling\/} precondition (\ref{26-2-92}) is satisfied provided \emph{before rescaling\/}
$\kappa^{\beta+1}h^{-\alpha}\gamma^{\alpha +\beta +1}\le c$. Let us pick up
$\gamma = \kappa^{-(\beta+1)/(\alpha +\beta +1) }h^{\alpha/(\alpha +\beta +1)}$.
Obviously if before rescaling condition (\ref{26-2-92}) fails, then
$h\ll \gamma\le 1$.

But then expression (\ref{26-2-87}) with $\psi$ replaced by $\psi_j$
does not exceed $Ch_\new^{-1}=C(h\gamma^{-1})^{-1}$ and the total expression (\ref{26-2-87})  does not exceed $C(h\gamma^{-1})^{-1}\gamma^{-3}=
Ch^{-1}\gamma^{-2}= C\kappa ^{\beta'} h^{-1-\alpha'}$ with
\begin{equation*}
\beta'= 2(\beta+1)/(\alpha+\beta+1),\qquad \alpha'=2\alpha/(\alpha+\beta+1);
\end{equation*}

So, actually we can pick up $M$ with $\alpha,\beta$ replaced by $\alpha',\beta'$ and we have a precondition (\ref{26-2-92}) with these new $\alpha',\beta'$ and we do not need an old precondition. Repeating the rescaling procedure again we derive a proper estimate with again weaker precondition etc.

One can see easily that $\alpha'+\beta'+1=3$ and therefore on each step $\alpha+\beta+1=3$ and we have recurrent relation for $\alpha'$: $\alpha'=\frac{2}{3}\alpha$;
and therefore we have sequence for $\alpha$ which decays and  becomes arbitrarily small. Therefore precondition (\ref{26-2-92}) has been reduced to $\kappa \le h^\delta$ and estimate $M=O(h^{-1})$ has been established. After this under assumption (\ref{26-2-62}) we can prove even sharper asymptotics.

To weaken assumption $\kappa \le h^\delta$ to $\kappa\le c$ we can use rescaling $x\mapsto x\gamma^{-1}$ with $\gamma =h^\delta$. We arrive to the error estimate $O(h^{-1-\delta})$ and therefore optimizer satisfies
$\|\nabla \times A\|\le h^{\frac{1}{2}-\delta}$ (where $\delta$ is increased if necessary but remains arbitrarily small). Then instead of
$\|\Delta A\|_{\sL^\infty} =O(1)$ we arrive to
$\|\Delta A\|_{\sL^\infty} = O(h^{-\delta})$ and to
$\|\partial^2 A\|_{\sL^\infty} =O(h^{-\delta})$; then
$\|\partial  A\|_{\sL^\infty} =O(h^{\frac{1}{2}-\delta})$; it is more than sufficient to unleash microlocal analysis without any need to appeal to Proposition~\ref{prop-26-2-6} which is the only place where we needed assumption~(\ref{26-2-21}).

Thus we arrive to

\begin{theorem}\label{thm-26-2-29}
Let $d=3$, $V\in \sC^{2,1}$, $\kappa \le c$ and let $\psi$ satisfy assumption of proposition~\ref{prop-26-2-22}. Then

\begin{enumerate}[label=(\roman*), fullwidth]
\item Asymptotics \textup{(\ref{26-2-85})} holds;
\item Further, if assumption \textup{(\ref{26-2-62})} is fulfilled then asymptotics \textup{(\ref{26-2-85})} holds;
\item  If \textup{(\ref{26-2-85})} or \textup{(\ref{26-2-86})} holds for
$E_\psi (A)$  (we need only an estimate from below) then
$\|\partial A\|=O((\kappa h)^{\frac{1}{2}})$ or
$\|\partial A\|=o((\kappa h)^{\frac{1}{2}})$ respectively.
\end{enumerate}
\end{theorem}

\subsection{Case \texorpdfstring{$1\le \kappa \le h^{-1}$}{1\textle \textkappa  \textle 1/h}}
\label{sect-26-2-4-2}

We can consider even the case $1\le \kappa \le h^{-1}$. The simple rescaling-and-partition arguments with $\gamma=\kappa^{-1}$  lead to the following
\begin{claim}\label{26-2-93}
As $1\le \kappa \le h^{-1}$ remainder estimate $O(\kappa^2h^{-1})$ holds and for a minimizer $A$ satisfies $\|\partial A\|^2\le C\kappa^3 h$.
\end{claim}
However we would like to improve it and, in particular to prove that if $\kappa$ is moderately large then the remainder estimate is still $O(h^{-1})$ and even $o(h^{-1})$ under non-periodicity assumption.

\begin{theorem}\label{thm-26-2-30}
Let $d=3$, $V\in \sC^{2,1}$,  and let $\psi$ satisfy assumptions of proposition~\ref{prop-26-2-22}. Then
\begin{enumerate}[label=(\roman*), fullwidth]
\item \label{thm-26-2-30-i}
As
\begin{equation}
\kappa \le \kappa^*_h\Def \epsilon h^{-\frac{1}{4}}|\log h|^{-\frac{3}{4}}
\label{26-2-94}
\end{equation}
asymptotics \textup{(\ref{26-2-85})} holds;
\item \label{thm-26-2-30-ii}
Furthermore as $\kappa=o(\kappa^*_h)$ and assumption \textup{(\ref{26-2-62})} is fulfilled then   asymptotics \textup{(\ref{26-2-86})} holds;
\item \label{thm-26-2-30-iii}
As   $\kappa^*_h \le \kappa \le ch^{-1}$
 the following estimate holds:
\begin{equation}
|\E^*_\psi - \int \Weyl_1(x)\psi^2(x) \,dx |\le
Ch^{-3} (\kappa h)^{\frac{8}{3}} |\log \kappa h|^2
\label{26-2-95}
\end{equation}
\end{enumerate}
\end{theorem}

\begin{proof}
(i) From (\ref{26-2-45}) we conclude as $\kappa \ge c $ that
\begin{equation*}
h^{1-\theta}|\partial A|_{\sC^\theta}\le C\kappa (\kappa +\bar{\mu}).
\end{equation*}
Then using arguments of subsection~\ref{sect-26-2-2-2} one can prove easily that for $\kappa \le h^{\sigma-\frac{1}{2}}$
\begin{equation*}
|F_{t\to h^{-1}\tau} \bar{\chi}_T(t) (hD_x)^\alpha (hD_x)^\beta
\bigl(U(x,y,t)-U_{(\varepsilon)} (x,y,t)- U'_{(\varepsilon)} (x,y,t)\bigr)|\le C h^{1-d}
\end{equation*}
where we use the same $2$-term approximation, $T=\epsilon \bar{\mu}^{-1}$. Let  us take then $x=y$, multiply by $\varepsilon^{-d}\psi ( \varepsilon ^{-1}(y-z))$ and integrate over $y$.  Using rough microlocal analysis one can prove easily that from both $U_{(\varepsilon)} (x,y,t)$ and $U'_{(\varepsilon)}(x,y,t)$ we get $O(h^{-2})$ and in the end of the day we arrive to the estimate
$|\Delta A_\varepsilon |\le C\kappa\bar{\mu}$ which implies
\begin{equation}
|\partial^2 A_\varepsilon|\le C\kappa \bar{\mu} |\log h| +C\mu
\label{26-2-96}
\end{equation}
where obviously one can skip the last term. Here we used property of the Laplace equation. For our purpose it is much better than
$|\partial^2 A_\varepsilon|\le C\kappa ^2 |\log h| +C\mu $ which one could derive easily.

Again using arguments of subsection~\ref{sect-26-2-2-2} one can prove easily that
\begin{gather}
|\Tr(\psi H^-_{A,V}\psi )- \Tr(\psi H^-_{A_\varepsilon,V}\psi )|\le
C\bar{\mu}^2 h^{2-d}\label{26-2-97}\\
\shortintertext{and therefore}
|\Tr(\psi H^-_{A,V}\psi )- \int \Weyl_1(x)\psi^2(x)\,dx|\le  C\bar{\mu}^2 h^{2-d}\label{26-2-98}\\
\intertext{and finally for an optimizer}
\|\partial A\|^2 \le C\kappa \bar{\mu}^2 h .\label{26-2-99}
\end{gather}
Here $\mu$ and $\bar{\mu}$ were calculated for $A$, but it does not really matter as due to  $|\partial^2 A|\le C\kappa ^2h^{-\delta}$ we conclude that
$|\partial A- \partial A_\varepsilon |\le C\kappa ^2h^{-\delta}\varepsilon \le C$ due to restriction to $\kappa$.

Then, as $d=3$
\begin{equation}
\mu^2 \bigl(\mu /(\kappa \bar{\mu}|\log h|\bigr))^3 \le \kappa \bar{\mu}^2 h
\label{26-2-100}
\end{equation}
and if $\mu \ge 1$ we have $\bar{\mu}=\mu$ and (\ref{26-2-100}) becomes
$\kappa^{-3}|\log h|^{-3} \le C \kappa h$ which impossible under (\ref{26-2-94}).

So, $\mu \le 1$ and (\ref{26-2-100}) implies (\ref{26-2-85}) and (\ref{26-2-99}), (\ref{26-2-100}) imply that for an optimizer
$\|\partial A \|\le C(\kappa h)^{\frac{1}{2}}$ and
$\mu \le C\kappa^4 h |\log h|^d$. So (i) is proven.

\bigskip\noindent
(ii) Proof of (ii) follows then in virtue of arguments of subsection~\ref{sect-26-2-2-2}.

\bigskip\noindent
(iii)  If $\kappa^*_h\le \kappa \le h^{-1}$ we apply partition-and-rescaling. So, $h\mapsto h'=h\gamma^{-1}$ and $\kappa \mapsto \kappa'= \kappa \gamma$ and to get into (\ref{26-2-94}) we need $\gamma = \epsilon \kappa^{-\frac{4}{3}}h^{-\frac{1}{3}}|\log (\kappa h)|^{-1}$ leading to the remainder estimate  $Ch^{-1}\gamma^{-2}$ which proves (ii).
\end{proof}

\begin{Problem}\label{Probem-26-2-31}
Repeat arguments of Subsubsections~\ref{sect-26-2-1-2}.2 and \ref{sect-26-2-1-3}.3 and of this Subsection as $d\ne 3$. When they hold?
\end{Problem}

\chapter{Global trace asymptotics in the case of Coulomb-like singularities}
\label{sect-26-3}

\section{Problem}
\label{sect-26-3-1}

We consider the same operator (\ref{26-1-4}) as before in $\bR^3$ but now we assume that $V$ has Coulomb-like singularities. Namely let $\y_m\in \bR^3$ ($m=1,\ldots,M$, where $M$ is fixed) be singularities (``nuclei''). We assume that
\begin{gather}
V=\sum_{1\le m\le M} \frac{z_m} {|x-\y_m|} + W(x)
\label{26-3-1}\\
\shortintertext{where}
z_m\ge 0, \ z_1+\ldots +z_M \asymp 1,
\label{26-3-2}
\end{gather}
and
\begin{multline}
|D^\alpha W|\le C_\alpha \sum_{1\le m\le M}  z_m\bigl(|x-\y_m|+1\bigr)^{-1}|x-\y_m|^{-|\alpha|}\\
\forall \alpha:|\alpha|\le 2
\label{26-3-3}
\end{multline}
but at the first stages we will use some weaker assumptions. Later we assume that $V(x)$ decays at infinity sufficiently fast. Let us define $\E^*)$ and $\E(A)$ by (\ref{26-2-2})--(\ref{26-2-1}). Finally, let
$\ell(x) \min_{1\le m\le M}\ell_m(x)$ \ with \ $\ell_m(x)\Def\frac{1}{2}|x-\y_m|$.

In this and next Sections we assume that
\begin{claim}\label{26-3-4}
$\kappa \in (0,\kappa^*]$ where $0<\kappa^*$ is a small constant.
\end{claim}
As $\kappa=0$ we set $A=0$ and consider $\E\Def \Tr^- (H_{A,V})$; then our results will not be new.

\section{Estimates of the minimizer}
\label{sect-26-3-2}

Let us consider a Hamiltonian with potential $V$ and let $A$ be a   minimizing expression (\ref{26-2-2}) magnetic field. We say that $A$ is a \emph{minimizer\/} and in the framework of our problems we will prove it existence.

\subsection{Preliminary analysis}
\label{sect-26-3-2-1}

First we start from the roughest possible estimate:
\begin{proposition}\label{prop-26-3-1}
Let $V$ satisfy \textup{(\ref{26-3-1})}--\textup{(\ref{26-3-3})} and let
$\kappa \le \kappa^*$. Then the near-minimizer $A$ satisfies
\begin{gather}
|\int \bigl(\tr e_{A,1}(x,x,0)-\Weyl_1(x)\bigr)\,dx|\le Ch^{-2}
\label{26-3-5}\\
\shortintertext{and}
\|\partial A\| \le C\kappa^{\frac{1}{2}}.
\label{26-3-6}
\end{gather}
\end{proposition}

\begin{proof}
Definitely (\ref{26-3-5})--(\ref{26-3-6}) follow from the results of L.~Erd\"os,  S.~Fournais,  and J. P.  Solovej \cite{EFS3} but we give an independent easier proof based on our Subsection~\ref{26-2-1}.

\bigskip\noindent
(i) First, let us pick up  $A=0$ and consider
$\Tr \bigl(\psi_\ell \uptheta(-H)\psi_\ell\bigl)$ with  cut-offs
$\psi_\ell (x)=\psi((x-\y_m)/\ell)$ where $\psi\in \sC_0^\infty (B(0,1))$ and equals $1$ in $B(0,\frac{1}{2})$. Here and below
$\uptheta (\tau -H_{A,V})$ is a spectral projector of $H$.

Then
\begin{equation}
|\Tr \bigl(\psi_\ell  H_{A,V} ^-(0)\psi_\ell\bigr)|\le Ch^{-2}\qquad
\text{as\ \ } \ell=\ell_*\Def h^2.
\label{26-3-7}
\end{equation}
On the other hand, contribution of $B(x,\ell )$ with $\ell(x)\ge \ell_*$ to the Weyl error does not exceed  $C\zeta ^2\hbar^{-1}=C\zeta^3 \ell h^{-1}$ where $\hbar = h/(\zeta \ell)$ in the rescaling; so  after summation over $\ell\ge \ell_*$ we also get $O(h^{-2})$ provided $\zeta^2\le C\ell^{-1}$. Therefore we arrive to the following rather easy inequality:
\begin{equation}
|\int \bigl(\tr e_{0,1}(x,x,0)-\Weyl_1(x)\bigr)\,dx |\le Ch^{-2} .
\label{26-3-8}
\end{equation}
This is what rescaling method gives us without careful study of the singularity.

\bigskip\noindent
(ii) On the other hand, consider $A\ne 0$. Let us prove first that
\begin{equation}
\Tr^- (\psi_\ell H \psi_\ell) \ge -Ch^{-2} -Ch^{-2}\int |\partial A|^2\,dx
\qquad\text{as\ \ } \ell=\ell_*.
\label{26-3-9}
\end{equation}
Rescaling $x\mapsto (x-\y_m) /\ell$ and  $\tau \mapsto \tau/\ell$ and therefore $h\mapsto h \ell^{-\frac{1}{2}}\asymp 1$ and $A\mapsto A\ell^{\frac{1}{2}}$  (because singularity is Coulomb-like),  we arrive to the same  problem with the same $\kappa$ (in contrast to Subsection~\ref{sect-26-2-4}  where
$\kappa\mapsto \kappa \ell$ because of different scale in $\tau$ and $h$) and with $\ell=h=1$.

However this estimate follows from L.~Erd\"os,  J. P. Solovej~\cite{erdos:solovej} (we reproduce Lemma 2.1 of this paper in Appendix~\ref{sect-26-A-1}.

\bigskip\noindent
(iii) Consider now $\psi_\ell$  as in (i) with $\ell\ge \ell_*$. Then according to theorem~\ref{thm-26-2-29}  rescaled
\begin{multline}
\Tr^- \bigl(\psi_\ell H_{A,V} \psi_\ell\bigr) -
\int \Weyl_1 (x)\psi_\ell^2(x) \,dx \\
\ge -C \zeta^3 \ell h^{-1} -
Ch^{-2} \int _{B(x, 2\ell/3)} |\partial A|^2\,dx.
\label{26-3-10}
\end{multline}
Really, rescaling of the first part is a standard one and in the second part we should have  in the front of the integral a coefficient
$\kappa^{-1}h^{-2}\zeta^2 \times \zeta^{-2} \ell  (h/\zeta \ell)^{-2}$ where factor $\zeta^2$ comes from the scaling of the spectral parameter, factor $\zeta^{-2}$ comes from the scaling of the magnitude of $A$, factor $\ell=\ell^3 \times \ell^{-2}$ comes from the scaling of $dx$ and $\partial$ respectively, and $\hbar\Def h/(\zeta \ell)$ is a semiclassical parameter after rescaling. Therefore this expression  acquires a factor $\zeta^2 \ell \le C$.

Then
\begin{equation}
\int \bigl(\tr e_{A,1}(x,x,0)-\Weyl_1(x)\bigr)\,dx \ge
-Ch^{-2}- Ch^{-2} \int |\partial A|^2\,dx
\label{26-3-11}
\end{equation}
and adding the magnetic field energy $\kappa^{-1}h^{-2}\|\partial A\|^2$ we find out that the left-hand expression of (\ref{26-3-5}) is greater than the same expression with $A=0$ plus
$(\kappa^{-1}-C)h^{-2}\|\partial A\|^2$ minus $Ch^{-2}$ which implies (\ref{26-3-5}) and (\ref{26-3-6}) as $A$ is supposed to be a near-minimizer.
\end{proof}

\begin{proposition}\label{prop-26-3-2}
Let $V$ satisfy \textup{(\ref{26-3-1})}--\textup{(\ref{26-3-3})}.
Then there exists a minimizer $A$.
\end{proposition}

\begin{proof}
After Proposition~\ref{prop-26-3-1} has been proven we just repeat arguments of the proof of Proposition~\ref{prop-26-2-2}. If $V\in \sL^{\frac{5}{2}}$ no change would be required but for $V\notin \sL^{\frac{5}{2}}$ one needs to consider modifications as in Remark~\ref{rem-26-3-3}(i) below.
\end{proof}

\begin{remark}\label{rem-26-3-3}
We are a bit ambivalent about convergence of $\int \Weyl_1(x)\, dx$ at infinity, as for Coulomb potential it diverges. In this case however we can \underline{either} assume in addition that $V\in \sL^{\frac{5}{2}}$, \underline{or} tackle it as in Proposition~\ref{prop-26-3-16} below.
\end{remark}

\subsection{Estimates to a minimizer. I}
\label{sect-26-3-2-2}

Let us repeat arguments of Subsection~\ref{sect-26-2-1-3}.3. However our task now is much more complicated: while we know a priory that
$\|\partial A\|^2\le C\kappa$ we will not be able to improve it significantly (or at all as $\kappa \asymp 1$).

Recall equation (\ref{26-2-14}) for a minimizer $A$. After rescaling
$x\mapsto x/\ell$, $\tau\mapsto \tau/\zeta^2$,
$h\mapsto \hbar=h/(\zeta\ell)$, $A\mapsto A\zeta^{-1}\ell$ this equation becomes
\begin{multline}
\Delta A_j=\\
-2\kappa \zeta^2\ell  \hbar^2 \Re \tr \upsigma_j \Bigl(
( \hbar Dk-\zeta^{-1} A)_x \cdot\boldupsigma e (x,y,\tau) +
e (x,y,\tau) \,^t(\hbar D-\zeta^{-1}A)_y\cdot\boldupsigma
  \Bigr)\Bigr|_{y=x}
\label{26-3-12}
\end{multline}
and since so far $\zeta^2\ell=1$ we arrive to
\begin{multline}
\Delta A_j=\\
-2\kappa  \hbar^2
\Re \tr \upsigma_j\Bigl(
 ( \hbar D-\zeta^{-1} A)_x \cdot \boldupsigma   e (x,y,\tau)  +
e (x,y,\tau)\,^t ( \hbar D-\zeta^{-1} A)_y \cdot \boldupsigma    \Bigr)\Bigr|_{y=x}.
\label{26-3-13}
\end{multline}

\smallskip\noindent
(i) Plugging for $u=\psi \uptheta (-H)f$ with cut-off function $\psi$ and repeating arguments of Subsubsection~\ref{sect-26-2-1-3}.3 we conclude that in the rescaled coordinates
\begin{multline}
\|(\hbar D_x \cdot \boldupsigma )u \|\le
\|((\hbar D_x -A)\cdot\boldupsigma)u\|+
C\|A\|_{\sL^6}\cdot  \|u\|_{\sL^3}  \\[2pt]
\shoveright{\le \|((\hbar D_x -A)\cdot\boldupsigma)u\|+
C\hbar ^{-\frac{1}{2}}\|A\|_{\sL^6} \cdot \|u\|^{\frac{1}{2}}
\cdot  \|\hbar D_x u\|^{\frac{1}{2}}\;\ }\\[2pt]
\le
\|((\hbar D_x -A)\cdot\boldupsigma)u\|+ \frac{1}{2} \|\hbar  D_x u\|
+ C (\hbar^{-\frac{1}{2}}\|A\|_{\sL^6})^2  \|u\|
\label{26-3-14}
\end{multline}
where $\|A\|_{\sL^6}$ calculated in the rescaled coordinates is equal to $\|A_\orig\|_{\sL^6,\orig}$ (where subscripts ``$\mathsf{orig}$''means that the norm is calculated in the original coordinates and $A$) which does not exceed
$C\kappa^{\frac{1}{2}}$ due to (\ref{26-3-6})\footnote{\label{foot-26-12} As usual we assume that the average of $A$ over $B(x,1 )$ is $0$.}
and therefore (since $\|(\hbar D_x \cdot \boldupsigma )u \| =\|\hbar D_x u \| $)
\begin{equation}
\|\hbar D_x u \| \le C\bigl( 1+\kappa\hbar^{-1}\bigr)\|f\|.
\label{26-3-15}
\end{equation}
Continuing arguments of Subsubsection~\ref{sect-26-2-1-3}.3 we conclude that in the rescaled coordinates
\begin{gather}
\|(\hbar D_x)^k u \| \le C\bigl( 1+\kappa\hbar^{-1}\bigr)^k\|f\|,
\label{26-3-16}\\[3pt]
\|(\hbar D_x)^k ((\hbar D_x -A)\cdot \boldupsigma ) u \| \le
C( 1+\kappa\hbar^{-1})^k \|f\|,
\label{26-3-17}
\end{gather}
for $k=0,1,2$ and therefore
\begin{equation}
\|\Delta A\|_{\sL^\infty(B(x,1))} \le
C\kappa \hbar^{-1} ( 1+\kappa\hbar^{-1})^3 ,
\label{26-3-18}
\end{equation}
Here we estimate different norms of $A$ locally. Then \underline{either}
\begin{equation}
\|\partial    A\|_{\sL^\infty(B(x,\frac{3}{4}))} +
\hbar^{\delta}\|\partial^2    A\|_{\sL^\infty(B(x,\frac{3}{4}))} \le
C \kappa \hbar^{-1}( 1+\kappa\hbar^{-1})^3  \label{26-3-19}
\end{equation}
\underline{or}
\begin{multline}
\|\partial    A\|_{\sL^\infty(B(x,1-\epsilon))}  +
\hbar^\delta \|\partial^2    A\|_{\sL^\infty(B(x,1-\epsilon))}\\
\le
C \|\partial A\|=
C\|\partial A_\orig\|_{\orig}\le
C\kappa^{\frac{1}{2}}
\label{26-3-20}
\end{multline}
In the latter case (\ref{26-3-20}) we have in the original coordinates
\begin{equation}
\|\partial  A\|_{\sL^\infty(B(x,\ell ))}\le
C\kappa^{\frac{1}{2}} \ell^{-\frac{3}{2}}
\label{26-3-21}
\end{equation}
and we are rather happy because then the effective intensity of the magnetic field in $B(x,\ell)$ is
$\zeta^{-1}\ell \|\partial  A\|_{\sL^\infty(B(x,1-\epsilon ))} \le C\kappa^{\frac{1}{2}}$.

\medskip\noindent
(ii) The former case (\ref{26-3-19}) is much more complicated because our estimate is really poor as $\kappa \asymp 1$ and we are going to act only in this assumption. Assume that
\begin{equation}
\| \partial  A\| _{\sL^\infty(B(x,1-\epsilon ))}\le \mu
\label{26-3-22}
\end{equation}
with $\mu \ge \hbar^{-\sigma} $. Selecting $u=\psi \uptheta(-H)f$ with $\gamma$-admissible $\psi$  we conclude that
$\| (A\cdot \boldupsigma) u\|\le \|A\|_{\sL^\infty}\|u\|\le C\mu  \gamma\|u\|$ (assuming without any loss of the generality that $A=0$ at some point of $\supp \psi$) and that
\begin{gather*}
\|(\hbar D)^k u\|\le C(1+\hbar\gamma^{-1}+\mu  \gamma)^k,\\[2pt]
\|(\hbar D)^k ((\hbar D -A)\cdot \boldupsigma ) u\|\le C(1+\hbar\gamma^{-1}+\mu  \gamma)^{k+1}\\
\intertext{and therefore}
|\Gamma_x  (\hbar D_x -A)\cdot \boldupsigma )e(.,.,0)|\le
C\hbar^{-3}(1+\hbar\gamma^{-1}+\mu  \gamma)^{\frac{7}{2}}\\
\shortintertext{and then}
\|\Delta A\|_{\sL^\infty(B(x,1-\epsilon ))} \le
C\hbar^{-1}(1+\hbar\gamma^{-1}+\mu  \gamma)^{\frac{7}{2}}.
\end{gather*}
Optimizing with respect to $\gamma = \mu ^{-\frac{1}{2}} h^{\frac{1}{2}}$ we conclude that either
\begin{equation*}
\|\partial^2  A\|_{\sL^\infty(B(x,1-\epsilon ))} \le
C\hbar^{-1-\delta}(1+\hbar\mu  )^{\frac{7}{4}}
\end{equation*}
or (\ref{26-2-21}) holds. In the former case using the second of estimates
\begin{align}
&\| A\| _{\sL^\infty(B(x,1-\epsilon ))}\le
C \| \partial^2 A\| _{\sL^\infty(B(x,1-\epsilon ))} ^{\frac{1}{5}}
\| \partial A\|  ^{\frac{4}{5}},
\label{26-3-23}\\
&\| \partial A\| _{\sL^\infty(B(x,1-\epsilon))}\le
C\| \partial^2 A\| _{\sL^\infty(B(x,1-\epsilon ))} ^{\frac{3}{5}}
\| \partial A\|  ^{\frac{2}{5}}
\label{26-3-24}
\end{align}
we conclude that (\ref{26-3-22}) with
\begin{equation*}
\mu  \Def \hbar^{-\frac{3}{5}-\delta} (1+\hbar\mu  )^{\frac{21}{20}}
\end{equation*}
and one can see easily that starting from $\mu  =\hbar^{-4}$ as given by (\ref{26-3-19}) we can arrive after number of iterations to
$\mu  = \hbar^{-\frac{3}{5}-\delta}$ and therefore
\begin{equation}
\| \partial^k  A\| _{\sL^\infty(B(x,1-\epsilon ))}\le C\hbar^{-\frac{1}{5}(1+2k)-\delta} \qquad k=0,1,2.
\label{26-3-25}
\end{equation}

\smallskip\noindent
(iii) This estimate is good enough to launch microlocal arguments. Assuming (\ref{26-3-22}) with $\mu \le h^{\sigma -1}$ we estimate as in Section~\ref{sect-26-2}
\begin{gather*}
|\Gamma _x ((\hbar D_x-A)\cdot \boldupsigma)e (.,.,0)|\le C\mu \hbar^{-\delta}\\
\shortintertext{and then}
\|\partial^2 A\|_{B(x,1-\epsilon)} \le C\mu h^{-\delta}\\
\shortintertext{and therefore}
\|\partial A\|_{B(x,1-\epsilon)} \le C\mu ^{\frac{3}{5}} h^{-\delta}\\
\end{gather*}
resulting in
$\mu\Def \mu ^{\frac{3}{5}} h^{-\delta}$ and after a number of iterations we get $\mu =h^{-\delta}$ and therefore iterating this procedure one more time and taking into account factor $\kappa$ we arrive to
\begin{claim}\label{26-3-26}
Either (\ref{26-3-20}) holds or
\begin{equation}
\| \partial^2 A\| _{\sL^\infty(B(x,1-\epsilon ))}\le C\kappa h^{-\delta}.
\label{26-3-27}
\end{equation}
\end{claim}
However to  prove that  the effective magnetic field  $O(1)$  we need to modify these arguments, and we do it in the next Subsubsection.

\subsection{Estimates to a minimizer. II}
\label{sect-26-3-2-3}

In this step we repeat arguments of Subsubsection~\ref{sect-26-2-2-1}.1  but now we have a problem: we cannot use $\mu =\|\partial A\|_\infty$ as we have  domains $\cX_r= \{x: \ell(x)\ge r\}$ rather than the whole space. So we get the following analogue of (\ref{26-2-47}) where $A$ is still rescaled and the norms are calculated  in the rescaled coordinates:
\begin{multline}
\|\Delta A\|_{\sC (B(x, \frac{3}{4})} +
\hbar \|\Delta \partial A\|_{\sC (B(x, \frac{3}{4})} \le \\
C\kappa \Bigl( 1 +  |\partial A|_{\sC (B(x, 1)} +
 h^{\frac{1}{2}(\theta-1)}
\|\partial  A\|_{\sC^{\theta} \frac{1}{2}(B(x, 1)}\Bigr)
\label{26-3-28}
\end{multline}
which implies
\begin{multline}
\|\partial A\|_{\sC (B(x, \frac{1}{2}))}+
\hbar^{\theta-1}  \|\partial A\|_{\sC^\theta (B,(x,\frac{1}{2}))} \le\\[3pt]
\epsilon  \hbar^{(\theta-1)}\zeta^{-1}
 \|\partial A\|_{\sC^\theta (B(x,1))} +
 C\kappa \|\partial A\|_{\sC (B(x, 1))} +
 C \|\partial A\|_{\sL^2 (B(x, 1))}
 \label{26-3-29}
\end{multline}
and the last term in the right-hand expression does not exceed $C\kappa^{\frac{1}{2}}$.

Let $\nu (r)=\sup_{x:\,\ell(x)\ge r} f(x)$ where $f(x)$ is the left-hand expression of (\ref{26-3-15}) calculated for given $x$ in the rescaled coordinates. Then (\ref{26-3-29}) implies that for $\kappa \in (0,\kappa^*)$ (where $\kappa^*>0$ is a small constant)
\begin{gather*}
\nu (r)\le \frac{1}{2}\nu (\frac{1}{2} r) +C\kappa ^{\frac{1}{2}}\\
\shortintertext{which in turn implies that}
\nu (r) \le \frac{1}{2}\nu (2^{-n}r) + 2C\kappa^{\frac{1}{2}},\qquad n\ge 1,\\
\shortintertext{and therefore}
\nu(r)\le 4C\kappa^{\frac{1}{2}}+ 4 \sup_{C_0 h^2\le \ell(x)\le 2C_0h^2}  f(x) \le C_1\kappa^{\frac{1}{2}}
\end{gather*}
due to the rough estimate (because $\hbar\asymp 1$  as $\ell(x)\asymp h^2$). Then returning to the original (not rescaled) coordinates and to the original (not rescaled) potential $A$ we arrive to estimates (\ref{26-3-30}) and (\ref{26-3-31}) below:

\begin{proposition}\label{prop-26-3-4}
Let $\kappa\le \kappa^*$, $\zeta=c\ell^{-\frac{1}{2}}$. Let $A$ be a minimizer. Then   for $\ell(x)\ge \ell_*=h^2$ estimate \textup{(\ref{26-3-21})} holds and also
\begin{gather}
|\partial^2 A (x)-\partial^2 A(y)|\le C\kappa^{\frac{1}{2}}\ell^{-\frac{5}{2}}|x-y|^{\theta} \ell^{\theta/2}\ell_*^{-\theta/2}\qquad 0<\theta <1,
\label{26-3-30}\\
\shortintertext{and}
|\partial A (x)-\partial A(y)|\le C\kappa^{\frac{1}{2}}\ell^{-\frac{5}{2}}|x-y| (1+|\log |x-y||)  .
\label{26-3-31}
\end{gather}
\end{proposition}

\begin{remark}\label{rem-26-3-5}
\begin{enumerate}[label=(\roman*), fullwidth]
\item\label{rem-26-3-5-i}
So far we used only assumption that
\begin{equation}
|\partial ^\alpha V|\le C\zeta^2 \ell^{-|\alpha|}\qquad \forall
\alpha :|\alpha|\le 2
\label{26-3-32}
\end{equation}
with $\zeta =\ell^{-\frac{1}{2}}$ but even this was excessive;

\item\label{rem-26-3-5-ii}
In this framework however we cannot prove better estimates as (\ref{26-3-21}) always remains a valid alternative even if $\zeta  \ll \ell^{-\frac{1}{2}}$;

\item\label{rem-26-3-5-iii}
Originally we need an assumption (\ref{26-2-36})  $|V|\ge \epsilon_0$, but for $d= 3$  one can easily get rid off it  by the standard rescaling technique.
\end{enumerate}
\end{remark}

Consider now zone $\{x:\ \ell(x) \le \ell_*\}$:

\begin{proposition}\label{prop-26-3-6}
Let $\kappa\le \kappa^*$, $\zeta\le c\ell^{-\frac{1}{2}}$. Let $A$ be a minimizer.   Then $|\partial A|\le C \kappa^{\frac{1}{2}} h^{-3}$ as
$\ell(x) \le \ell_*=h^2$.
\end{proposition}

\begin{proof}
Proof is standard, based on the rescaling (then $\hbar=1$) and equation (\ref{26-2-14}) for a minimizer $A$. We leave easy details to the reader.
\end{proof}

Let us slightly improve estimate to a minimizer $A$. We already know that
$|\partial A (x)|\le C_0\beta $ with $\beta =\ell^{-\frac{3}{2}}$ and using the  standard rescaling technique we conclude that
\begin{equation}
|\Delta A|\le C\kappa \zeta^2 \beta + C\kappa \zeta^3 \ell^{-1}
\label{26-3-33}
\end{equation}
which does not exceed $C\kappa \ell^{-\frac{5}{2}}$ which implies

\begin{proposition}\label{prop-26-3-7}
In our framework
\begin{enumerate}[label=(\roman*), fullwidth]
\item\label{prop-26-3-7-i}
As $\ell (x)\ge h^2$
\begin{gather}
|A|\le C\kappa \ell^{-\frac{1}{2}},\qquad
|\partial A|\le C\kappa \ell ^{-\frac{3}{2}}\label{26-3-34}\\
\shortintertext{and}
|\partial A (x)-\partial A(y)|\le C_\theta\kappa \ell^{-\frac{3}{2}-\theta}|x-y|^\theta \qquad \text{as\ \ } |x-y|\le \frac{1}{2}\ell (x)
\label{26-3-35}
\end{gather}
for any $\theta \in (0,1)$;
\item\label{prop-26-3-7-ii}
As $\ell (x)\le \ell_*= h^2$  these estimates hold with $\ell (x)$ replaced by $\ell^*$.
\end{enumerate}
\end{proposition}

\begin{remark}\label{rem-26-3-8}
\begin{enumerate}[label=(\roman*), fullwidth]
\item\label{rem-26-3-8-i}
Here in comparison with old estimates we replaced factor $\kappa^{\frac{1}{2}}$ by $\kappa$ which is an advantage;

\item\label{rem-26-3-8-ii}
These estimates imply that $\int _{\{x:\, \ell(x)\le 1\}} |\partial A|^2\,dx \le C \kappa^2 |\log h|$ while in fact it must not exceed $C\kappa^2$.
\end{enumerate}
\end{remark}

\subsection{Estimates to a minimizer. III}
\label{sect-26-3-2-4}

Consider now external zone $\cY\Def \{x:\,\ell(x) \ge 1\}$ and assume that
\begin{equation}
\zeta (x)\le C\ell(x)^{-\nu}\qquad \text{as\ \ } \ell(x) \ge 1
\label{26-3-36}
\end{equation}
with $\nu >1$.

Then if also
$|\partial A (x)| =O(\ell(x)^{-\nu_1})$ as $\ell(x)\ge 1$ then the right hand expression of (\ref{26-3-33}) does not exceed
$C\kappa (\ell^{-3\nu-1}+\ell^{-\nu_1-2\nu})$ and therefore we \emph{almost\/} upgrade estimate to $|\partial A (x)|$ to $O(\ell^{-3\nu}+\ell^{-\nu_1-2\nu+1})$ and repeating these arguments sufficiently many times to $O(\ell^{-3\nu})$.

However, there are obstacles to this conclusion: first, as $\nu >1$ we conclude that
\begin{equation*}
A_j = \sum_m \upalpha_{j,m} |x-\y_m|^{-1} + O(\ell^{-1-\delta})
\end{equation*}
with constant $\upalpha_{j,m}$; however assumption $\nabla\cdot A=0$ implies $\upalpha_{j,m}=0$ and we pass this obstacle.

Indeed, let our equation be $\Delta A_j=\Phi_j$ and therefore
\begin{equation*}
A_j(x)=-\frac{1}{4\pi}\int |x-y|^{-1} \Phi_j(y)\,dy.
\end{equation*}
Let $a$ be the minimal distance between nuclei, $1=\sum_{0\le m\le M} \phi_m$ where $\phi_m$ is supported in $\frac{1}{3}a$-vicinity of $\y_m$ and equals $1$ in $\frac{1}{4}a$-vicinity of $\y_m$, $m=1,\ldots, 1$. Let
\begin{equation*}
I_{j,m}=\int \Phi_j(y)\phi_m(y)\,dy,\qquad \eta =\max_{1\le m \le M}|I_{j,m}|.
\end{equation*}
Then as $x$ belomgs to $b$-vicinity of $\y_m$ with $b\le \epsilon a$ one can prove easily that
\begin{equation*}
|\partial_{x_k}\int |x-y|^{-1} \Phi_j(y)\phi_{m'}(y)\,dy|\le
C\eta a^{-2}+Ca^{-3}
\end{equation*}
 as $m'=0,1,\ldots, M$, $m'\ne m$.

Also one can prove easily that
\begin{equation*}
|\partial_{x_j} \Bigl(\int|x-y|^{-1} \Phi_j(y)\phi_{m'}(y)\,dy- |x-\y_m|^{-1}I_{j,m}\Bigr)|\le C|x-\y_m|^{-3}
\end{equation*}
and combining with the previous inequality and $\nabla \cdot A=0$ we conclude that $|I_{m,j}| \le C\eta a^{-2}b^2+Ca^{-3}b^2 +C b^{-1}$ as $b\le \epsilon a$. Then selecting $b=\epsilon_1a$ with sufficiently small constant $\epsilon_1$ we conclude that $\eta \le Ca^{-1}$ which in turn implies that
$|\partial_k A_j(x)|\le C\ell^{-3}$.

\medskip
The second obstacle
\begin{equation*}
A_j = \sum_{k,m} \upalpha_{jk,m} (x_k-\y_{k,m})|x-\y_m|^{-3} +
O(\ell^{-2})
\end{equation*}
with constant $\upalpha_{jk,m}$ we cannot pass as assumption $\nabla\cdot A=0$ implies only that modulo gradient
$A=\sum_m \upbeta_m\times \nabla \ell_m^{-1}$ with constant vectors $\upbeta_m$ and one cannot pass this obstacle.

Therefore we upgrade (\ref{26-3-34})--(\ref{26-3-35}) there:

\begin{proposition}\label{prop-26-3-9}
In our framework assume additionally that  \textup{(\ref{26-3-36})} holds. Then as $\nu>\frac{4}{3}$
\begin{gather}
|A|\le C\kappa \ell^{-2},\qquad
|\partial A|\le C\kappa \ell ^{-3}\label{26-3-37}\\
\shortintertext{and}
|\partial A (x)-\partial A(y)|\le C_\theta\kappa \ell^{-3-\theta}|x-y|^\theta \qquad \text{as\ \ } |x-y|\le \frac{1}{2}\ell (x)
\label{26-3-38}
\end{gather}
as $\ell(x)\ge 1$ (for all $\theta \in (0,1)$);
\end{proposition}

\begin{remark}\label{rem-26-3-10}
\begin{enumerate}[label=(\roman*), fullwidth]
\item\label{rem-26-3-10-i}
In application to the ground state energy we are interested in $\nu=2$;

\item\label{rem-26-3-10-ii}
Observe that for $a\ge 1$
\begin{equation}
\int_{\{\ell(x)\asymp a\}} |\partial A|^2\,dx = O(\kappa^2 a^{-3});
\label{26-3-39}
\end{equation}
\item\label{rem-26-3-10-iii}
We were not able to improve (\ref{26-3-37})--(\ref{26-3-39}) no matter how fast $\zeta$ decays.
\end{enumerate}
\end{remark}

\section{Basic trace estimates}
\label{sect-26-3-3}

Recall that the standard Tauberian theory results in the remainder estimate $O(h^{-2})$. Really, as the effective magnetic field intensity is no more than $C\kappa $, contribution of $B(x,\ell(x))$ to the Tauberian error\footnote{\label{foot-26-13} Or Weyl error as we will explain transition from Tauberian to Weyl estimates below.} does not exceed
$C\zeta^2 \times \hbar^{-1}= C\zeta^3 \ell h^{-1}$ which as
$\zeta \asymp \ell^{-\frac{1}{2}}$ translates into $C\ell^{-\frac{1}{2}}h^{-1}$ and summation over $\{x: \ell(x)\ge \ell_*=h^2\}$ results in $Ch^{-2}$. On the other hand, contribution of $\{x: \ell(x)\le \ell_*=h^2\}$ into asymptotics does not exceed $C\hbar^{-3}\ell_*^{-1}= Ch^{-2}$ as $\hbar=1$.

However now we can unleash  arguments of V.~Ivrii, V. and I.~M.~Sigal~\cite{ivrii:ground}.  Recall that we are looking at
\begin{equation}
\Tr (\psi H^-_{A,V}\psi)=  \Tr (\phi_1 H^-_{A,V}\phi_1) +
\Tr (\phi_2 H^-_{A,V}\phi_2)
\label{26-3-40}
\end{equation}
where $\psi^2=\phi_1^2+\phi_2^2$, $\supp \phi_1\subset \{x, |x|\le 2r\}$,
$\supp \phi_2\subset \{x, r\le |x|\le b\}$  and we compare it with the same expression calculated for $H_{A,V^0}$ with $V^0=Z_m|x|^{-1}$. Here we assume that
\begin{gather}
a \le 1, \qquad z\asymp 1
\label{26-3-41}\\
\shortintertext{and}
|D^\alpha (V-V^0)|\le c_0 a^{-1} \ell^{-|\alpha|}\qquad \forall \alpha:|\alpha|\le 3.
\label{26-3-42}
\end{gather}
The latter assumption is too restrictive and could be weaken. Then as $\phi(x)$ is an $\ell$-admissible partition element
\begin{gather}
\Tr  \bigl(\uptheta (-H^-_{A,V})\phi^2\bigr)=
\int \Weyl (x)\phi^2(x) \,dx + O(r h^{-2})
\label{26-3-43}\\
\shortintertext{and}
\Tr  \bigl( H^-_{A,V} \phi^2\bigr)=
\int \Weyl_1 (x)\phi^2(x) \,dx + O(r^{-\frac{1}{2}} h^{-1})
\label{26-3-44}
\end{gather}
where the error estimates are $O(\hbar^{-2})$ and $O(\zeta^2\hbar^{-1})$ respectively. One can justify transition from the Tauberian to Weyl errors by considering Tauberian expressions and considering $H_{A_\varepsilon,V}$ and $H_{A,V}$ as unperturbed and perturbed operators respectively; their difference is $O(\zeta^3 \varepsilon^2)$ with $\varepsilon =\hbar^{1-\delta}$.

Then the contribution\footnote{\label{foot-26-14} After standard rescaling $x\mapsto x\ell^{-1}$,
$\xi\mapsto \xi \zeta^{-1}$, $h\mapsto\hbar$, $\tau\mapsto \tau\zeta^{-2}$, and $t\mapsto t \zeta \ell^{-1}$.} of time interval  $\{t:\,t \asymp T\}$ to Tauberian expression for (\ref{26-3-43}) of the first term in the approximation does not exceed $C\hbar^{-4} T \times  (\hbar T^{-1})^s$, of the second term
$C\hbar^{-4} T \times  (\hbar T^{-1})^s T \hbar^{-1} \varepsilon^2$, and of the third term
$C\hbar^{-4} T \times  (\hbar T^{-1}) T^2\hbar^{-2}\varepsilon^4$. In the end the first tem gives us Weyl expression, the second term turns out to be $0$, and the third term is less than the announced error.

Similarly, the contribution\footref{foot-26-14} of time interval
$\{t:\,t \asymp T\}$ to Tauberian expression for (\ref{26-3-44}) of the first term in the approximation does not exceed
$C\zeta^2\hbar^{-4} T \times  (\hbar T^{-1})^s$,
of the second term
$C\zeta^2\hbar^{-4} T \times  (\hbar T^{-1})^s T \hbar^{-1} \varepsilon^2$ and of the third term
$C\hbar^{-4} T \times  (\hbar T^{-1})^2 T^2\hbar^{-2}\varepsilon^4$. Again, in the end the first tem gives us Weyl expression, the second term turns out to be $0$, and the third term is less than the announced error.

The same estimates also hold for operator $H_{A,V^0}$ and then using $\ell$-admissible partition of unity we conclude that
\begin{multline}
\Tr \bigl(\phi_2 (H^-_{A,V} - H^-_{A,V^0})\phi_2\bigr)= \\[3pt]
\int \bigl(\Weyl_1(x) -\Weyl_1^0(x)\bigr)\,\phi^2_2(x) \,dx + O(r^{-\frac{1}{2}}h^{-1})
\label{26-3-45}
\end{multline}
where $\Weyl^0_1$ and $\Weyl^0$ are calculated for operator with potential $V^0$. Indeed, we just proved this for each operator $H_{A,V}$ and $H_{A,V^0}$ separately.

On the other hand, considering $V^\eta= V^0 (1-\eta) +V \eta = V^0 + W\eta$ and following  V.~Ivrii, V. and I.~M.~Sigal~\cite{ivrii:ground} we can rewrite the similar expression albeit for $\phi_2=1$ as
\begin{equation}
\Tr \int_0^1 W\uptheta (-H_{A,V^\eta})\, d\eta
\label{26-3-46}
\end{equation}
and applying the semiclassical approximation (under temporary assumption that $W$ is supported in $\{x:\ |x|\le 4 r\}$) one can prove that for $\phi_1=1$
\begin{multline}
\Tr \bigl(\phi_1 (H^-_{A,V} - H^-_{A,V^0})\phi_1\bigr)= \\[3pt]
\int \bigl(\Weyl_1(x) -\Weyl_1^0(x)\bigr)\,\phi^2_1(x) \,dx + O(a^{-1}rh^{-2}).
\label{26-3-47}
\end{multline}
Really, due to (\ref{26-3-43}) and (\ref{26-3-42}) the contribution of ball $B(x,\ell(x))$ does not exceed $Ca ^{-1} \hbar^{-2}= Ca^{-1} \ell (x)h^{-2}$ and summation with respect to partition as $\ell (x)\le 4r$ returns $Ca^{-1}rh^{-2})$; meanwhile contribution of $\{x: \ell(x)\le \ell_*\}$ does not exceed $Ca^{-1}\hbar^{-2}=Ca^{-1}$ as there $\hbar=1$.

One can get easily rid off the temporary assumption and take $\phi_1$ supported in $\{x:\ell(x)\le 2r\}$ instead.

Therefore we arrive to

\begin{proposition}\label{prop-26-3-11}
Under assumption \textup{(\ref{26-3-42})}
\begin{multline}
\Tr \bigl(\psi (H^-_{A,V} - H^-_{A,V^0})\psi\bigr)= \\[3pt]
\int \bigl(\Weyl_1(x) -\Weyl_1^0(x)\bigr)\,\psi^2(x) \,dx + O\bigl(a^{-\frac{1}{3}}h^{-\frac{4}{3}}\bigr)
\label{26-3-48}
\end{multline}
\end{proposition}

Really, $a^{-\frac{1}{3}}h^{-\frac{4}{3}}$ is
$r^{-\frac{1}{2}}h^{-1}+ a^{-1}rh^{-2}$ optimized by
$r\asymp r_*\Def (ah)^{\frac{2}{3}}$; as $h^2 \le a$ we note that
$h^2\le r_* \le a$.

\begin{corollary}\label{cor-26-3-12}
\begin{enumerate}[label=(\roman*), fullwidth]
\item\label{cor-26-3-12-i}
As $M=1$ equality \textup{(\ref{26-3-48})} remains valid with $\psi=1$ and
$a= 1$.
\item\label{cor-26-3-12-ii}
As $M\ge 2$ and $a\ge h^2$ equality  \textup{(\ref{26-3-48})} becomes
\begin{multline}
\Tr \bigl(\psi (H^-_{A,V} - H^-_{A,V^0})\psi\bigr)= \\[3pt]
\int \bigl(\Weyl_1(x) -\Weyl_1^0(x)\bigr)\,\psi^2(x) \,dx + O\bigl((a^{-\frac{1}{3}}+1)h^{-\frac{4}{3}}\bigr)
\label{26-3-49}
\end{multline}
where we reset case $a\ge 1$ to $a=1$.
\end{enumerate}
\end{corollary}

\section{Improved trace estimates}
\label{sect-26-3-4}

\subsection{Improved Tauberian estimates}
\label{sect-26-3-4-1}

Let us apply much more advanced arguments of Section~\ref{book_new-sect-12-5}; recall that these arguments are using long term propagation of singularities. Unfortunately using these arguments we are not able to improve the above results unless $\kappa\ll 1$.

First, let us consider $\psi $ which is $r$-admissible partition element located in $\{x:\ell(x)\asymp r\}$ and we need to estimate an absolute value of
\begin{equation}
F_{t\to h^{-1}\tau} \bar{\chi}_T(t) \Tr \bigl(e^{ih^{-1}tH}\psi\bigr)
\label{26-3-50}
\end{equation}
and to do it we need to estimate the same expression with $\bar{\chi}_T(t)$ replaced by $\chi _{T'}(t)$ with $t_0\le T'\le T$ with
$t_0=\epsilon  \ell \zeta^{-1}= \epsilon r^{\frac{3}{2}}$. We can break $\psi=\psi^+ + \psi^-$ with $\psi^\pm =\psi^\pm (x,hD)$ such that the trajectories in the positive (negative) time direction from support of its symbol $\psi^+(x,\xi)$ are going after time $Ct_0$ in the direction of increased $\ell(x)$ and since we consider trace we need to consider only $\psi^+$ and only
$\chi\in \sC^\infty ([\frac{1}{2},1])$.

The trouble is that we have not rough but non-smooth magnetic field\footnote{\label{foot-26-15} Or rough but with the roughness parameter $\hbar$ which is a bit too short.}; so let us consider
$t_0+t_1+\ldots +t_n \asymp T'$ where $t_j= \epsilon r_j^{\frac{3}{2}}$, $r_j=c^j r$, $j=0,1,\ldots,n$ and estimate an error appearing when we replace in (modified) (\ref{26-3-50}) $e^{ih^{-1}tH}\psi^+$ by
\begin{equation}
e^{ih^{-1}(t-t_n) H} \psi^+_{n+1} e^{ih^{-1}t_n H} \psi_n^+  \cdots e^{ih^{-1}t_1H} \psi_1^+ e^{ih^{-1}t_0H}\psi^+
\label{26-3-51}
\end{equation}
with $\psi^+_j$ defined similarly and Hamiltonian flow from $\supp (\psi^+_j)$ for $t=t_j$ is inside of $\{(x,\xi):\,\psi^+_{j+1}(x,\xi)=1\}$. Therefore we need to estimate an error when we insert $\psi^+_j$.

According to our propagation results (namely, Proposition~\ref{prop-26-2-11}) after $\psi^+_1,\ldots,\psi^+_{j-1}$ were inserted,  insertion of $\psi^+_j$ brings a relative error not exceeding
$C \bigl(\hbar_j^{\theta}\3\partial A\3_{\theta, Y_j}  +  \hbar_j^{s+1}\bigr)$ where $\hbar_j= h r_j^{-\frac{1}{2}}$ and  $Y_j$ is an $\epsilon r_j$-vicinity of $x$-projection of $\supp (\psi_j)$; $s$ is an arbitrarily large exponent.

Recall that $\3\partial A\3_{\theta, Y_j}\le C\kappa \hbar_j^{1-\theta}$ as
$\theta \in (1,2)$; therefore this relative error does not exceed
$C \hbar_j (\kappa+\hbar_j^s)$. So inserting all $\psi^+_j$ brings a relative error
$C\sum_{j\ge 0} \hbar_j(\kappa  +\hbar_j^s)\asymp C \hbar (\kappa+\hbar^s)$ and since a priory expression (\ref{26-3-50}) is bounded by $C\hbar^{-3} T $ we conclude that
\begin{claim}\label{26-3-52}
An absolute value of expression (\ref{26-3-50}) with $\bar{\chi}_T(t)$ replaced by $\chi_{T'}(t)$ with
$T_*\asymp r^{\frac{3}{2}}\le T'\le T^*$\,\footnote{\label{foot-26-16} We discuss the choice of $T^*$ later.} does not exceed
$C\hbar^{-2}(\kappa +\hbar^s)T'$.
\end{claim}
Then an absolute value of expression (\ref{26-3-50}) with $\bar{\chi}_T(t)$ replaced by $\bar{\chi}_T(t)-\bar{\chi}_{T_*}$ does not exceed
$C\hbar^{-2}(\kappa +\hbar^s)$ and since expression (\ref{26-3-50}) with $T=T_*$ does not exceed $C\hbar^{-2}t_0$ we conclude that
\begin{claim}\label{26-3-53}
An absolute value of expression (\ref{26-3-50})  with $T_*\le T\le T^*$ does not exceed $C\hbar^{-2}T_*+C\hbar^{-2}(\kappa +\hbar^s)T$.
\end{claim}

Then
\begin{claim}\label{26-3-54}
An error when we replace $\Tr \bigl(\theta (- H_{A,V})\psi\bigr)$ by its Tauberian expression  with ``time'' $T$ does not exceed
$C\hbar^{-2}\bigl(T_* T^{-1}+ \kappa +\hbar^s\bigr)$
\end{claim}
and
\begin{claim}\label{26-3-55}
An error when we replace $\Tr \bigl(H_{A,V}^-\psi \bigr)$ by its Tauberian expression  with ``time'' $T$ does not exceed
$C\hbar^{-1}\bigl(T_* T^{-1}+ \kappa +\hbar^s\bigr)T_* T^{-1}\zeta^2$.
\end{claim}
In the latter statement we need to remember how everything scales.

Observe that presence of the magnetic field due to its estimates relatively perturbs dynamics by $O(\kappa)$ and therefore as $\kappa$ is sufficiently small does not affect $T^*$. Then assuming that
\begin{multline}
|\nabla^\alpha (V-V^0)|\le \epsilon a^{-1}r^{-|\alpha|}\qquad \forall \alpha:|\alpha|\le 1, \qquad V^0=Zr^{-1}\\
\text{with \ \ } Z\asymp 1,\ \ a\ge h^{2-\delta}
\label{26-3-56}
\end{multline}
we can take $T^* \asymp a^{\frac{3}{2}}$ and therefore we conclude that
The Tauberian error in (\ref{26-3-55}) does not exceed
\begin{equation}
Ch^{-1}a^{-\frac{3}{2}}r\bigl(r^{\frac{3}{2}} a^{-\frac{3}{2}}+ \kappa+ h^s r^{-\frac{1}{2}s}\bigr)
\label{26-3-57}
\end{equation}
and we arrive to statement \ref{prop-26-3-13-i} below.

Meanwhile The Tauberian error in
$\Tr \bigl( (H_{A,V}^- - H_{A,V^0}^-)\psi \bigr)$ does not exceed
\begin{equation}
Ca^{-1}h^{-2}r \bigl(r^{\frac{3}{2}} a^{-\frac{3}{2}}+ \kappa +h^sr^{-\frac{1}{2}s}\bigr)
\label{26-3-58}
\end{equation}
and we arrive to statement \ref{prop-26-3-13-ii} below:

\begin{proposition}\label{prop-26-3-13}
Assume that \textup{(\ref{26-3-56})} is fulfilled and let  $\psi $ be $\ell$-admissible function supported in
$\{x:\, \ell(x)\asymp r \}$ with $h^2 \le r\le a$. Let $A$ satisfy minimizer estimate. Then
\begin{enumerate}[label=(\roman*), fullwidth]
\item\label{prop-26-3-13-i}
The Tauberian error with $T=T^*\asymp a^{\frac{3}{2}}$ in
$\Tr \bigl(H_{A,V}^-\psi  \bigr)$ does not exceed \textup{(\ref{26-3-57})};

\item\label{prop-26-3-13-ii}
The Tauberian error with $T=T^*\asymp a^{\frac{3}{2}}$ in $\Tr \bigl((H_{A,V}^- - H_{A,V^0}^-)\psi \bigr)$
does not exceed \textup{(\ref{26-3-58})}.
\end{enumerate}
\end{proposition}

\begin{proof}
An easy proof following arguments of Section~\ref{book_new-sect-12-5} is left to the reader.
\end{proof}

Note that summation in \ref{prop-26-3-13-i} with respect to $r: b\le r\le a$ returns $Ch^{-1}a^{-\frac{1}{2}}$ and  summation in \ref{prop-26-3-13-ii} with respect to $r: h^2\le r\le b$ returns $Ch^{-2}(a^{-\frac{5}{2}}b^{\frac{5}{2}}+ \kappa a^{-1}b\bigr)+C a^{-1}$. Note also that Statement~\ref{prop-26-3-13-ii} remains true for $r$-admissible function supported in $\{x:\, \ell(x)\le r\}$ with $r\asymp h^2$. Then we arrive to

\begin{corollary}\label{cor-26-3-14}
Assume that \textup{(\ref{26-3-56})} is fulfilled. Then
\begin{enumerate}[label=(\roman*), fullwidth]
\item\label{cor-26-3-14-i}
Let $\phi_2 $ be $\ell$-admissible function supported in
$\{x:\, b\le \ell(x)\le a\}$.  Then the Tauberian error with
$T=T^*\asymp a^{\frac{3}{2}}$  in
$\Tr \bigl(H_{A,V}^-\phi_2  \bigr)$ does not exceed $Ch^{-1}a^{-\frac{1}{2}}$;

\item\label{cor-26-3-14-ii}
Let $\phi_1$ be $\ell$-admissible function supported in
$\{x:\,  \ell(x)\le b\}$.  Then the Tauberian error with
$T=T^*\asymp a^{\frac{3}{2}}$ in
$\Tr \bigl((H_{A,V}^- - H_{A,V^0}^-)\phi_1 \bigr)$ does not exceed
$Ch^{-2}(a^{-\frac{5}{2}}b^{\frac{5}{2}}+ \kappa a^{-1}b\bigr)+C a^{-1}$.
\end{enumerate}
\end{corollary}

\begin{remark}\label{rem-26-3-15}
Obviously we do not need any new assumptions on $\kappa$ to estimate the sum of expressions obtained in Statements \ref{cor-26-3-14-i} and \ref{cor-26-3-14-ii} of Corollary~\ref{cor-26-3-14} by $Ch^{-1}a^{-\frac{1}{2}}$ (as $b\le a^{\frac{1}{2}}h$) here but we need to move from Tauberian expression to Weyl' expression.
\end{remark}

\subsection{Improved Weyl estimates}
\label{sect-26-3-4-2}

Note that in virtue of (\ref{26-3-52}) for element $\psi$ contribution of the time interval $\{t:\,|t|\asymp T'\}$ to the Tauberian expression for $\Tr\bigl(H^-_{A,V}\psi\bigr)$ does not exceed
$C\hbar^{-1}\bigl(\kappa +\hbar^s\bigr)T_* T^{\prime\,-1}\zeta^2$ and therefore replacing $\bar{\chi}_T(t)$ by $\bar{\chi}_{T_*}(t)$ introduces an error not exceeding
\begin{gather}
C\hbar^{-1}\bigl(\kappa +\hbar^s\bigr)\zeta^2\asymp
Ch^{-1}r^{-\frac{1}{2}} \bigl(\kappa + h^sr^{-\frac{1}{2}s}\bigr)
\label{26-3-59}\\
\intertext{and summation with respect to $r: b\le r\le a$ returns}
Ch^{-1}b^{-\frac{1}{2}} \bigl(\kappa + h^sb^{-\frac{1}{2}s}\bigr).
\label{26-3-60}
\end{gather}

On the other hand, also in virtue of (\ref{26-3-52}) for element $\psi$ contribution of the time interval $\{t:\,|t|\asymp T'\}$ to the Tauberian expression for
$\Tr\bigl((H^-_{A,V}-H^-_{A,V^0})\psi\bigr)$  does not exceed
$C\hbar^{-2}a^{-1}\bigl(\kappa +\hbar^s\bigr)$ and therefore replacing $\bar{\chi}_T(t)$ by $\bar{\chi}_{T_*}(t)$ introduces an error not exceeding
$Ch^{-2}a^{-1}r\bigl(\kappa +h^sr^{-\frac{1}{2}s}\bigr)|\log T/T_*|$.

Further, in virtue of (\ref{26-3-53}) the Tauberian error does not exceed $Ch^{-2}a^{-1}r\bigl(r^{\frac{3}{2}}T^{-1}+\kappa + h^s r^{-\frac{1}{2}s}\bigr)$ and adding these two errors together optimizing the sum by
$T\le a^{\frac{3}{2}}$ we get
$T\asymp r^{\frac{5}{2}}(\kappa + h^sr^{-\frac{1}{2}s})^{-1}$ and the sum
\begin{equation}
Ch^{-2}a^{-1}r\bigl(\kappa +h^sr^{-\frac{1}{2}s}\bigr)
\bigl(|\log (\kappa +h^sr^{-\frac{1}{2}s})|+ 1\bigr) +Ch^{-2}a^{-\frac{5}{2}}r^{\frac{5}{2}}.
\label{26-3-61}
\end{equation}

Meanwhile repeating arguments of Subsection~\ref{sect-26-3-3} one can see easily that
\begin{claim}\label{26-3-62}
The difference between Tauberian expression with $T=T_*$ and Weyl expression for $\Tr\bigl(H^-_{A,V}\psi\bigr)$ does not exceed (\ref{26-3-59}) with any $s<2$
\end{claim}
and
\begin{claim}\label{26-3-63}
The difference between Tauberian expression with $T=T_*$ and Weyl expression for $\Tr\bigl((H^-_{A,V}- H^-_{A,V^0})\psi\bigr)$ does not exceed
\begin{equation}
Ch^{-2}a^{-1}r\bigl(\kappa +h^sr^{-\frac{1}{2}s}\bigr)
\label{26-3-64}
\end{equation}
with any $s<2$ and thus does not exceed (\ref{26-3-61}).
\end{claim}
Summation with respect to $r: h^2\le r\le b$  of (\ref{26-3-61}) returns
\begin{equation}
Ch^{-2}a^{-1}b \kappa |\log \kappa|  + Ch^{-2}b^{\frac{5}{2}}a^{-\frac{5}{2}}+ Ca^{-1};
\label{26-3-65}
\end{equation}
adding expression (\ref{26-3-60}) and optimizing the sum by $b:\,h^2\le b\le a$ we get
$b\asymp (ah |\log \kappa|)^{\frac{2}{3}}$ and expression
\begin{equation}
Ch^{-\frac{4}{3}}a^{-\frac{1}{3}}\kappa |\log \kappa|^{\frac{1}{3}}+
Ch^{-1}a^{-\frac{1}{2}}.
\label{26-3-66}
\end{equation}

Thus we have proven

\begin{proposition}\label{prop-26-3-16}
\begin{enumerate}[label=(\roman*), fullwidth]
\item \label{prop-26-3-16-i}
In the framework of Proposition~\ref{prop-26-3-11}
\begin{multline}
\Tr \bigl(\psi (H^-_{A,V} - H^-_{A,V^0})\psi\bigr)= \\[3pt]
\int \bigl(\Weyl_1(x) -\Weyl_1^0(x)\bigr)\,\psi^2(x) \,dx + O\bigl(h^{-\frac{4}{3}}a^{-\frac{1}{3}}\kappa |\log \kappa|^{\frac{1}{3}}+
h^{-1}a^{-\frac{1}{2}}\bigr).
\label{26-3-67}
\end{multline}
\item \label{prop-26-3-16-ii}
In particular as
\begin{equation}
\kappa \le c a^{-\frac{1}{6}}h^{\frac{1}{3}}
|\log ah^{-2}|^{-\frac{1}{3}}
\label{26-3-68}
\end{equation}
the error in \textup{(\ref{26-3-67})} does not exceed $Ch^{-1}a^{-\frac{1}{2}}$ as in the case without magnetic field.
\end{enumerate}
\end{proposition}

\begin{remark}\label{rem-26-3-17}
\begin{enumerate}[label=(\roman*), fullwidth]
\item\label{rem-26-3-17-i}
Obviously we could consider $a=1$ and then just rescale $x\mapsto xa^{-1}$, $\tau\mapsto \tau a$, $h\mapsto ha^{-\frac{1}{2}}$;

\item\label{rem-26-3-17-ii}
One may wonder if the same approach works for estimate of $A$.  First of all, there is no improvement for estimate for $|\partial^2 A|$ as it follows from the estimate for $|\Delta A|$ which is a pointwise estimate;

\item\label{rem-26-3-17-iii}
However as $\partial A$ and $A$ are mollifications of $\Delta A$ one can improve estimates for them as $\kappa \ll 1$ and $\ell \ll a$; however there are no improvements as either $\kappa \asymp 1$ or $\ell \ge a$. Since these improvements do not lead to the improvements of our final results we do not pursue them.
\end{enumerate}
\end{remark}

\section{Single singularity}
\label{sect-26-3-5}

\subsection{Coulomb potential}
\label{sect-26-3-5-1}

Consider now exactly Coulomb potential: $V=Z |x|^{-1}$. Let us establish the existence of the Scott correction:

\begin{proposition}\label{prop-26-3-18}
Let $V=Z |x|^{-1}$, $h>0 $, $Z>0$ and $0< \kappa \le \kappa^*$. Then
\begin{enumerate}[label=(\roman*), fullwidth]
\item\label{prop-26-3-18-i}
The following limit exists
\begin{multline}
\lim_{r\to\infty} \biggl( \inf_A \Bigl( \Tr \bigl( \phi_r H_{A,V} \phi_r\bigr)^- + \frac{1}{\kappa h^2} \int |\partial A|^2 \,dx  \Bigr)\\
- \int  \Weyl_1 (x)\phi_r ^2(x)\,dx \biggr) \Fed 2 Z^2 h^{-2} S(Z\kappa);
\label{26-3-69}
\end{multline}
\item\label{prop-26-3-18-ii}
And it coincides with
\begin{multline}
\lim_{\eta\to 0^+} \biggl( \inf_{A,V} \Bigl( \Tr \bigl(  H_{A,V}  +\eta\bigr)^- + \frac{1}{\kappa h^2}  \int |\partial A|^2 \,dx \Bigr) \\
-  \int \Weyl_1(H_{A,V}+\eta,x)\,dx \biggr)
\label{26-3-70}
\end{multline}
\item\label{prop-26-3-18-iii}
And  with
\begin{multline}
\inf_A \biggl( \int \Bigl(e_1(H_{A,V}; x,x,0)-
\Weyl_1(H_{A,V} ,x)\Bigr)\,dx+\\
\frac{1}{\kappa h^2}  \int |\partial A|^2 \,dx \biggr);
 \label{26-3-71}
 \end{multline}
\item\label{prop-26-3-18-iv}
We also can replace in \ref{prop-26-3-18-i}
$\Tr \bigl( \phi_r H_{A,V} \phi_r\bigr)^-$ by
$\Tr \bigl( \phi_r H^-_{A,V} \phi_r\bigr)$.
\end{enumerate}
Here $\phi \in \sC_0^\infty (B(0,1))$, $\phi=1$ in $B(0,\frac{1}{2})$, $\phi_r=\phi(x/r)$.
\end{proposition}

\begin{proof}
Due to scaling $x\mapsto Z h^{-2} x$, $A\mapsto Z^{-1} h  A $,
$\partial A\mapsto Z^{-2}h^3 \partial A$ one needs to consider only $Z=h=1$; all expressions on the left scale exactly as $Z^2h^{-2}S(Z\kappa)$.

\bigskip\noindent
(i) Let us compare
\begin{equation*}
Q(r,\kappa, A)\Def \Tr (\phi_r H_{A,V}\phi_r)^- -
\int \Weyl_1(H_{A,V} ,x)\phi_r^2 (x) \,dx+
\frac{1}{\kappa}  \int |\partial A|^2 \,dx
\end{equation*}
and $Q(r',\kappa,A)$ with $r\ge 1$ and $r'\ge 2r$. Note that
\begin{multline*}
Q(r',\kappa, A)\ge Q(r, (1+\epsilon)\kappa, A)+\\[3pt]
\sum _{1\le j\le J}
\Bigl(\Tr (\psi_{2^jr} H_{A,V}\psi_{2^jr})^--
\int \Weyl_1(H_{A,V} ,x)\psi_{2^jr}^2 (x) \,dx+
\frac{\epsilon}{2\kappa}  \int |\partial A|^2\bar{\psi}_{2^jr}^2(x)\,dx \Bigr)
\end{multline*}
where $\psi $  and $\bar{\psi}$ are smooth compactly supported functions, equal $0$ in $B(0,\frac{1}{2})$ and $\bar{\psi}=1$ in the vicinity of $\supp(\psi)$,
$J=\lfloor\log_2 r'/r\rfloor$, $\epsilon>0$ is arbitrarily small.

Therefore we can replace in the sum $A$ by $A_j$ and $\bar{\psi}_{2^jr}$ by $1$ in $\int |\partial A|^2\bar{\psi}_{2^jr}^2(x)\,dx$; but then in the virtue of Section~\ref{sect-26-2} each term in the sum is bounded from below by
$-C'(\epsilon)\int \rho^3 \ell^{-1}\bar{\psi}_{t}\,dx=
-C'(\epsilon)t^{-\frac{1}{2}}$ with $t=2^jr$. Then
\begin{equation}
Q(r',\kappa, A)\ge Q(r, (1+\epsilon)\kappa, A)-C'(\epsilon) r^{-\frac{1}{2}}.
 \label{26-3-72}
\end{equation}
We know that $Q(r, (1+\epsilon)\kappa, A)$ is bounded from below by $-C(r)$ but now we conclude that this bound is uniform with respect to $r$. Then
$2S(\kappa)\Def \liminf_{r\to +\infty}\inf _A Q(r,\kappa, A)>-\infty$. Further, (\ref{26-3-72}) implies that
\begin{gather}
2S(\kappa)+ C'(\epsilon) r^{-\frac{1}{2}}\ge \inf_A  Q(r, (1+\epsilon)\kappa, A)
\notag\\
\shortintertext{and therefore}
\limsup_{r\to +\infty}\inf_A Q(r,\kappa, A)\le 2S( (1+\epsilon)^{-1}\kappa ).
\label{26-3-73}
\end{gather}
Furthermore, plugging $A=0$ we can see easily that  $Q(r,\kappa, A)$ is uniformly bounded from above and therefore $2S(\kappa)<+\infty$; also our arguments imply that $\int |\partial A|^2\,dx$ is uniformly bounded for near optimizers and therefore $S(\kappa)$ is continuous with respect to $\kappa<\kappa^*$; combining with (\ref{26-3-73}) we arrive to Statement~\ref{prop-26-3-18-i}.

\bigskip\noindent
(ii) Similarly, (\ref{26-3-72}) holds with $H_{A,V}$ replaced by $H_{A,V}+\eta$ and then we can take $r'=\infty$ and apply $\inf_A$ to both sides arriving to
\begin{multline*}
\inf_A\biggl(\Tr \bigl( H_{A,V}+\eta\bigr)^- -
\int \Weyl_1(H_{A,V+\eta} ,x)\,dx+
\frac{1}{\kappa}  \int |\partial A|^2 \,dx\biggr)\ge \\
\inf_A\biggl( \Tr \bigl( \phi_r(H_{A,V}+\eta)\phi_r \bigr)^- -
\int \Weyl_1(H_{A,V+\eta} ,x)\phi_r^2 (x) \,dx+\\
\frac{1}{ (1+\epsilon)\kappa}  \int |\partial A|^2\biggr) \,dx-C'(\epsilon)r^{-\frac{1}{2}}.
\end{multline*}
After this as $\eta\to +0$ the right-hand expression tends to itself with $\eta=0$; tending $r\to+\infty$ we get there $2S((1+\epsilon)\kappa)$ in virtue of~Statement~\ref{prop-26-3-18-i} and tending $\epsilon\to+0$ we arrive to
\begin{multline}
\liminf_{\eta\to +0} \inf_A \biggl(\Tr \bigl( H_{A,V}+\eta\bigr)^- -
\int \Weyl_1(H_{A,V+\eta} ,x)\,dx+\\
\frac{1}{\kappa}  \int |\partial A|^2 \,dx\biggr)\ge
2S(\kappa).
\label{26-3-74}
\end{multline}
On the other hand, consider
\begin{equation*}
\Tr \bigl( H_{A,V}+\eta\bigr)^- -
\int \Weyl_1(H_{A,V+\eta} ,x)\,dx+
\frac{1}{\kappa}  \int |\partial A|^2 \,dx
\end{equation*}
and replace $\Tr \bigl( H_{A,V}+\eta\bigr)^-$ by
\begin{equation}
\underbracket{\Tr \bigl(\bigl( H_{A,V}+\eta\bigr)^-\phi_r^2\bigr)+
\frac{1}{\kappa}\int |\partial A|^2\,dx }+
\Tr \bigl(\bigl( H_{A,V}+\eta\bigr)^-(1-\phi_r^2)\bigr).
\label{26-3-75}
\end{equation}
Let $A$ be a minimizer of the first expression; then in virtue of Propositions~\ref{prop-26-2-24}--\ref{prop-26-2-26} this minimizer is sufficiently ``good'' on $\epsilon r$-vicinity of $\supp (1-\phi_r^2)$ that the the difference between the second term and its Weyl expression does not exceed $Cr^{-\frac{1}{2}}$; one can prove it easily by $\ell(x)$-admissible partition of unity as in part (i) of the proof and we leave details to the reader.

Observe that the first term in (\ref{26-3-75}) is
$\inf_A \Tr \bigl(\phi_r\bigl( H_{A,V}+\eta\bigr)^-\phi_r\bigr)$. In this expression we can take limit as $\eta\to +0$ just setting $\eta=0$ and therefore the left-hand expression in (\ref{26-3-74}) with $\liminf$ replaced by $\limsup$ does not exceed
\begin{equation*}
\inf_A\Bigl(\Tr \bigl(\phi_r H_{A,V}^- \phi_r\bigr)-
\int \Weyl_1(H_{A,V} ,x) \phi_r^2(x)\,dx+\\
\frac{1}{\kappa}  \int |\partial A|^2 \,dx \Bigr)+Cr^{-\frac{1}{2}};
\end{equation*}
taking limit as $r\to +\infty$ we conclude that the left-hand expression in (\ref{26-3-74}) with $\liminf$ replaced by $\limsup$ does not exceed
\begin{equation*}
\liminf_{r\to+\infty} \inf_A\Bigl(\Tr \bigl(\phi_r H_{A,V}^- \phi_r\bigr)-
\int \Weyl_1(H_{A,V} ,x) \phi_r^2(x)\,dx+
\frac{1}{\kappa}  \int |\partial A|^2 \,dx \Bigr).
\end{equation*}
This expression does not exceed (\ref{26-3-69}) and therefore combining with (\ref{26-3-74}) we prove Statements~\ref{prop-26-3-18-ii} and~\ref{prop-26-3-18-iv}.

\bigskip\noindent
(iii) Similarly
\begin{multline*}
 \int \Bigl(e_1(H_{A,V}; x,x,0)-
\Weyl_1(H_{A,V} ,x)\Bigr)\phi_r^2(x)\,dx+
\frac{1}{\kappa}  \int |\partial A|^2 \,dx \ge\\
\inf_A \Bigl(\Tr \bigl(\phi_r H_{A,V}\phi_r\bigr)^-
-\int \Weyl_1(x)\phi_r^2 (x)\,dx +
\frac{1}{\kappa}  \int |\partial A|^2\Bigr)-Cr^{-\frac{1}{2}}
 \end{multline*}
and therefore
\begin{multline*}
\inf_A \liminf_{r\to\infty} \int \Bigl(e_1(H_{A,V}; x,x,0)-
\Weyl_1(H_{A,V} ,x)\Bigr)\phi_r^2(x)\,dx+\\
\frac{1}{\kappa}  \int |\partial A|^2 \,dx \ge
2S(\kappa).
\end{multline*}
On the other hand, as in (ii) taking $A$ to be a minimizer of the first expression in (\ref{26-3-75}) we see that
\begin{multline*}
\inf_A \limsup_{r\to\infty} \int \Bigl(e_1(H_{A,V}; x,x,0)-
\Weyl_1(H_{A,V} ,x)\Bigr)\phi_r^2(x)\,dx+\\
\frac{1}{\kappa}  \int |\partial A|^2 \,dx \le
2S(\kappa)
\end{multline*}
and Statement~\ref{prop-26-3-18-iii} has been proven.
\end{proof}

\begin{remark}\label{rem-26-3-19}
\begin{enumerate}[label=(\roman*), fullwidth]
\item\label{rem-26-3-19-i}
Statements similar to\ref{prop-26-3-18-i}, \ref{prop-26-3-18-ii}  were proven in L.~Erd\"os,  S.~Fournais,  and J. P.  Solovej \cite{EFS3} (see Theorem~2.4 and Lemma~2.5 respectively).
\item\label{rem-26-3-19-ii}
Again as observed in in L.~Erd\"os,  S.~Fournais,  and J. P.  Solovej \cite{EFS3} we do not know if (a) $S(\kappa)<S(0)$  for $\kappa >0$ or just (b) $S(\kappa)=S(0)$ as $\kappa<\kappa^*$ and $S(\kappa)=-\infty$ as $\kappa>\kappa^*$. If we knew that the optimizer is unique then obviously $A=0$ and it would be relatively easy but rather unexciting the latter case.
\item\label{rem-26-3-19-iii}
While we assumed that $\kappa<\kappa^*$ with $\kappa^*>0$ and it is possible that $S(\kappa)=-\infty$ as $\kappa>\kappa*$ with some $\kappa^*<\infty$ we are not aware about any proof of this, so in fact it could happen that $\kappa^*=+\infty$ and then condition $\kappa <\kappa^*$ is superficial and one needs to study asymptotics of $S(\kappa)$ as $\kappa \to +\infty$.
\end{enumerate}
\end{remark}

\begin{proposition}\label{prop-26-3-20}
 As $0<\kappa <\kappa'$
 \begin{equation}
 S(\kappa')\le S(\kappa) \le S(\kappa') + C \kappa '(\kappa^{-1}-\kappa'^{-1}).
 \label{26-3-76}
 \end{equation}
\end{proposition}

\begin{proof}
Monotonicity of $S(\kappa)$ is obvious.

Let $0<\kappa <\kappa' <\kappa''\le \kappa^*$. Then for any $\varepsilon>0$ if $r=r_\varepsilon$ is large enough    then the left-hand expression in (\ref{26-3-69}) for $\kappa'$ (without $\inf$ and $\lim$) is greater than $S(\kappa'') -\varepsilon + (\kappa'^{-1}-\kappa''^{-1})\|\partial A\|^2$;   also, if $A$ is an almost minimizer there, it is less than $S(\kappa')+\varepsilon$.

Therefore
$(\kappa'^{-1}-\kappa''^{-1})\|\partial A\|^2\le |S(\kappa'') -S(\kappa') | +2\varepsilon$. But then
\begin{multline*}
S(\kappa) -\varepsilon \le S(\kappa')+\varepsilon + (\kappa^{-1}-\kappa'^{-1})\|\partial A\|^2\le \\[3pt]
S(\kappa')+\varepsilon +
C (\kappa^{-1}-\kappa'^{-1})(\kappa'^{-1}-\kappa''^{-1})^{-1}
\bigl(|S(\kappa'') -S(\kappa') | +2\varepsilon\bigr)
\end{multline*}
and therefore
\begin{equation}
(\kappa^{-1}-\kappa'^{-1}) ^{-1} |S(\kappa) - S(\kappa') |\le
(\kappa'^{-1}-\kappa''^{-1})^{-1}|S(\kappa') -S(\kappa'') |
\label{26-3-77}
\end{equation}
which for $\kappa''=\kappa^*$ implies (\ref{26-3-76}).
\end{proof}

\begin{remark}\label{rem-26-3-21}
Using global equation (\ref{26-2-14}) we conclude that for $Z=h=1$
\begin{align}
&|\partial^\alpha  A| \le C\kappa \ell^{-1-|\alpha|}\qquad
&&\text{as\ \ }\ell\ge 1,\ |\alpha|\le 1, \label{26-3-78}\\[3pt]
&|\partial^\alpha  A| \le C\kappa
&&\text{as\ \ } \ell\le 1,\ |\alpha|\le 1, \label{26-3-79}
\end{align}
\begin{gather}
\|\partial A\|^2 \le C\kappa^2.\label{26-3-80}\\
\shortintertext{Then}
 S'(\kappa) \le C, \qquad |S(\kappa (1+\eta))-S(\eta)|\le C\kappa \eta.
 \label{26-3-81}
\end{gather}
\end{remark}

\subsection{Main theorem}
\label{sect-26-3-5-2}

In the ``atomic'' case $M=1$ we arrive instantly to

\begin{theorem}\label{thm-26-3-22}
Let $M=1$ and $\kappa \le \kappa^*$. Then
\begin{enumerate}[label=(\roman*), fullwidth]
\item\label{thm-26-3-22-i}
Asymptotics holds
\begin{equation}
\E^* = \int \Weyl_1(x)\,dx + 2 z^2 S(z \kappa) h^{-2}+
O(h^{-\frac{4}{3}}\kappa |\log \kappa|^{\frac{1}{3}}+h^{-1});
\label{26-3-82}
\end{equation}
\item\label{thm-26-3-22-ii}
If $\kappa =o (h^{\frac{1}{3}}|\log h|^{-\frac{1}{3}})$  then
\begin{equation}
\E^* = \int \Weyl^*_1(x)\,dx + 2 z^2 S(z \kappa) h^{-2}+
o( h^{-1}).
\label{26-3-83}
\end{equation}
\end{enumerate}
\end{theorem}

\begin{proof}
If $A$ satisfies minimizer properties  then in virtue of Proposition~\ref{prop-26-3-16}
\begin{multline}
\Tr^-  H_{A,V} -\int \Weyl_1 (x)\,dx \equiv
\Tr^-  H_{A,V^0} -\int \Weyl^0_1 (x)\,dx \\
+ O(h^{-\frac{4}{3}}\kappa |\log \kappa|^{\frac{1}{3}}+h^{-1})
\label{26-3-84}
\end{multline}
and adding magnetic energy and plugging either minimizer for $V$ or for $V^0$ we get
\begin{align}
&\inf_A \Bigl(\Tr^-  H_{A,V} -\int \Weyl_1 (x)\,dx +
\frac{1}{\kappa h^2}\int |\partial A|^2\,dx\Bigr)
\lesseqgtr  \label{26-3-85}\\
&\inf_A \Bigl(\Tr^-  H_{A,V^0} -\int \Weyl^0_1 (x)\,dx +
\frac{1}{\kappa h^2}\int |\partial A|^2\,dx\Bigr)\notag \\
&\qquad\qquad\qquad\qquad\qquad\qquad\qquad\qquad
\pm  C(h^{-\frac{4}{3}}\kappa |\log \kappa|^{\frac{1}{3}}+h^{-1}).
\notag
\end{align}
Sure as $V$ (and surely $V^0$) are not sufficiently fast decaying at infinity the left (and for sure the right hand) expression in (\ref{26-3-84}) should be regularized as in Subsection~\ref{sect-26-3-5-1}. However for potential decaying fast enough (faster than $|x|^{-2-\delta}$) regularization is not needed.

For $V^0$ we have an exact expression which concludes the proof of Statement~\ref{thm-26-3-22-i}.

The proof of (ii) is similar albeit with the small improvement based on the behavior of the classical dynamics (without magnetic field) exactly as in Chapter~\ref{book_new-sect-24}.
\end{proof}

\section{Several singularities}
\label{sect-26-3-6}

Consider now ``molecular'' case $M\ge 1$. Then we need more delicate arguments.

\subsection{Decoupling of singularities}
\label{sect-26-3-6-1}

Consider partition of unity $1=\sum_{0\le m\le M} \psi_m^2$ where $\psi_m$ is supported in $\frac{1}{3}a$-vicinity of $\y_m$ as $m=1,\ldots, M$ and $\psi_0=0$ in $\frac{1}{4}a$-vicinities of $\y_m$ (``near-nuclei'' and ``between-nuclei''partition elements).

\paragraph{Estimate from above.}
\label{sect-26-3-6-1-1}

Then
\begin{equation}
\Tr (H^-_{A,V}) =\sum_{0\le m\le M} \Tr  (\psi_m H^-_{A,V} \psi_m)
\label{26-3-86}
\end{equation}
and to estimate $\E^*$ from the above we impose an extra condition to $A$:
\begin{equation}
A=0 \qquad \text{as \ \ } \ell (x) \ge \frac{1}{5}a.
\label{26-3-87}
\end{equation}
Then in this framework we estimate
\begin{equation}
\Tr^- (\psi_0 H^-_{A,V} \psi_0)-\int \Weyl_1(x)\psi_0^2 (x)\,dx \le
C h^{-1} a^{-\frac{1}{2}}.
\label{26-3-88}
\end{equation}
Proof of this inequality is trivial by using $\ell$-admissible partition and applying results of the theory without any magnetic field.

So, to estimate $\E^*$ from above\footnote{\label{foot-26-17} Modulo error in (\ref{26-3-49}).} we just need to estimate from  above the minimum with respect to $A$ satisfying (\ref{26-3-87}) of the expression
\begin{equation}
\Tr (\psi_m H^-_{A,V} \psi_m)-\int \Weyl_1(x)\psi_m^2 (x)\,dx +
\frac{1}{\kappa h^2}  \int_{\{\ell_m(x)\le \frac{1}{5}a\}}  |\partial A|^2\,dx.
\label{26-3-89}
\end{equation}

\paragraph{Estimate from below.}
\label{sect-26-3-6-1-2}

In this case we use the same partition of unity $\{\psi_m^2\}_{j=0,1,\ldots, M}$ and estimate
\begin{gather}
\Tr (H^- _{A,V})\ge \sum_{0\le m\le M} \Tr^- (\psi_m H_{A,V'} \psi_m)
\label{26-3-90}\\
\shortintertext{with}
V'= V+ 2h^2 \sum _{0\le m\le M}  (\partial \psi )^2
\label{26-3-91}\\
\shortintertext{and we also use decomposition}
\int |\partial A|^2\,dx = \sum_{0\le m\le M} \int \omega_m^2 |\partial A|^2\, dx
\label{26-3-92}
\end{gather}
with
\begin{multline}
\omega_m(x) =1 \ \  \text{as\ \ } \ell_m(x)\le \frac{1}{10}a,\quad
\omega_m (x) \ge 1-C\varsigma  \ \  \text{as\ \ } \ell_m (x)\le \frac{1}{2}a
\\ m=1,\ldots, M,
\label{26-3-93}
\end{multline}
\begin{equation}
\omega_0 \ge \epsilon_0 \varsigma\quad \text{as\ \ } \ell (x)\ge \frac{1}{5}a.
\label{26-3-94}
\end{equation}
So far $\varsigma>0$ is a constant but later it will be a small parameter.
Then since
\begin{multline}
\Tr^- (\psi_0 H_{A,V'} \psi_0)-\int \Weyl_1(x)\psi_0^2 (x)\,dx +
\frac{1}{\kappa h^2} \int \omega_0^2 |\partial A|^2\,dx\ge \\Ch^{-1}a^{-\frac{1}{2}}
\label{26-3-95}
\end{multline}
(again proven by partition) in virtue of the previous Section~\ref{sect-26-2} we are left with the estimates from below for
\begin{equation}
\Tr^- (\psi_m H_{A,V'} \psi_m)-\int \Weyl_1(x)\psi_m^2 (x)\,dx +
\frac{1}{\kappa h^2} \int \omega_m^2 |\partial A|^2\,dx.
\label{26-3-96}
\end{equation}

\begin{remark}\label{rem-26-3-23}
\begin{enumerate}[label=(\roman*), fullwidth]
\item\label{rem-26-3-23-i}
Note that the error in $\Weyl_1$ when we replace $V'$ there by $V$ does not exceed  $Ch^{-1}(1+a^{-\frac{1}{2}})$ which is  less than the error in (\ref{26-3-49}). Here we can also assume that $A$ satisfies (\ref{26-3-87}); we need just to replace $\varsigma$ by $\epsilon_0\varsigma$  in (\ref{26-3-93})--(\ref{26-3-94}).

\item\label{rem-26-3-23-ii}
We can further go down by replacing $\Tr^- (\psi_m H_{A,V'} \psi_m)$ by
$\Tr  (\psi_m H^-_{A,V'} \psi_m)$.

\item\label{rem-26-3-23-iii}
Therefore we basically have the same object for both estimates albeit with marginally different potentials ($V$ in the estimate from above and $V'$ in the estimate from below)  and with a weight $\omega_m^2$ satisfying (\ref{26-3-93})--(\ref{26-3-94}); in both cases $\omega =1$ as
$\ell(x)\le \frac{1}{10}a$ but in the estimate from above $\omega (x)$ grows to $C_0$ and in the estimate from below $\omega(x)$ decays to $\varsigma$ as $\ell(x)\ge \frac{1}{3}a$ and in both cases condition (\ref{26-3-87}) could be either imposed or skipped.

\item\label{rem-26-3-23-iv}
From now on we consider a single singularity at $0$ and we skip subscript $m$. However if there was a single singularity from the beginning, all arguments of this and forthcoming paragraphs would be unnecessary.
\end{enumerate}
\end{remark}

\paragraph{Scaling.}
\label{sect-26-3-6-1-3}

\bigskip\noindent
(i) We are done as $Z\asymp 1$ but as $Z \ll 1$\,\footnote{\label{foot-26-18} As $Z$ denotes $Z_m$ now we assume only that $Z_1+\ldots+Z_M\asymp 1$.} we need a bit more fixing. The problem is that $V \asymp Z \ell ^{-1}$ only as $|x|\le aZ$; otherwise $V\lesssim a^{-1}$ (where we still assume that $a\le 1$). To deal with this we apply in the zone $\{x:\,  a Z \le |x|\le a\}$ the same procedure as before and its contribution to the error will be  $Ch^{-1} a^{-\frac{1}{2}}$ as
$\rho = a^{-\frac{1}{2}}$ here. Actually we also need to keep
$|x|\ge Z^{-1} h^2$; so we assume that $Z^{-1}h^2\le Z a$ i.e.
$Z  \ge a^{-\frac{1}{2}} h$.

Now let us scale $x\mapsto x'=x (aZ)^{-1} $, and multiply $H_{a,V}$ by $a$ and therefore also multiply $A$ by $a^{\frac{1}{2}}$, so
$A\mapsto A'=a^{\frac{1}{2}}A$, $h\mapsto h'=h a^{-\frac{1}{2}} Z^{-1}$; then the  magnetic energy becomes
$\kappa^{-1}h^{-2} Z\int \omega^2 (x) |\partial' A'|^2\, dx'$ where factors $a^{-1}$ and $aZ$ come from substitution $A=a^{-\frac{1}{2}}A'$ and scaling respectively. We need to multiply it by $a$ (as we multiplied an operator); then plugging $h^{-2}=h'^{-2} a^{-1} Z^{-2}$ we get the same expression as before but with $Z'=1$, $a'=1$ and $h'=h a^{-\frac{1}{2}} Z^{-1}\le 1$ and
$\kappa'=\kappa Z$ instead of $h$ and $\kappa$.

If we establish here an error
$O\bigl(h'^{-1} +\kappa' |\log \kappa'|^{\frac{1}{3}}h'^{-\frac{4}{3}}\bigr)$ the final error will be this expression multiplied by $a^{-1}$ i.e.
$O\bigl(a^{-\frac{1}{2}}Z h^{-1} + \kappa a^{-\frac{1}{2}} Z^{\frac{7}{3}}
|\log \kappa Z|^{\frac{1}{3}}h^{-\frac{4}{3}}\bigr)$ which is less than the same expression with $Z=1$.

\bigskip\noindent
(ii) On the other hand, let $Z  \le a^{-\frac{1}{2}} h$. Recall, we assume that $a\ge C_0 h^2$. Then  we can apply the same arguments as before but with $\bar{Z}=a^{-\frac{1}{2}} h$ and we arrive to the same situation as before albeit with $h'=1$, $a'=1$, $\kappa'= \kappa a^{-\frac{1}{2}} h$ and with
$Z'= Z \bar{Z}^{-1}$. Then we have the trivial error estimate  $O(a^{-1})$ which is less than $O(a^{-\frac{1}{2}}h^{-1})$.

\subsection{Main results}
\label{sect-26-3-6-2}

Combining results of the previous Subsubsections and Paragraphs with proposition~\ref{prop-26-3-9} we arrive to

\begin{theorem}\label{thm-26-3-24}
Let $M\ge 2$, $\kappa \le \kappa^*$ and \textup{(\ref{26-3-36})} hold with $\nu>\frac{4}{3}$. Then
\begin{enumerate}[label=(\roman*), fullwidth]
\item\label{thm-26-3-24-i}
Asymptotics holds
\begin{equation}
\E^* = \int \Weyl_1(x)\,dx + 2 \sum_{1\le m\le M} Z_m^2 S(Z_m \kappa) h^{-2}+ O(R_1+R_2)
\label{26-3-97}
\end{equation}
with
\begin{align}
R_1\, =\,&\left\{\begin{aligned}
&h^{-1}+\kappa |\log \kappa|^{\frac{1}{3}}h^{-\frac{4}{3}}\qquad\ \ &&\text{as\ \ } a\ge 1\\
&a^{-\frac{1}{2}}h^{-1}+\kappa |\log \kappa|^{\frac{1}{3}}a^{-\frac{1}{3}}h^{-\frac{4}{3}}\qquad
&&\text{as\ \ } h^2\le a\le 1
\end{aligned}\right.
\label{26-3-98}\\
\shortintertext{and}
R_2=\kappa h^{-2}
&\left\{\begin{aligned}
&a^{-3}\qquad\qquad\qquad\qquad &&\text{as\ \ } a\ge |\log h|^{\frac{1}{3}},\\
&|\log h^2a^{-1}|^{-1} \qquad &&\text{as\ \ }  h^2\le a \le |\log h|^{\frac{1}{3}};
\end{aligned}\right.
\label{26-3-99}
\end{align}
\item\label{thm-26-3-24-ii}
If $\kappa =o( h^{\frac{1}{3}}|\log h|^{-\frac{1}{3}})$, $\kappa a^{-3}=o(h)$ and $a^{-1}=o(1)$ then
\begin{equation}
\E^* = \int \Weyl^*_1(x)\,dx + 2 \sum_{1\le m\le M} Z_m^2 S(Z_m \kappa) h^{-2}+ o(h^{-1}).
\label{26-3-100}
\end{equation}
\end{enumerate}
\end{theorem}

\begin{proof}
To prove theorem we need to prove the following estimate
\begin{equation}
\frac{1}{\kappa}\|\partial A \|^2 _{\{b\le \ell(x)\le 2b\}}
\le Cb^{-3}
\label{26-3-101}
\end{equation}
where $r_*\le b\le a$ is a ``cut-off''. On the other hand we know that
\begin{equation}
\frac{1}{\kappa}\|\partial A\|^2=-\frac{\partial S}{\partial \kappa}=O(1)
\label{26-3-102}
\end{equation}
and we need to recover the last factor in the definition of $R_2$.

As $a\ge 1$ we can have $\kappa a^{-3}$ because in virtue of (\ref{26-3-39}) the square of the partial norm in (\ref{26-3-102}) does not exceed $Ca^{-3}\kappa^2$.

On the other hand, as $h^2\le r_* \le a$ we can select $b: r_*\le b\le a$ such that  he partial norm in (\ref{26-3-102}) does not exceed
$C|\log (a/h^2) |^{-1}\cdot \|\partial A\|^2$.
\end{proof}

\begin{remark}\label{rem-26-3-25}
\begin{enumerate}[label=(\roman*), fullwidth]
\item\label{rem-26-3-25-i}
As $a\le |\log h|$  we do not need assumption (\ref{26-3-36});

\item\label{rem-26-3-25-ii}
Following arguments of Section~\ref{book_new-sect-12-5}  estimates (\ref{26-3-83}) and (\ref{26-3-102}) could be improved to $O(h^{-1+\delta})$ provided  $a\ge h^{-\delta_1}$,
$\kappa \le h^{\frac{1}{3}+\delta_1}|\log h|^{-\frac{1}{3}} $ and
$\kappa \le a^3 h^{1+\delta_1}$.
\end{enumerate}
\end{remark}

\subsection{Problems and remarks}
\label{sect-26-3-6-3}

\begin{Problem}\label{Problem-26-3-26}
\begin{enumerate}[label=(\roman*), fullwidth]
\item\label{Problem-26-3-26-i}
As $\kappa \in [0, \kappa^*]$ with small enough $\kappa^*>0$ does $S(\kappa)$ really depend on $\kappa$ or $S(\kappa)=S(0)$  (see Remark~\ref{rem-26-3-19})?

\item\label{Problem-26-3-26-ii}
If $S(\kappa)$ really depends on $\kappa$, what is asymptotic behavior of $S(\kappa)-S(0)$ as $\kappa\to +0$: can one improve $S(\kappa)-S(0)=O(\kappa)$?

\item\label{Problem-26-3-26-iii}
Do we really need an assumption $\kappa \in [0, \kappa^*]$  (again see Remark~\ref{rem-26-3-19})?

\item\label{Problem-26-3-26-iv}
 Can one improve estimates to minimizer?
\end{enumerate}
\end{Problem}

\chapter{Asymptotics of the ground state energy}
\label{sect-26-4}

\section{Problem}
\label{sect-26-4-1}

Now we are ready to tackle our original object (\ref{26-1-2})--(\ref{26-1-3}). So, let us consider our usual quantum Hamiltonian
\begin{gather}
\sfH=\sum_{1\le j\le N}H ^0_{x_j}+\sum_{1\le j<k\le N} |x_j-x_k| ^{-1}
\label{26-4-1}\\
\shortintertext{in}
\fH= \bigwedge_{1\le n\le N} \sfH, \qquad \sfH=\sL^2 (\bR^3, \bC^2)
\label{26-4-2}\\
\shortintertext{with}
H ^0=\bigl((i\nabla -A)\cdot \boldupsigma \bigr) ^2-V(x)
\label{26-4-3}
\end{gather}
We are interested in the \emph{ground state energy\/} $\E^*_N (A)$ of our system i.e. in the lowest eigenvalue of the operator $\sfH$ on $\fH$:
\begin{gather}
\E ^*_{N} (0)=\inf \Spec (\sfH) \qquad \text{on\ \ } \fH
\label{26-4-4}\\
\intertext{as $A=0$ and more generally in}
\E ^*_N=\inf_A \Bigl( \inf \Spec_{\fH}(\sfH) +
\frac{1}{\alpha} \int |\nabla \times A|^2\,dx\Bigr)
\label{26-4-5}\\
\shortintertext{where}
V(x)=\sum_{1\le m\le M} \frac{Z_m}{|x-\y_m|}
\label{26-4-6}\\
N \asymp Z\gg 1, \qquad Z\Def  Z_1+\ldots + Z_M, \qquad  Z_1 > 0,\ldots, Z_M>0\label{26-4-7}\\
\intertext{$M$ is fixed, under assumption}
0<\alpha \le \kappa^*Z^{-1}\label{26-4-8}
\end{gather}
with sufficiently small constant $\kappa^*>0$.

Our purpose is to derive an asymptotics
\begin{gather}
\E^*_N \approx  \cE^\TF_N + \sum_{1\le m\le M} 2Z_m^2 S(\alpha Z_m)
\label{26-4-9}
\intertext{and estimate an error (usually) provided}
b\Def \min _{1\le m <m'\le M} |\y_m-\y_{m'}| \ge Z^{-\frac{1}{3}}.
\label{26-4-10}
\end{gather}

Recall that Thomas-Fermi potential $W^\TF$ and Thomas-Fermi density $\rho^\TF$ satisfy equations
\begin{gather}
\rho^\TF = \frac{1}{3\pi^2} (W^\TF+\nu)_+^{\frac{3}{2}}\label{26-4-11}\\
\shortintertext{and}
W^\TF =V^0 + |x|^{-1}* \rho^\TF\label{26-4-12}
\end{gather}
where $\nu$ is a chemical potential.

\section{Lower estimate}
\label{sect-26-4-2}

Consider corresponding to $\sfH$ quadratic form exactly as in Sections~\ref{book_new-sect-24-2} and~\ref{book_new-sect-25-6}
\begin{multline}
\blangle \sfH \Psi, \Psi \brangle = \sum_j (H^0_{x_j}\Psi,\Psi) +
(\sum_{1\le j<k\le N} |x_j-x_k| ^{-1}\Psi,\Psi)=\\
\sum_j (H_{x_j}\Psi,\Psi) + ((V-W)\Psi,\Psi) +
(\sum_{1\le j<k\le N} |x_j-x_k| ^{-1}\Psi,\Psi)
\label{26-4-13}
\end{multline}
with
\begin{equation}
H =\bigl((i\nabla -\mathbf {A})\cdot \boldupsigma \bigr) ^2-W(x)
\label{26-4-14}
\end{equation}
where we select $W$ later. By Lieb-Oxford inequality the last term is estimated from below:
\begin{gather}
\blangle \sum_{1\le j<k\le N} |x_j-x_k| ^{-1}\Psi,\Psi\brangle \ge \frac{1}{2}\D(\rho_\Psi,\rho_\Psi) - C\int \rho_\Psi^{\frac{4}{3}}\,dx
\label{26-4-15}\\
\shortintertext{where}
\rho_\Psi (x)= N
\int |\Psi(x,x_2,\ldots, x_N)|^2
\,dx_2\cdots dx_N
\label{26-4-16}
\intertext{is a spatial density associated with $\Psi$ and}
\D(\rho,\rho')\Def \iint |x-y|^{-1}\rho(x)\rho'(y)\,dxdy.
\label{26-4-17}
\end{gather}
Therefore again repeating arguments of Section~\ref{book_new-sect-24-2}
\begin{multline}
\blangle \sfH \Psi, \Psi \brangle \ge \\
\sum_j (H_{x_j}\Psi,\Psi) - 2((V-W)\Psi,\Psi) +
\frac{1}{2} \D(\rho_\Psi,\rho_\Psi) -
C\int \rho_\Psi^{\frac{4}{3}}\,dx=\\
\shoveright{\sum_j (H_{x_j}\Psi,\Psi) - \D(\rho,\rho_\Psi) +
\frac{1}{2}\D(\rho_\Psi,\rho_\Psi) -
C\int \rho_\Psi^{\frac{4}{3}}\,dx=\quad\ }\\
\sum_j (H_{x_j}\Psi,\Psi) - \frac{1}{2} \D(\rho,\rho) +
\frac{1}{2}\D(\rho-\rho_\Psi,\rho-\rho_\Psi) -
C\int \rho_\Psi^{\frac{4}{3}}\,dx
\label{26-4-18}
\end{multline}
as
\begin{equation}
W-V= |x|^{-1}*\rho.
\label{26-4-19}
\end{equation}

Note that due to antisymmetricity of $\Psi$
\begin{equation}
\sum_j (H_{x_j}\Psi,\Psi)\ge  \sum_{1\le j\le N:\, \lambda_j<0} \lambda_j\ge \Tr^- (H)
\label{26-4-20}
\end{equation}
where $\lambda_j$ are eigenvalues of $H$.

To estimate the last term in (\ref{26-4-18}) we reproduce the proof of Lemma 4.3 from L. Erd\"os,  S. Fournais and   J. P. Solovej~\cite{EFS3}:

According to magnetic Lieb-Thirring inequality for $U\ge 0$:
\begin{equation}
\sum_{j\le N} \blangle (H^0_{x_j} -U)\Psi,\Psi \brangle \ge
-C\int U^{5/2}\,dx -C\gamma^{-3} U^4\,dx -\gamma \int \mathsf{B}^2dx
\label{26-4-21}
\end{equation}
where $\mathsf{B}=\nabla \times A$, $\gamma>0$ is arbitrary. Selecting
$U=\beta \min (\rho_\Psi^{5/3},\gamma \rho_\Psi^{4/3})$ with $\beta>0$ small but independent from $\gamma$ we ensure
$\frac{1}{2}U\rho_\Psi \ge CU^{5/2}+C\gamma^{-3}U^4$ and then
\begin{equation}
\sum_{j\le N}  \blangle (H^0_{x_j})\Psi,\Psi \brangle \ge
\epsilon \int \min (\rho_\Psi^{5/3} , \gamma \rho^{4/3})dx
-\gamma \int \mathsf{B}^2\,dx
\label{26-4-22}
\end{equation}
which implies
\begin{multline}
\int \rho_\Psi^{4/3}dx \le
\gamma^{-1} \int \min (\rho_\Psi^{5/3} , \gamma \rho^{4/3})dx +
\gamma \int \rho_\Psi dx \le \\
c\gamma^{-1} \sum_{j:\lambda_j<0} \blangle (H^0_{x_j})\Psi,\Psi \brangle +
c \int \mathsf{B}^2 dx +
c \gamma N
\label{26-4-23}
\end{multline}
where we use
$\int \rho_\Psi dx =N$.

\begin{remark}\label{rem-26-4-1}
As one can prove easily (see also L. Erd\"os,  S. Fournais and   J. P. Solovej~\cite{EFS3}) that
\begin{equation}
\sum_{j\le N} \blangle (H^0_{x_j})\Psi,\Psi \brangle \le CZ^{\frac{4}{3}}N
\label{26-4-24}
\end{equation}
even if $N \not\asymp Z$; then we conclude that
\begin{equation}
\int \rho_\Psi^{4/3}dx \le CZ^{\frac{2}{3}}N + C_1\int \mathsf{B}^2 dx.
\label{26-4-25}
\end{equation}
It is sufficient unless we want to recover Dirac-Schwinger terms which unfortunately is possible only as
$\alpha \ll Z^{-\frac{10}{9}}|\log Z|^{-\frac{1}{3}}$. To recover remainder estimate $o(Z^{\frac{5}{3}})$ (or marginally better) we just apply Theorem~\ref{book_new-thm-25-A-2}. We will do it later (see Theorem~\ref{26-4-5}).
\end{remark}

Therefore skipping the non-negative third term in the right-hand expression of (\ref{26-4-18}) we conclude that
\begin{multline}
\blangle \sfH \Psi, \Psi \brangle +
 \frac{1}{\alpha} \int |\nabla \times A|^2\,dx \ge \\
\Tr^-(H) + (\frac{1}{\alpha} -C_1) \int |\nabla \times A|^2\,dx
- \frac{1}{2} \D(\rho,\rho) - CZ^{\frac{5}{3}}.
\label{26-4-26}
\end{multline}
Applying Theorem~\ref{thm-26-3-24}  we conclude that
\begin{claim}\label{26-4-27}
The sum of the first and the second terms in the right-hand expression of (\ref{26-4-26}) is greater than
\begin{equation*}
\cE^\TF +\sum_m 2Z_m^2S(\alpha Z_m)-
C Z^{\frac{4}{3}} (R_1+R_2)
\end{equation*}
with $R_1$ and $R_2$ defined by (\ref{26-3-98}) and (\ref{26-3-99}) respectively with $\kappa =\alpha Z$, $h=Z^{-\frac{1}{3}}$ and
\begin{equation}
a\Def Z^{\frac{1}{3}}\min_{1\le m<m'\le M}|\y_m-\y_{m'}|.
\label{26-4-28}
\end{equation}
\end{claim}
\vspace{-5pt}

To prove this claim one needs just to rescale
\begin{phantomequation}\label{26-4-29}\end{phantomequation}
\begin{multline}
x\mapsto x Z^{\frac{1}{3}},\quad
a\mapsto a Z^{\frac{1}{3}},\quad
W\mapsto Z^{-\frac{4}{3}}W, \\
A\mapsto A,\quad
\nabla \times A \mapsto Z^{\frac{1}{3}} \nabla \times A
\tag*{$\textup{(\ref*{26-4-29})}_{1-5}$}
\end{multline}
and introduce
\begin{equation}
h=Z^{-\frac{1}{3}}, \quad \kappa=\alpha Z.
\label{26-4-30}
\end{equation}
Observe that due to $\textup{(\ref*{26-4-29})}_{3}$  we need to multiply our estimate by $Z^{\frac{4}{3}}$.

Here one definitely needs the regularity properties like in Section~\ref{sect-26-3} but we have them as $\rho=\rho^\TF$, $W=W^\TF$. Also one can see easily that ``$-C_1$'' brings correction  not exceeding
$C_2\alpha Z^2$ as $\alpha Z\le 1$.

Meanwhile  for $\rho=\rho^\TF$, $W=W^\TF$
\begin{equation}
\frac{2}{15\pi^2} \int W^{\frac{5}{2}}\,dx -\frac{1}{2}\D(\rho,\rho) =\cE^\TF.
\label{26-4-31}
\end{equation}
Lower estimate of Theorem~\ref{thm-26-4-3} below has been proven.

\begin{remark}\label{rem-26-4-2}
$\rho=\rho^\TF$, $W=W^\TF$ delivers the maximum of the right-hand expression of (\ref{26-4-31}) among $\rho, W$ satisfying (\ref{26-4-19}).
\end{remark}

\section{Upper Estimate}
\label{sect-26-4-3}

Upper estimate is easy. Plugging as in Section~\ref{book_new-sect-24-2} $\Psi$ the \emph{Slater determinant\/} (\ref{book_new-24-2-16})  of $\psi_1,\ldots, \psi_N$ where $\psi_1,\ldots,\psi_N$ are eigenfunctions of $H_{A,W}$  we get
\begin{multline}
\blangle \sfH \Psi, \Psi \brangle =\Tr^- (H_{A,W}-\lambda_N) +\lambda_N N + \\
\int (W-V)(x)\rho_\Psi(x)\, dx +
\frac{1}{2} \D (\rho_\Psi, \rho_\Psi ) -\\
\frac{1}{2}N(N-1)
\iint |x_1-x_2|^{-1} |\Psi (x_1,x_2,x_3,\ldots,x_N|^2\,dx_1\cdots dx_N
\label{26-4-32}
\end{multline}
where we do not care about last term as we drop it (again as long as we cannot get sharp enough estimate)  and the first term in the second line is in fact
$- \D(\rho  ,\rho_\Psi )$  provided (\ref{26-4-19}) holds. Thus we get
\begin{multline}
\Tr^- (H_{A,W}-\lambda_N) +\lambda_N N -\frac{1}{2}\D(\rho ,\rho) +
\frac{1}{2}\D (\rho_\Psi -\rho  , \rho_\Psi -\rho) +\\
\frac{1}{\alpha} \int |\partial A|^2\,dx
\label{26-4-33}
\end{multline}
where we added magnetic energy. Definitely we have several problems here: $\lambda_N$ depends on $A$ and there may be less than $N$ negative eigenvalues.

However in the latter case we can obviously replace $N$ by the lesser number $N'\Def \max (n\le N, \lambda_n\le 0)$ as $\E^*_N$ is decreasing function of $N$. In this case the first term  in (\ref{26-4-33}) would be just
$\Tr^- (H_{A,W})$ and the second would be $0$. Then we apply theory of the previous Section~\ref{sect-26-3} immediately without extra complications.

Consider $A$ a minimizer (or its mollification) for operator $H_{A,W}-\mu$ with potential $W=W^\TF$ and $\mu\le 0$. Then
\begin{equation}
\N(\mu)\Def \#\{ \lambda_k<\mu\} =\int (W +\mu )_+^{\frac{3}{2}}\, dx + O(Z^{\frac{2}{3}}).
\label{26-4-34}
\end{equation}
One can prove (\ref{26-4-34}) easily using the regularity properties of $A$ established in the previous Section~\ref{sect-26-3} and the same rescaling (\ref{26-4-29})--(\ref{26-4-30}) as before. We leave this easy proof to the reader.

Then repeating arguments of Subsubsection~\ref{book_new-sect-24-4-2-1}.1 ``\nameref{book_new-sect-24-4-2-1}'' we conclude that \underline{either}
$N\ge Z- C_0Z^{\frac{2}{3}}$ and then
$|\nu|\le C_1Z^{\frac{8}{9}}$ and we can take
$\mu=0$ and $|\lambda_{N'}|\le C_1Z^{\frac{8}{9}}$ \underline{or}
$N\le Z- C_0Z^{\frac{2}{3}}$ and then we take $\mu=\nu$,
$|\lambda_N|\asymp |\nu|\asymp (Z-N)_+^{\frac{4}{3}}$ and
$|\lambda_N-\nu|\le C_1|\nu|^{\frac{1}{4}}Z^{\frac{2}{3}}$.

Then following again to arguments of Subsection~\ref{book_new-sect-24-4-2} we conclude that

\begin{claim}\label{26-4-35}
Expression (\ref{26-4-33})  without term
$\frac{1}{2}\D (\rho_\Psi-\rho,\rho_\Psi-\rho)$ does not exceed
\begin{equation*}
\cE^\TF +\sum_m 2Z_m^2S(\alpha Z_m)+
C Z^{\frac{4}{3}} (R_1+R_2)
\end{equation*}
with $R_1$ and $R_2$ defined by (\ref{26-3-98}) and (\ref{26-3-99}) respectively with $\kappa =\alpha Z$, $h=Z^{-\frac{1}{3}}$ and $a$ defined by (\ref{26-4-28}).
\end{claim}

Now we need to estimate properly $\D (\rho_\Psi-\rho,\rho_\Psi-\rho)$ which  as in i.e. Subsubsection~\ref{book_new-sect-24-4-2-2}.2 ``\nameref{book_new-sect-24-4-2-2}'' does not exceed the sum of
\begin{gather}
\D (e(x,x,\mu)-\Weyl(x,\mu),\, e(x,x,\mu)-\Weyl(x,\mu)),\label{26-4-36}\\[3pt]
\D(e(x,x,\lambda_N)-\Weyl(x,\lambda_N),\,e(x,x,\lambda_N)-\Weyl(x,\lambda_N)),
\label{26-4-37}\\
\shortintertext{and}
\D(\Weyl(x,\mu)-\Weyl(x,\lambda_N),\,\Weyl(x,\mu)-\Weyl(x,\lambda_N)).
\label{26-4-38}
\end{gather}

Following arguments of Subsubsection~\ref{book_new-sect-24-4-2-2}.2 ``\nameref{book_new-sect-24-4-2-2}'' one can prove easily that due to regularity properties of $A$ both two semiclassical terms do not exceed $CZ^{\frac{5}{3}}$ and due to estimates to $|\lambda_N-\mu|$ the last term does not exceed $CZ^{\frac{5}{3}}$ either.

This concludes the proof of the upper estimate in Theorem~\ref{thm-26-4-3} below.

\section{Main theorems}
\label{sect-26-4-4}

\begin{theorem}\label{thm-26-4-3}
\begin{enumerate}[label=(\roman*), fullwidth]
\item\label{thm-26-4-3-i}
Under assumptions \textup{(\ref{26-4-7})} and \textup{(\ref{26-4-8})}
\begin{equation}
\E^*_N = \cE^\TF_N + \sum_{1\le m\le M} 2Z_m^2 S(\alpha Z_m) + O\bigl(Z^{\frac{4}{3}}(R_1+R_2)\bigr)
\label{26-4-39}
\end{equation}
with $R_1$ and $R_2$ defined by \textup{(\ref{26-3-98})} and \textup{(\ref{26-3-99})} respectively with $\kappa =\alpha Z$, $h=Z^{-\frac{1}{3}}$ and $a$ defined by \textup{(\ref{26-4-28})}, $a=\infty$ as $M=1$;

\item\label{thm-26-4-3-ii}
In particular under assumption \textup{(\ref{26-4-10})}
the following estimate holds
\begin{multline}
\E^*_N = \cE^\TF_N + \sum_{1\le m\le M} 2Z_m^2 S(\alpha Z_m) +\\
O\bigl(\alpha |\log (\alpha Z)|^{\frac{1}{3}}  Z^{\frac{25}{9}}+ Z^{\frac{5}{3}}
+ \alpha a^{-3}Z ^2 \bigr)
\label{26-4-40}
\end{multline}
\end{enumerate}
\end{theorem}

Recall that $\cE^\TF_N$ is a \emph{Thomas-Fermi energy\/} and $S(Z_m)Z_m^2$ are magnetic \emph{Scott correction terms\/}\index{Scott correction terms}.

\begin{theorem}\label{thm-26-4-4}
\begin{enumerate}[label=(\roman*), fullwidth]
\item\label{thm-26-4-4-i}
Let assumptions \textup{(\ref{26-4-7})} and  \textup{(\ref{26-4-8})}  be fulfilled and let $\Psi=\Psi_{\mathsf{A}}$ be a ground state for a near optimizer $\mathsf{A}$ of the original multiparticle problem. Then
\begin{equation}
\D(\rho_\Psi-\rho^\TF,\,  \rho_\Psi-\rho^\TF)\le CZ^{\frac{5}{3}};
\label{26-4-41}
\end{equation}
\item\label{thm-26-4-4-ii}
Furthermore, as $b\ge Z^{-\frac{1}{3}}$
\begin{equation}
\D(\rho_\Psi-\rho^\TF,\,  \rho_\Psi-\rho^\TF)\le CZ^{\frac{5}{3}}\bigl(Z^{-\delta} +(bZ^{\frac{1}{3}})^{-\delta}+ (\alpha Z)^{\delta}\bigr).
\label{26-4-42}
\end{equation}
\end{enumerate}
\end{theorem}

\begin{proof}
(i) Note that all the terms in estimates from below and from above are $O(Z^{\frac{5}{3}})$ except the common term
\begin{equation}
\Tr^-(H_{A,W}+\mu) +\frac{1}{\alpha} \int |\nabla \times A|^2\,dx
\label{26-4-43}
\end{equation}
where $A$ is a minimizer for this term and therefore estimate (\ref{26-4-41}) has been proven because estimate from below also contains $\D(\rho_\Psi-\rho^\TF,\,  \rho_\Psi-\rho^\TF)$.

\bigskip\noindent
(ii) To prove Statement~\ref{thm-26-4-4-ii} one needs

\medskip\noindent
(a) To improve estimate (\ref{26-4-34}) to
\begin{equation}
\N(\mu)=\int (W +\mu )_+^{\frac{3}{2}}\, dx + O\bigl(Z^{\frac{2}{3}} \bigl[Z^{-\delta}+ (bZ^{\frac{1}{3}})^{-\delta}+
(\alpha Z)^{\delta}\bigr]\bigr),
\label{26-4-44}
\end{equation}
(b) To estimate terms (\ref{26-4-36})--(\ref{26-4-38}) by
the  right-hand expression of (\ref{book_new-24-4-44}), and

\medskip\noindent
(c) To accommodate Dirac term in both upper and lower estimates.

\bigskip
Tasks (a), (b) are easy and we leave it to the reader (cf. arguments of Subsection~\ref{book_new-24-4-3}); we use that after rescaling effective magnetic field intensity becomes $O(\alpha Z)$ in zone $\{x: \ell (x)\asymp Z^{-\frac{1}{3}}\}$ due to already established estimates to $A$.

To fulfill (c) note that in the upper estimate we already have term
\begin{equation}
-\frac{1}{2}\tr \iint |x-y|^{-1}e _N(x,y)e_N ^\dag (x,y)\,dx dy.
\label{26-4-45}
\end{equation}
On the other hand, in virtue of Theorem~\ref{book_new-thm-25-A-2} we replace  in the lower estimate term
$-C\int \rho_\Psi^{\frac{4}{3}}(x)\,dx$ by (\ref{26-4-6}) with  $O(Z^{\frac{5}{3}-\delta})$ error (again we leave easy details to the reader).

One can prove by the same arguments as as in the non-magnetic case that for $\alpha Z\le \kappa^*$ it is $\Dirac + O(Z^{\frac{5}{3}-\delta})$.
\end{proof}

Finally, combining arguments sketched in the proof of Theorem~\ref{thm-26-4-4} with the improved estimate of (\ref{26-4-43}) (see Theorem~\ref{thm-26-3-24}\ref{thm-26-3-24-ii}) we arrive to

\begin{theorem}\label{thm-26-4-6}
Let assumptions \textup{(\ref{26-4-8})} and \textup{(\ref{26-4-10})} be fulfilled, and let
$\alpha  \le Z^{-\frac{10}{9}}|\log Z|^{-\frac{1}{3}}$. Then
\begin{multline}
\E^*_N = \cE^\TF_N + \sum_{1\le m\le M} 2Z_m^2 S(\alpha Z_m) +\Dirac+\Schwinger +\\
O\bigl(\alpha |\log (\alpha Z)|^{\frac{1}{3}}  Z^{\frac{25}{9}}+ Z^{\frac{5}{3}-\delta} +  \alpha b^{-3}Z ^2 \bigr)
\label{26-4-46}
\end{multline}
where $\Dirac$ and $\Schwinger$ are \emph{Dirac\/} and \emph{Schwinger correction terms\/} defined exactly as in non-magnetic case by \textup{(\ref{book_new-24-1-29})} and  \textup{(\ref{book_new-24-1-30})} respectively.
\end{theorem}

\section{Free nuclei model}
\label{sect-26-4-5}

Consider now free nuclei model (see Subsubsection~\ref{book_new-sect-24-4-4-2}.2).

\begin{theorem}\label{thm-26-4-7}
Let us consider $\y_m=\y_m^*$ minimizing the full energy
\begin{gather}
\widehat{\E}^*_N\Def \E^*_N +\sum _{1\le m <m'\le M} Z_mZ_{m'}|\y_m-\y_{m'}|^{-1}.
\label{26-4-47}\\
\shortintertext{Assume that}
Z_m \asymp N\qquad \forall m=1,\ldots,M.
\label{26-4-48}
\shortintertext{Then}
b\ge \min \bigl( Z^{-\frac{5}{21}+\delta},\,
Z^{-\frac{5}{21}}(\alpha Z)^{-\delta},\,
\alpha^{-\frac{1}{4}}Z^{-\frac{1}{2}}\bigr)
\label{26-4-49}
\end{gather}
and in the remainder estimates in \textup{(\ref{26-4-40})} and~\textup{(\ref{26-4-46})} one can skip $b$-connected terms; so we arrive to
\begin{equation}
\E^*_N = \cE^\TF_N + \sum_{1\le m\le M} 2Z_m^2 S(\alpha Z_m) +
O\bigl(\alpha |\log (\alpha Z)|^{\frac{1}{3}}  Z^{\frac{25}{9}}+ Z^{\frac{5}{3}}\bigr)
\label{26-4-50}
\end{equation}
and
\begin{multline}
\E^*_N = \cE^\TF_N + \sum_{1\le m\le M} 2Z_m^2 S(\alpha Z_m) +\Dirac+\Schwinger +\\
O\bigl(\alpha |\log (\alpha Z)|^{\frac{1}{3}}  Z^{\frac{25}{9}}+ Z^{\frac{5}{3}-\delta} \bigr)
\label{26-4-51}
\end{multline}
respectively and also the same asymptotics with $\widehat{\E}^*_N$ and $\widehat{\cE}^\TF_N$ instead of $\E^*_N$ and $\cE^\TF_N$.
\end{theorem}

\begin{proof}
Optimization with respect to $\y_1,\ldots,\y_M$ implies
\begin{equation}
\E^* +\sum _{1\le m<m'\le M}
\frac{Z_mZ_{m'}}{|\y_m-\y_{m'}|}< \sum_{1\le m\le M}\E^*_m
\label{26-4-52}
\end{equation}
where $\E^*=\E^*(\y_1,\ldots,\y_M;Z_1,\ldots,Z_m,N)$ and
$\E^*_m=\E^*(\y_m, Z_m)$ are calculated for separate atoms. In virtue of theorem
\ref{thm-26-4-3}
\begin{multline}
\cE^\TF+ \sum _{1\le m<m'\le M} \frac{Z_mZ_{m'}}{|\x_m-\x_{m'}|} -
\sum_{1\le m\le M} \cE^\TF_m \le\\
C\alpha |\log (\alpha Z)|^{\frac{1}{3}} Z^{\frac{25}{9}}+ Z^{\frac{5}{3}}+
C\alpha b^{-3} Z ^2;
\label{26-4-53}
\end{multline}
however due to strong non-binding theorem in Thomas-Fermi theory the left-hand expression is $\gtrsim b^{-7}$ as $b\ge Z^{-\frac{1}{3}}$ and therefore (\ref{26-4-53}) implies
\begin{equation*}
b\gtrsim \min\bigl( Z^{-\frac{5}{21}},\,
\alpha^{-\frac{1}{7}}|\log (\alpha Z)|^{-\frac{1}{21}}Z^{-\frac{25}{63}}
{\color{gray} ,\,\alpha^{-\frac{1}{4}}Z^{-\frac{1}{2}}}\bigr)
\end{equation*}
where the third expression is larger than the second one for sure. Unfortunately for $\alpha \ge Z^{-\frac{10}{9}-\delta'}$ it is not as good as we claimed in (\ref{26-4-49}). Still this estimate implies both (\ref{26-4-50}) and (\ref{26-4-51}).

To prove (\ref{26-4-49}) we record that
$b\gg Z^{-\frac{2}{7}}$\,\footnote{\label{foot-26-19} There is no binding with $b\le Z^{-\frac{1}{3}}$ because remainder estimate is (better than) $CZ^2$ and binding energy excess is $\asymp Z^{\frac{7}{3}}$.} we employ arguments used in the proof of Proposition~\ref{book_new-prop-25-8-12} and prove that
\begin{multline*}
|\Tr^- (H_{A,W}+\mu) -\sum_{1\le m\le M}\Tr^- (H_{A,W_m}+\mu)-\\
\int \Bigl( \Weyl(H_{A,W}+\mu;\, x) -
\sum_{1\le m\le M}  \Weyl(H_{A,W_m}+\mu;\, x)\Bigr)\,dx|\le
CZ^{\frac{5}{3}}\bigl(Z^{-\delta}+ (\alpha Z)^\delta\bigr)
\end{multline*}
where $A$ be a minimizer for ``molecular'' expression (\ref{26-4-45}) and $W_m$ are atomic potentials. The same estimate holds if we replace
$\Tr^- (H_{A,W_m}+\mu)$ by $\Tr^- (H_{A_m,W_m}+\mu)$ with
$A_m= A \phi (b^{-1}|x-\y_m|)$ with $\phi \in \sC^\infty_0 (B(0,\frac{1}{3})$, equal $1$ in $B(0,\frac{1}{4})$. We leave an easy proof to the reader.

Then using lower estimate for $\inf \Spec_\fH (\sfH)$ and upper estimates for
$\inf \Spec_{\fH_m} (\sfH_m)$ through $\Tr^- (H_{A,W}+\mu)$ and
$\Tr^- (H_{A_m,W_m}+\mu')$ respectively (where $\sfH_m$ are associated with $H_{A_m,W_m}$) we arrive to
\begin{multline*}
\inf \Spec_\fH (\sfH) \ge
\sum_{1\le m\le M} \inf \Spec_{\fH_m} (\sfH_m)+
\cE^\TF -\sum_{1\le m\le M} \cE^\TF_m \\[3pt]
-CZ^{\frac{5}{3}}\bigl( Z^{-\delta}+(\alpha Z)^{\delta}\bigr)
\end{multline*}
and therefore
\begin{multline}
\E_{A} \ge \sum_{1\le m\le M} \E_{m,A_m} +
\cE^\TF -\sum_{1\le m\le M} \cE^\TF_m -\\
CZ^{\frac{5}{3}}\bigl( Z^{-\delta}+(\alpha Z)^{\delta}\bigr)-
C\alpha b^{-3} Z^2
\label{26-4-54}
\end{multline}
where the last term is due to replacement of
$\frac{1}{\alpha}\int |\nabla \times A|^2\,dx$ by ``atomized'' expressions
$\sum _{1\le m\le M}\frac{1}{\alpha}\int |\nabla \times A_m|^2\,dx$.

The last inequality (\ref{26-4-54}) then obviously holds with $A_m$ replaced by optimizers for ``atomic'' expressions (\ref{26-4-45}) and \emph{now\/} strong non-binding theorem implies that
\begin{equation*}
b^{-7}\le CZ^{\frac{5}{3}}\bigl( Z^{-\delta}+(\alpha Z)^{\delta}\bigr)+
C\alpha b^{-3} Z^2
\end{equation*}
which implies (\ref{26-4-49}) where we change $\delta>0$ as needed.
\end{proof}

\chapter{Miscellaneous problems}
\label{sect-26-5}

In our analysis in Sections~\ref{book_new-sect-24-5} and~\ref{book_new-sect-24-6} the crucial role was played by an estimate of $\D(\rho_\Psi-\rho^\TF,\, \rho_\Psi-\rho^\TF)$ and since what we have now (see Theorem~\ref{thm-26-4-4}) is (almost) as good as we had then, all arguments of these Sections still work with the minimal modifications. We leave most of the easy details to the reader but we need to deal with different magnetic fields for different $N$.

\section{Excessive negative charge}
\label{sect-26-5-1}

\begin{theorem}\footnote{\label{foot-26-20} cf. Theorem~\ref{book_new-thm-24-5-2}.}\label{thm-26-5-1}
Let condition \textup{(\ref{26-4-48})} be fulfilled.
\begin{enumerate}[label=(\roman*), fullwidth]
\item\label{thm-26-5-1-i}
In the framework of the fixed nuclei model let us assume that \newline $\I^*_N\Def \E^*_{N-1}-\E^*_N>0$. Then
\begin{equation}
(N-Z)_+\le CZ^{\frac{5}{7}}
\left\{\begin{aligned}
&1 \qquad &&\text{as\ \ } a\le Z^{-\frac{1}{3}},\\
&Z^{-\delta}+ (aZ^{\frac{1}{3}})^{-\delta}+(\alpha Z)^\delta
&&\text{as\ \ } a\ge Z^{-\frac{1}{3}};
\end{aligned}\right.
\label{26-5-1}
\end{equation}
\item\label{thm-26-5-1-ii}
In particular for a single atom and for molecule with
$a\ge Z^{-\frac{1}{3}+\delta}$
\begin{equation}
(N-Z)_+\le Z^{\frac{5}{7}}\bigl(Z^{-\delta}+(\alpha Z)^\delta\bigr);
\label{26-5-2}
\end{equation}
\item\label{thm-26-5-1-iii}
In the framework of  the free nuclei model let us assume that
$\widehat{\I}^*_N\Def \widehat{\E}^*_{N-1}-\widehat{\E}^*_N>0$. Then estimate \textup{(\ref{26-5-2})} holds.
\end{enumerate}
\end{theorem}

\begin{proof}
The proof of follows the proof of Theorem~\ref{book_new-thm-24-5-2} and since it is not non-magnetic field specific we find that $(N-Z)_+\le Q^{\frac{3}{7}}$ where $Q$ is an estimate for $\D(\rho_\Psi-\rho^\TF,\, \rho_\Psi-\rho^\TF)$ which we established already; recall that $\rho^\TF \asymp \ell^{-6}$ as
$\ell \gtrsim Z^{-\frac{1}{3}}$ also plays important role.

Here we pick up $A=A_N$ (exactly as in the analysis of free nuclei model we pick up $\underline{y}$ for $N$-electrons) and conclude that
\begin{equation*}
\I_N(A)\Def \E_{N-1}(A)- \E_{N} (A)\ge \E^*_{N-1}- \E^*_{N}>0
\end{equation*}
and then repeat arguments of the proof of Theorem~\ref{book_new-thm-24-5-2}.

We leave easy details to the reader.
\end{proof}

\section{Estimates for ionization energy}
\label{sect-26-5-2}

\begin{theorem}\footnote{\label{foot-26-21} cf. Theorem~\ref{book_new-thm-24-5-3}.}\label{thm-26-5-2}
Let condition \textup{(\ref{26-4-48})} be fulfilled and let
$N\ge Z-C_0 Z^{\frac{5}{7}}$. Then
\begin{enumerate}[label=(\roman*), fullwidth]
\item\label{thm-26-5-2-i}
In the framework of the fixed nuclei model
\begin{equation}
\I^*_N \le  CZ^{\frac{20}{21}}.
\label{26-5-3}
\end{equation}
\item\label{thm-26-5-2-ii}
In the framework of the free nuclei model with
$N\ge Z-C_0 Z^{\frac{5}{7}}\bigl(Z^{-\delta}+{\alpha Z}^{\delta}\bigr)$
\begin{equation}
\widehat{\I}_N^*\Def \widehat{\E}^*_{N-1}- \widehat{\E}^*_{N-1} \le Z^{\frac{20}{21}}\bigl(Z^{-\delta'}+(\alpha Z)^{\delta'}\bigr).
\label{26-5-4}
\end{equation}
\end{enumerate}
\end{theorem}

\begin{proof}
Recall that Theorem~\ref{book_new-thm-24-5-3} was proven simultaneously with Theorem~\ref{book_new-thm-24-5-2}; we follow the same scheme here picking up $A=A_N$. Thus here and in the first part of the proof of Theorem~\ref{thm-26-5-3} we estimate from above $\I^*_N(A)$.

Again easy details a left to the reader.
\end{proof}

\begin{theorem}\footnote{\label{foot-26-22} cf. Theorem~\ref{book_new-thm-24-6-3}.} \label{thm-26-5-3} Let condition \textup{(\ref{26-4-48})} be fulfilled and let  $N\le Z-C_0 Z^{\frac{5}{7}}$. Then in the framework of fixed nuclei model under assumption \textup{(\ref{book_new-24-6-2})}
\begin{equation}
(\I^*_N +\nu )_+\le C (Z-N)^{\frac{17}{18}}Z^{\frac{5}{18}}
\left\{\begin{aligned}
&1 \qquad &&\text{as\ \ } a\le Z^{-\frac{1}{3}},\\
&Z^{-\delta}+ (aZ^{\frac{1}{3}})^{-\delta}
&&\text{as\ \ } a\ge Z^{-\frac{1}{3}}.
\end{aligned}\right.
\label{26-5-5}
\end{equation}
\end{theorem}

\begin{proof}
To estimate $\I^*_N+\nu$ from above we follow exactly the arguments of Subsection~\ref{book_new-sect-24-6-1} to estimate $\I_N(A_N)\ge \I^*_N$.
\end{proof}

\begin{Problem}\label{rem-26-5-4}
To prove the same estimate for $(\I^*_N+\nu)_-$.
\end{Problem}

\begin{remark}\label{rem-26-5-5}
To estimate $\I^*_N+\nu$ from below we  pick up $A=A_{N-1}$; then
\begin{equation*}
\I_{N}(A)= \E_{N-1}(A)-\E_{N}(A)\le \E^*_{N-1} -\E^*_{N} =\I^*_N
\end{equation*}
and we should follow the arguments of Subsection~\ref{book_new-sect-24-6-2}. However in contrast to all other proofs of this Section here we should use spectral properties of $H_{A,V}$ (or at least an estimate from above for its lowest eigenvalue after localization to $\supp \theta$ while in all other results we need an estimate from below for the same lowest eigenvalue after localization to $\supp \theta$) and to do so we need some uniform (i.e. with constants which do not depend on $N$) smoothness estimates for $A_{N-1}$ where $A_{N-1}$ is the minimizer for $\E_{N-1}(A)$ as defined in Sections~\ref{sect-26-4} and here (rather than as defined in Sections~\ref{sect-26-2}--\ref{sect-26-3}).
While (\ref{26-A-1}) is an analogue of (\ref{26-2-14}), it is still not the same and while it implies some estimate it is not even remotely as good as we achieved in Sections~\ref{sect-26-2} and \ref{sect-26-3}.

Sure $\rho_\Psi$ is not very smooth either but it close to rather smooth $\rho^\TF$; on the other hand minimizer $A_N$ is an almost-minimizer for the one-particle trace problem studied  Sections~\ref{sect-26-2}--\ref{sect-26-3} but we don't know how close it to the minimizer (or one of the minimizers) of the latter problem.
\end{remark}

\section{Free nuclei model: Excessive positive charge}
\label{sect-26-5-3}

\begin{theorem}\footnote{\label{foot-26-23} cf. Theorem~\ref{book_new-thm-24-6-4}.}\label{thm-26-5-6}
Let condition \textup{(\ref{26-4-48})} be fulfilled. Then in the framework of free nuclei model with $M\ge 2$ the stable molecule does not exist unless
\begin{equation}
Z-N \le Z^{\frac{5}{7}}\bigl( Z^{-\delta}+(\alpha Z)^\delta\bigr).
\label{26-5-6}
\end{equation}
\end{theorem}

\begin{proof}
Again we just repeat the proof of Theorem~\ref{book_new-thm-24-6-4}.
\end{proof}

\begin{subappendices}
\chapter{Appendices}

\section{Minimizers and ground states}
\label{sect-26-A-1}

First establish a conditional existence of the minimizer\footnote{\label{foot-26-24} We do not know if it is unique.} and the corresponding ground state of the original problem:

\begin{theorem}\label{thm-26-A-1}
Let $\alpha Z <\kappa^*$ and let $\E^*_N <\E^*_{N-1}$. Then there exist a minimizer $A=A_N$ for the original multiparticle problem and the corresponding ground state $\Psi_N$.
\end{theorem}

\begin{proof}
We know that as $\alpha Z <\kappa^*$ (with $\kappa^*>0$ which does not depend on $Z$ or positions of the nuclei) $\E^*(A)$ is bounded from below; then
$\|\nabla \times A'\|^2$ is bounded from above for near-minimizer $A'$ (but constants do depend on $Z$ and $(\kappa^*-\alpha Z)$ here). On the other hand, if $A'\in \sC_0^\infty$ and $\E_N(A')< \E_{N-1}(A')$ there exists a ground state $\Psi_N(A')$.

Therefore if $A_{(k)}\in \sC_0^\infty$ is a minimizing sequence for $\E_N(A)$ we have also a sequence $\Psi_N(A_{(k)})$ with $\|\Psi_N(A_{(k)})\|=1$.
Going if necessary to subsequence we can assume that $A_{(k)}$ converges weakly in $\sH^1$ and strongly in $\sL^p_\loc$ for any $p<6$; let $A$ be its limit.

One can prove easily that $\Psi_N (A_{(k)})$ converge weakly in $\sH^1$ and strongly in $\sL^2$ to $\Psi$ and
\begin{gather*}
(\sfH_{A,V} \Psi,\Psi)=\lim_{k\to \infty} (\sfH_{A_{(k)},V} \Psi_N(A_{(k)}),\Psi_N(A_{(k)}))\\
\shortintertext{and then}
(\sfH_{A,V} \Psi,\Psi)+\frac{1}{\alpha} = \lim_{k\to \infty} \E _N (A_{(k)})
\end{gather*}
which is $\E^*_N$ since $A_{(k)}$ is a minimizing sequence and then $\Psi$ must be a ground state.
\end{proof}

Now in this framework we establish properties of these minimizer and the ground state:

\begin{proposition}\footnote{\label{foot-26-25} cf. Proposition~\ref{book_new-prop-25-A-7}.}\label{prop-26-A-2}
Let $\Psi=\Psi_N$ and $A=A_N$ be a ground state and minimizer with energy $\E^*_N<\E^*_{N-1}$.

\begin{enumerate}[label=(\roman*), fullwidth]
\item\label{prop-26-A-2-i}
$\Psi\in \sC^1$ and $\Psi= O(e^{-\epsilon |\underline{x}|})$ as $|\underline{x}|\to \infty$;

\item\label{prop-26-A-2-ii}
$A\in \sC^1$ and $A= O(|x|^{-2})$, $\nabla \times A= O(|x|^{-3})$ as $|x|\to \infty$;

\item\label{prop-26-A-2-iii}
Let $N<Z$. Then $V_\Psi -V \in \sC^1$ and $V_\Psi = (Z-N)|x|^{-1} +O(|x|^{-2})$,
$\nabla{V}_\Psi = (Z-N)|x|^{-2} +O(|x|^{-3})$ as $|x|\to \infty$.
\end{enumerate}
\end{proposition}
\begin{proof}
Obvious proof using also an equation
\begin{multline}
\frac{2}{\alpha}\Delta A_j =\\
-2N \Re \tr \int \Psi^\dag (x,x_2,\ldots,x_N) \upsigma_j (D-A)_x\cdot \boldupsigma \Psi (x,x_2,\ldots,x_N)\,dx_2\cdots dx_N,
\label{26-A-1}
\end{multline}
is left to the reader. This equation is similar to (\ref{26-2-14}) and is also derived from variational principles, its right-hand expression is $\frac{\updelta \Lambda}{\updelta A}$ where $\Lambda$ is the lowest eigenvalue of $\sfH_{A,V}$ on Fock' space.\end{proof}

\section{Zhislin' theorem}
\label{sect-26-A-2}

\begin{theorem}[Zhislin's theorem]\footnote{\label{foot-26-26} cf. Theorem~\ref{book_new-thm-25-A-8}.}\label{thm-26-A-3}
$\E^*_{N+1}< \E^*_N$ as $N<Z$.
\end{theorem}

\begin{proof}
An easy proof repeating with obvious modifications proof of Theorem~\ref{book_new-thm-25-A-8} is left to the reader.
\end{proof}

\section{L. Erd\"os--J. P. Solovej' lemma}
\label{sect-26-A-3}

We reproduce here Lemma 2.1 from L.~Erd\"os,  J. P. Solovej~\cite{erdos:solovej}.

\begin{lemma}\label{lemma-26-A-4}
There is a positive universal constant $\kappa^*$ such
that for any $Z, \alpha$ with $Z\alpha \le \kappa^*$ we have
\begin{equation*}
\inf_N \inf_{A}
 H_{A,V} \ge -
 CZ^{\frac{7}{3}}\delta^{1/2} - Z^{\frac{2}{3}} \delta^{-2}
\end{equation*}
if $C Z^{-\frac{2}{3}}\le \delta\le C_1$ with a sufficiently large constant $C$.
\end{lemma}

\begin{proof}
Consider pair of smooth functions $\theta_0$ and $\theta_1$,
\begin{equation*}
\theta_0^2 + \theta_1^2= 1, \; \supp \theta_1 \subset B(0, 2r),
\; \theta_1= 1 \text{\ \ on\ \ } B(0,r),
\quad |\nabla \theta_0|, \ |\nabla \theta_1|\le Cr^{-1}
\end{equation*}
with $r=\delta Z^{-\frac{1}{3}}$.

Let $\tilde{\chi}_0$ be a smooth cutoff function supported on $B(0,3r)$ such that $|\nabla \tilde{\chi}_0|\le  Cr^{-1}$ and $\tilde{\chi}_0=1$ on $B(0,2r)$. Let $ \bar{A}$ be an average of $A$ over $B(0,3r)$. Let
$A_0\Def  (A-\bar{A}) \tilde{\chi}_0$, $B_0\Def \nabla\times A_0$;
then $\nabla \otimes A_0 =
\tilde{\chi}_0 \nabla\otimes A + (A-\bar{A})\otimes\nabla\tilde{\chi}_0$. Clearly
\begin{multline*}
\int_{\bR^3} B_0^2\,dx \le \int_{\bR^3} |\nabla\otimes A_0|^2\,dx\le
 2\int_{\bR^3} \tilde{\chi}_0^2 |\nabla\otimes A|^2\,dx +
 Cr^{-2}\int_{B(0,3r)}(A-\bar{A})^2\,dx  \\
\le  C_1\int_{B(0,3r)} |\nabla\otimes A|^2\,dx
\end{multline*}
for some universal constant $C_1$, where in the last step we used the Poincar\'e inequality. Let $\varphi$ be a real phase  such that  $\nabla\varphi=\bar{A}$.
Since $\tilde{\chi}_0\equiv 1$ on the support of $\theta_1$, we have
\begin{equation*}
\theta_1 H_{A,0}  \theta_1= \theta_1  e^{-i\varphi}
H_{A-\bar{A},0} e^{i\varphi}\theta_1 =
\theta_1  e^{-i\varphi} H_{A_0,V} e^{i\varphi}\theta_1 .
\end{equation*}
After these localizations, we have
\begin{multline}
 H_{N,Z;A}^1\Def \\
 \sum_{j=1}^N \Bigl[ \theta_1\Bigl(H_{A,0} - \frac{Z}{|x| }- \bigl(|\nabla\theta_0|^2 +|\nabla\theta_1|^2\big)\theta_1\Bigr]_j
  +  \frac{1}{\alpha }\int_{B(0,3r)}|\nabla\otimes A|^2 \,dx \\
\ge \sum_{j=1}^N \Bigl[ \theta_1 e^{-i\varphi}\bigl(H_{A_0,0} - W(x)\bigr) e^{i\varphi}\theta_1\Bigr]_j +\frac{1}{2C_1\alpha}\int B_0^2\,dx
\label{26-A-2}
\end{multline}
with
\begin{equation*}
W(x) =\Bigl[\frac{Z}{|x|} + Cr^{-2}\Bigr]\, \mathbf{1}(|x|\le 2r)
\end{equation*}
where $\mathbf{1}(X)$ is a characteristic function of $X$.

Now we use the ``running energy scale'' argument in E.~Lieb, M.~Loss,~M. and J.~Solovej~\cite{lieb:loss:solovej}.
\begin{multline}
\sum_{j=1}^N \Bigl[\theta_1 e^{-i\varphi}
\bigl[H_{A',0}- W\bigr]e^{i\varphi}\theta_1\Bigr]_j
  \ge -\int_0^\infty\N_{-e}( H_{A',0}- W)\,d e\\
\shoveright{\ge -\int_0^\mu\N_{-e}(H_{A',0}- W)\,d e -
\int_\mu^\infty\N_{0}\bigl(\frac{\mu}{e}H_{A',0}- W+e\bigr)\,d e\;}\\
 \ge -\int_0^\mu\N_{-e}(H_{A',0}- W)\,d e -
\int_\mu^\infty\N_{0}\bigl(H_{A',0}-\frac{e}{\mu} W+ \frac{e^2}{\mu}\Bigr)\,d e,
\label{26-A-3}
\end{multline}
where $\N_{-e}(H)$ denotes the number of eigenvalues of a self-adjoint operator $H$ below $-e$.

In the first term we use the bound $H_{A_0,0}\ge (D-A_0)^2 - |B_0|$ and the CLR (i.e. Cwikel-Lieb-Rozenblum) bound:
\begin{multline}
  \int_0^\mu\N_{-e}(H_{A_0,0}- W)\,d e \le
  C\int_0^\mu \,d e\int_{\bR^3}  (W+ |B_0|-e)_+^{\frac{3}{2}}\,dx\\
 \le C\int_0^\mu \,d e\int_{\bR^3}(W-e/2)_+^{\frac{3}{2}}\,dx  +
 C\int_0^\mu \,d e\int_{\bR^3} (|B_0|-e^2/2\mu)_+^{\frac{3}{2}}\,dx\\
\le C\int_{\bR^3} W^{\frac{5}{2}} \,dx +
C\mu^{\frac{1}{2}}\int_{\bR^3} B_0^2\,dx
= CZ^{\frac{5}{2}}r^{\frac{1}{2}}+Cr^{-2} +
C\mu^{\frac{1}{2}}\int_{\bR^3} B_0^2\,dx .
\label{26-A-4}
\end{multline}
In the second term of (\ref{26-A-3}) we use
\begin{multline*}
  H_{A_0,0}-\frac{e}{\mu} W \ge
   \frac{1}{2}\bigl[ (D-A_0)^2
 - \frac{2eZ}{\mu|x|}\,\mathbf{1} (|x|\le 2 r)\bigr] \\[3pt]
 + \frac{1}{2} (D-A_0)^2
 - |B_0|  -\frac{Ce}{\mu r^2}\,\mathbf{1}(|x|\le 2r),
\end{multline*}
and that
\begin{gather*}
(D-A_0)^2  - \frac{2eZ}{\mu|x|}\,\mathbf{1} (|x|\le 2r) \ge
(D-A_0)^2  - \frac{4eZ}{\mu|x|}\ge - \bigl(\frac{2eZ}{\mu}\bigr)^2\\
\shortintertext{i.e.}
H_{A_0,0}-\frac{e}{\mu} W \ge
\frac{1}{2} (D-A_0)^2 - 2\bigl(\frac{eZ}{\mu}\bigr)^2
 - |B_0| -\frac{Ce}{\mu r^2}\, \mathbf{1}(|x|\le 2r).
\end{gather*}
Let $\mu=4Z^2$, then using $Ce/\mu r^2 \le e^2/4\mu$ for $\mu\le e$
(i.e. $C\le (\delta Z^{2/3})^2$), we get
\begin{multline}
 \int_\mu^\infty\N_{0}\bigl(H_{A_0,0}- \frac{e}{\mu}W+\frac{e^2}{\mu}\bigr)\,d e
\\
\le \int_\mu^\infty\N_{0}\bigl(\frac{1}{2} (D-A_0)^2
 - |B_0| +\frac{e^2}{4\mu}\bigr)\,d e
\le C\int_0^\mu \,d  e\int_{\bR^3}
    (|B_0|-e^2/4\mu)_+^{3/2}\,dx \\
\leq C\mu^{\frac{1}{2}}\int_{\bR^3} B_0^2\,dx .
\label{26-A-5}
\end{multline}
Note that if $Z\alpha \le \kappa^*$ with some sufficiently small universal constant $\kappa^*$, then (\ref{26-A-5}) can be controlled by the corresponding
term in \eqref{26-A-2}. Combining the  estimates (\ref{26-A-2}), (\ref{26-A-3}),
(\ref{26-A-4}) and (\ref{26-A-5}) we obtain
\begin{equation*}
  H_{A,V} \ge - CZ^{\frac{5}{2}}r^{\frac{1}{2}} -Cr^{-2}
\end{equation*}
and  Lemma  follows.
\end{proof}

\end{subappendices}

\input Preprint-26.bbl
\end{document}

%% file: Preprint-26.bbl

%% file: book26.bbl
\begin{thebibliography}{99}
\bibliographystyle{book}

\providecommand{\bysame}{\leavevmode\hbox to5em{\hrulefill}\thinspace}


\bibitem[A]{arnold:classical}
V. I. Arnold:
\emph{Mathematical Methods of Classical Mechanics}.
Springer-Verlag (1990).

\bibitem[Ba]{bach}
V. Bach: \emph{Error bound for the Hartree-Fock energy of
atoms and molecules,}
Commun. Math. Phys.  147:527--548 (1992).


\bibitem[Be]{benguria}
R. Benguria:
\emph{Dependence of the Thomas-Fermi  energy on the nuclear coordinates,}
Commun. Math. Phys., 81:419--428 (1981).


\bibitem[BeL]{benguria:lieb}
R. Benguria and E. H. Lieb:
\emph{The positivity of the pressure in Thomas-Fermi theory,}
Commun. Math. Phys., 63:193--218 (1978).


\bibitem[BrL]{brezis:lieb}
H. Brezis and E. H. Lieb:
\emph{Long range potentials in Thomas-Fermi theory,}  Commun. Math. Phys.
65, 231-246 (1979).


\bibitem[E]{erdos:magnetic} L. Erd{\H o}s,
Magnetic Lieb-Thirring inequalities.
\newblock {Commun. Math. Phys.}, 170:629--668 (1995).

\bibitem[ES3]{erdos:solovej}  L. Erd{\H o}s, J.P. Solovej, 
\emph{Ground state energy of large atoms in a self-generated
magnetic field.} Commun. Math. Phys. \textbf{294}, No. 1, 229-249 (2009)
\href{http://arxiv.org/abs/0903.1816}{arXiv:0903.1816}

\bibitem[EFS1]{EFS1}  L. Erd{\H o}s, S. Fournais, J.P. Solovej: 
\emph{Stability and semiclassics in self-generated
fields.} 
\href{http://arxiv.org/abs/1105.0506}{arXiv:1105.0506}


\bibitem[EFS2]{EFS2}  L. Erd{\H o}s, S. Fournais,
J.P. Solovej: \emph{Second order semiclassics with self-generated
magnetic fields.} 
\href{http://arxiv.org/abs/1105.0512}{arXiv:http://arxiv.org/abs/1105.0512}


\bibitem[EFS3]{EFS3}  L. Erd{\H o}s, S. Fournais, J.P. Solovej: 
\emph{Scott correction for large atoms and molecules in a self-generated magnetic field} \href{http://arxiv.org/abs/1105.0521}{arXiv:1105.0521}


\bibitem[FS]{FS} C.~Fefferman and L.A.~Seco: \emph{On the energy of a
    large atom}, Bull.~AMS \textbf{23}, 2, 525--530 (1990).

\bibitem[FSW1]{FSW1} R. L. Frank, H. Siedentop, S. Warzel:
\emph{The ground state energy of heavy atoms: 
relativistic lowering of the leading energy correction.} 
Commun. Math. Phys.  \textbf{278}
   no. 2, 549--566 (2008)




\bibitem[FSW2]{FSW2} R. L. Frank, H. Siedentop, S. Warzel:
\emph{The energy of heavy atoms according to Brown and Ravenhall: the Scott correction.}
Doc. Math. \textbf{14}, 463--516 (2009).







\bibitem[FLL]{FLL}
J. Fr\"ohlich, E. H. Lieb, and M. Loss:
\emph{Stability of Coulomb systems with magnetic fields. 
I. The one-electron atom. }
Commun.\ Math.\ Phys.\ \textbf{104} 251--270 (1986) 



\bibitem[GS]{graf:solovej}
G. M. Graf  and J. P. Solovej:
\emph{A correlation estimate with applications to quantum systems with Coulomb interactions,}
Rev.~Math.~Phys., 6(5a):977--997 (1994).
Reprinted in  The state of matter a volume dedicated to
E.~H.~Lieb, Advanced series in mathematical physics, 20,
M.~Aizenman and H.~Araki (Eds.), 142--166, World Scientific 1994.



\bibitem[H]{H} W.~Hughes: \emph{An atomic energy bound that gives
    Scott's correction}, Adv.~Math. \textbf{79}, 213--270 (1990).





\bibitem[Ivr1]{ivrii:MQT1} V. Ivrii,
\href{http://weyl.math.toronto.edu/victor2/preprints/MQT1.pdf}
{Asymptotics of the ground state energy of heavy molecules in
the strong magnetic field. I}.  Russian Journal
of Mathematical Physics, 4(1):29--74 (1996).

\bibitem[Ivr2]{ivrii:MQT2}
\href{http://weyl.math.toronto.edu/victor2/preprints/MQT2.pdf}
{Asymptotics of the ground state energy of heavy molecules in
the strong magnetic field. II}.  Russian Journal  of Mathematical Physics,
5(3):321--354 (1997).

\bibitem[Ivr3]{ivrii:MQT-Vishik}  V. Ivrii,
\href{http://weyl.math.toronto.edu/victor2/preprints/MQT-Vishik.pdf}{Heavy molecules in the strong magnetic field.}
 Russian Journal of Math. Phys., 4(1):29--74 (1996).



\bibitem[Ivr4]{ivrii:MQT3}  V. Ivrii,
 \href{http://weyl.math.toronto.edu/victor2/preprints/MQT3.pdf}{
Heavy molecules in the strong magnetic field. Estimates for ionization energy and excessive charge} 6(1):56--85 (1999).


\bibitem[Ivr5]{ivrii:IRO2}
\href{http://weyl.math.toronto.edu/victor2/preprints/IRO2.pdf}{Sharp spectral asymptotics for operators with irregular coefficients.
Pushing the limits. II}. Comm. Part. Diff. Equats.,  28 (1\&2):125--156, (2003).


\bibitem[Ivr6]{ivrii:IRO3} V. Ivrii
\href{http://arxiv.org/abs/math/0510326.pdf}{Sharp spectral asymptotics for operators with irregular coefficients.
III. Schr\"odinger operator with a strong magnetic field}, arXiv:math/0510326
(Apr. 11, 2011), 101pp.


\bibitem[Ivr7]{ivrii:Preprint-A} V. Ivrii
\href{http://arxiv.org/pdf/1108.4188.pdf}{Local trace asymptotics in the self-generated magnetic field}, arXiv:math/1108.4188 (December 23, 2011), 24pp.

\bibitem[Ivr8]{ivrii:Preprint-B} V. Ivrii
\href{http://arxiv.org/pdf/1112.2487.pdf}{Global trace asymptotics in the self-generated magnetic field in the case of Coulomb-like singularities}, arXiv:math/1112.2487 (December 23, 2011), 19pp.

\bibitem[Ivr9]{ivrii:Preprint-C} V. Ivrii
\href{http://arxiv.org/pdf/1112.5538.pdf}{Asymptotics of the ground state energy for atoms and molecules in the self-generated magnetic field}, arXiv:math/1112.5538 (December 23, 2011), 11pp.



\bibitem[Ivr10]{ivrii:Preprint-24} V. Ivrii
\href{http://arxiv.org/pdf/1210.1132.pdf}{Asymptotics of the ground state energy of heavy molecules and related topics}, arXiv:math/0510326
(October 03, 2012), 70pp.

\bibitem[Ivr11]{ivrii:Preprint-25} V. Ivrii
\href{http://arxiv.org/pdf/1210.1329}{Asymptotics of the ground state energy of heavy molecules and related topics. II}, arXiv:math/1210.1329
(January 23, 2013), 141pp.




\bibitem[Ivr11]{futurebook} V. Ivrii
 \emph{Microlocal Analysis and Sharp Spectral Asymptotics}, 
 in progress: available online at \newline
\href{http://www.math.toronto.edu/ivrii/futurebook.pdf}{http://www.math.toronto.edu/ivrii/futurebook.pdf}


\nopagebreak
\bibitem[IS]{ivrii:ground} Ivrii, V. and Sigal, I.~M
\href{http://weyl.math.toronto.edu/victor2/preprints/Scott.pdf}{Asymptotics of the ground state energies of large {C}oulomb systems.}
\newblock {Ann. of Math.}, 138:243--335 (1993).


\bibitem[L]{L} E. H. Lieb: \emph{Thomas-Fermi and related
theories of atoms and molecules}, Rev. Mod. Phys. \textbf{65}. No. 4, 603-641
(1981)



\bibitem[L2]{L2} E. H. Lieb: \emph{Variational principle for
many-fermion systems}, Phys. Rev. Lett. \textbf{46}, 457--459 (1981)
and \textbf{47} 69(E) (1981)

\bibitem[L2]{lieb:selecta} \emph{The stability of matter: from atoms to
stars} (Selecta). Springer-Verlag (1991).

\bibitem[LLS]{lieb:loss:solovej} E. H. Lieb, M. Loss and J. P. Solovej: {\em
Stability of Matter in Magnetic Fields}, Phys. Rev. Lett. \textbf{75},
 985--989 (1995)


\bibitem[LO]{LO} E. H. Lieb and S. Oxford: \emph{Improved Lower Bound
    on the Indirect Coulomb Energy}, Int. J. Quant. Chem. \textbf{19},
  427--439, (1981)




\bibitem[LS]{LS} E. H. Lieb and B. Simon: \emph{The Thomas-Fermi theory
of atoms, molecules and solids}, Adv. Math. \textbf{23}, 22-116 (1977)


\bibitem[LSY1]{LSY1}
E.~H. Lieb, ~J.~P. Solovej  and J. Yngvarsson:
\emph{Asymptotics of heavy atoms in high magnetic fields: I. Lowest Landau band
regions,} Comm. Pure Appl. Math. 47:513--591 (1994).

\bibitem[LSY2]{LSY2}
E.~H. Lieb, ~J.~P. Solovej  and J. Yngvarsson:.
\emph{Asymptotics of heavy atoms in high magnetic fields: II. Semiclassical regions,}
Comm. Math. Phys., 161: 77--124 (1994).


\bibitem[RS]{ruskai:solovej}
M. B. Ruskai, M. B. and J. P. Solovej:
\emph{Asymptotic neutrality of
polyatomic molecules}. In  Schr\"odinger Operators,  Springer Lecture
Notes in Physics 403, E. Balslev (Ed.), 153--174, Springer Verlag
(1992).

\bibitem[SW1]{SW1} H.~Siedentop and R.~Weikard: \emph{On the leading
    energy correction for the statistical model of an atom:
    interacting case}, Commun.~Math.~Phys.~ \textbf{112}, 471--490
  (1987) 

\bibitem[SW2]{SW2} H.~Siedentop and R.~Weikard: \emph{On the leading
    correction of the Thomas-Fermi model: lower bound}, Invent.~Math.
  \textbf{97}, 159--193 (1990)

\bibitem[SW3]{SW3} H.~Siedentop and R.~Weikard:
 \emph{A new phase space localization technique with
    application to the sum of negative eigenvalues of {S}chr\"odinger
    operators}, Ann.~Sci.~\'Ecole Norm. Sup. (4), \textbf{24}, no.~2,
  215--225 (1991).
  
 \bibitem[Sob1]{Sob1} A.~V.~Sobolev: \emph{Quasi-classical asymptotics
of local Riesz means for the Schr\"odinger operator in a moderate
magnetic field.} Ann. Inst. H. Poincar\'e, \textbf{62}  no. 4, 325-360,  (1995)

 \bibitem[Sob]{Sob} A.~V.~Sobolev: \emph{Discrete spectrum 
asymptotics for the Schr\"{o}dinger operator with a
 singular potential and a magnetic field},
 Rev.~Math.~Phys \textbf{8} (1996) no.~6, 861--903.
  

\bibitem[SS]{SS}  J. P. Solovej, W. Spitzer:  \emph{A new
coherent states approach to semiclassics which gives
Scott's correction.} Comm. Math. Phys.  \textbf{241}  (2003),  no. 2-3, 383--420.

\bibitem[SSS]{SSS}  J. P. Solovej, T.\O.  S\o rensen, W. Spitzer:  \emph{
Relativistic Scott correction for atoms and molecules.}
Comm. Pure Appl. Math. Vol. LXIII. 39-118 (2010).



\bibitem[Zh]{zhislin:spectrum}
G. Zhislin,
\emph{Discussion of the spectrum of Schrodinger operator for systems of many particles}.
\newblock {Tr. Mosk. Mat. Obs.}, 9, 81--128 (1960).


 \end{thebibliography}
